\documentclass{article}
\usepackage{amsmath}
\usepackage{amsfonts,amsmath,amsthm,amssymb}
\usepackage{mathtools}
\usepackage[numbers,sort]{natbib}
\usepackage{geometry}
\usepackage{pgf,pgfplots}
\usepackage{tikz}
\usepackage{caption}
\usepackage{subcaption}
\usepackage[]{enumitem}
\usepackage{bigints}
\usepackage{float}
\usepackage[title]{appendix}

\usepackage{enumitem}

\usepackage{bm}

\usepackage{xcolor}
\RequirePackage[colorlinks,citecolor=blue,urlcolor=blue]{hyperref}

\theoremstyle{plain}
\newtheorem{lem}{Lemma}[section]
\newtheorem{cor}[lem]{Corollary}
\newtheorem{prop}[lem]{Proposition}
\newtheorem{thm}[lem]{Theorem}

\theoremstyle{definition}
\newtheorem{remark}[lem]{Remark}
\newtheorem{defn}[lem]{Definition}
\newtheorem{eg}[lem]{Example}
\newtheorem{assum}[lem]{Assumption}

\newcommand{\R}{\mathbb{R}}
\renewcommand{\P}{\mathbb{P}}

\newcommand{\E}{\mathbb{E}}
\newcommand{\F}{\mathcal{F}}
\newcommand{\N}{\mathbb{N}}

\renewcommand{\( }{\left(}
\renewcommand{\)}{\right)}
\DeclareMathOperator*{\diver}{div}

\newcommand{\Dcal}{{\mathcal D}}
\newcommand{\Ecal}{{\mathcal E}}
\newcommand{\Fcal}{{\mathcal F}}

\newcommand{\Mcal}{{\mathcal M}}

\newcommand{\Pcal}{{\mathcal P}}

\numberwithin{equation}{section}

\makeatletter
\renewenvironment{proof}[1][\proofname] {\par\pushQED{\qed}\normalfont\topsep6\p@\@plus6\p@\relax\trivlist\item[\hskip\labelsep\bfseries#1\@addpunct{.}]\ignorespaces}{\popQED\endtrivlist\@endpefalse}
\makeatother

\begin{document}
\sloppy
			\title{Open Markets and Hybrid Jacobi Processes\footnote{Acknowledgements: The first author would like to acknowledge the support of the Natural Sciences and Engineering Research Council of Canada (NSERC). 
					This work has been partially supported by the National Science Foundation under grant NSF DMS-2206062.		We would like to thank Kostas Kardaras, Ioannis Karatzas and Kiseop Lee for helpful discussions. We also thank two anonymous referees for comments and suggestions that improved the manuscript.}}
			\author{
				David Itkin\footnote{Department of Mathematics, Imperial College London, \texttt{d.itkin@imperial.ac.uk}}
				\and
				Martin Larsson\footnote{Department of Mathematical Sciences, Carnegie Mellon University, \texttt{larsson@cmu.edu}}
			}
		\maketitle

		
		
	
	\begin{abstract}
		We propose a unified approach to several problems in Stochastic Portfolio Theory (SPT), which is a framework for equity markets with a large number $d$ of stocks. Our approach combines \emph{open markets}, where trading is confined to the top $N$ capitalized stocks as well as the market portfolio consisting of all $d$ assets, with a parametric family of models which we call \emph{hybrid Jacobi processes}. We provide a detailed analysis of ergodicity, particle collisions, and boundary attainment, and use these results to study the associated financial markets. Their properties include (1) stability of the capital distribution curve and (2) explicit and not artificially leveraged growth optimal strategies. The sub-class of \emph{rank Jacobi models} are additionally shown to (3) serve as the worst-case model for a robust asymptotic growth problem under model ambiguity and (4) exhibit stability in the large-$d$ limit. Our definition of an open market is a relaxation of existing definitions which is essential to make the analysis tractable.
	\end{abstract}
	


\tableofcontents

\section{Introduction}

In this paper we propose a unified approach to several problems in Stochastic Portfolio Theory (SPT) by combining \emph{open markets} (explained further below) with a parametric family of models which we call \emph{hybrid Jacobi processes}. 

SPT was introduced by Fernholz \cite{fernholz2002stochastic,fernholz1999diversity} as a descriptive theory of large equity markets. A key object of study in SPT is the market weight vector $(X_1,\ldots,X_d)$ consisting of the relative capitalizations of the available stocks. It is an empirical fact that the capital distribution curve, which consists of the ranked market weights $X_{(1)} \ge \ldots \ge X_{(d)}$, has remained remarkably stable over time in US equity markets; see \cite[Chapter 5]{fernholz2002stochastic}. Mathematically, this stability can be captured using ranked market weight processes that are ergodic. We note that for stability to be present it is not required that the market weight process is Markovian. Indeed, setups such as \cite{kardaras2021ergodic,itkin2020robust,itkin2022ergodic} allow for more general ergodic processes and we also take this view in Section~\ref{sec:robust}. Nevertheless, for tractability, a number of ergodic Markovian probabilistic models have been proposed, e.g.\ \cite{banner2005atlas,fernholz2005relative,ichiba2011hybrid,cuchiero2019polynomial,pickova2014generalized,filipovic2016polynomial,fernholz2002stochastic} and the hybrid Jacobi market introduced in this paper also models the market weights with a Markov process.

However, in the classical \emph{closed market} setup, models of this type suffer from certain deficiencies. The underlying issue is that the ergodicity of the ranked market weight process forces small capitalization stocks to eventually grow. Indeed, this is the case because in such models the bankruptcy of assets is prohibited and the market is stable. Of course, in real world equity markets, small companies need not grow and may eventually default. An observable consequence of this effect in closed market models is that they produce growth optimal strategies with artificially high leverage. A typical example of this phenomenon is discussed in Example~\ref{sec:vol_stabilized_example} below. The stability and lack of bankruptcy encourages extreme long positions in the small stocks, financed by corresponding short positions in large capitalization stocks. Consequently, although leveraged strategies may very well be growth-optimal in the real world, the closed market setup mechanically creates additional leverage. This makes the growth-optimal strategy \emph{artificially leveraged} in many closed market setups.

A related phenomenon is that the growth optimal strategy tends to be strongly dependent on the number $d$ of assets included in the model. This is an issue because in practice $d$ is a design parameter. Rather than explicitly modelling \emph{all} available stocks, one restricts attention to, say, a sufficiently liquid subset.

We claim that the tension between the stability of the capital distribution curve and the artificially leveraged behaviour of the growth optimal strategy can be resolved by working with an \emph{open market}. An open market is one where the assets available for trading change over time. Concretely, the investor is allowed to trade in a fixed number of large capitalization stocks, but is constrained from holding small capitalization stocks. In such a setting the smallest investible asset can be overtaken by a smaller asset present in the model. In this case the investor must trade out of their position in the asset which was overtaken and invest in the overtaking asset. As such, open markets do not suffer from the deficiency described above. Namely, one cannot be certain, even when the ranked market weights are ergodic, that the smallest investible assets will necessarily provide growth. 

The study of open markets has recently gained considerable traction. In \cite{fernholz2018numeraire} the author constructs an open market under which the growth optimal strategy is the open market portfolio. In \cite{karatzas2020open}, a theory is developed for, among other things, arbitrage, growth optimality and numeraire properties for open markets. In \cite{campbell2021functional}, the authors numerically optimize growth under open market constraints.

However, tractability remains a concern for open markets. It has proven difficult to produce concrete non-trivial open market models that admit explicit growth optimal strategies. Explicit strategies are useful because their properties can be understood by inspection.

We overcome this lack of tractability by relaxing the existing notion of an open market. Fixing $N \leq d$ we allow the investor to trade in (i) the $N$ assets with the largest capitalizations and (ii) the \emph{full} market portfolio consisting of all $d$ assets. As such, we do allow trading in small capitalization stocks, but only through the market portfolio. Due to the existence of various market-tracking securities, such as the Wilshire 5000 or the CRSP US Total Market Index, we consider (ii) to be a practical inclusion that is approximately implementable through proxies. Under a structural condition on the covariation between large and small capitalization stocks, this setup is tractable: the growth optimization problem in the open market becomes as easy or difficult as the classical growth optimization problem in the full market. These results are obtained in Section~\ref{sec:open_market_SPT}.

To demonstrate that the open market setup indeed resolves the tension discussed above, we instantiate our framework using a parametric family of market weight models which we call \emph{hybrid Jacobi processes}. These processes resemble multivariate Jacobi processes and Wright--Fisher diffusions, but have the feature that the drift of a particle may depend on its rank in addition to its name, akin to \cite{ichiba2011hybrid}. Section~\ref{sec:hybrid_polynomial_models} defines and constructs hybrid Jacobi processes, and establishes a number of properties including ergodicity, invariant density, particle collisions, and boundary attainment.

In Section~\ref{sec:hybrid_market} we use the hybrid Jacobi processes to model the market weights. We show that they fit into the framework of Section~\ref{sec:open_market_SPT} and, consequently, obtain necessary and sufficient conditions for the existence of a growth optimal strategy in the open market as well as an explicit formula for it (when it exists). Section~\ref{sec:examples} contains several examples. In particular, we demonstrate that this model can resolve the tension between the stability of the capital distribution curve and the not artificially leveraged behaviour of the growth optimal strategy. Moreover, we show that many specifications of the parameters lead to growth-optimal strategies that are stable with respect to $d$, the total number of stocks in the model. 

Like any parametric models, the hybrid Jacobi processes only represent idealizations of real-world dynamics. It is not clear a priori to what extent the model output is sensitive to model misspecification. This is of particular concern for the drift dynamics, which are notoriously difficult to estimate based on financial data. Following \cite{kardaras2021ergodic,itkin2020robust, itkin2022ergodic}, we address this point of criticism by solving an \emph{asymptotic} and \emph{robust} growth optimization problem in the open market, with only two inputs fixed a priori: the covariation matrix of the ranked market weights and their invariant density. These inputs are chosen to be consistent with a \emph{rank Jacobi process}, whose drift is purely rank-based. We maximize the worst-case asymptotic growth rate across the class of models consistent with these inputs  (and which obey certain technical restrictions), and show that the optimal solution is the growth optimal strategy from the rank Jacobi model. This is done in Section~\ref{sec:robust}. In Section~\ref{sec:conclusion} we discuss open problems and directions for future research.


In summary, in this paper we show how to address, in a single framework, the following four qualitative properties that we believe a satisfactory SPT model ought to possess.
\begin{enumerate} [label = ({\arabic*}),noitemsep]
	\item The capital distribution curve is stable,
	\item The growth optimal strategy is not artificially leveraged,
	\item The model output is robust to misspecification and only depends on estimable quantities,
	\item The model output is stable with respect to $d$, the total number of stocks.
\end{enumerate}

Some material pertaining to Sections~\ref{sec:open_market_SPT}--\ref{sec:robust} are placed in the appendix. Appendix~\ref{app:integral} contains useful integral identities used in the paper, while Appendices~\ref{app:hybrid_polynomial_models}-\ref{app:robust} contain the proofs of certain results omitted in the main text.

\subsection{Notation} The following notation is used throughout the paper.
\begin{itemize}
	\item $e_1,\dots,e_d$ denote the standard basis vectors in $\R^d$, and $\boldsymbol{1}_d := \sum_{i=1}^d e_i$ denotes the $d$-dimensional vector of all ones. We write $\delta_{ij}$ for the Kronecker delta.
	\item $\mathcal{T}_d$ is the set of permutations on $\{1,\dots,d\}$. For $x \in \R^d$ and $\tau \in \mathcal{T}_d$ we write $x_\tau$ for the vector $(x_{\tau(1)},\dots,x_{\tau(d)})$. We write $x_{()} = (x_{(1)},\dots,x_{(d)})$ for the \emph{ranked} vector. This is the permutation of $x$ that satisfies $x_{(1)} \geq x_{(2)} \geq \dots \geq x_{(d)}$.
	\item We define the \emph{rank identifying function} ${\bf{r}}_i:\R^d \to \{1,\dots,d\}$ for $i=1,\ldots,d$, as well as the \emph{name identifying function} ${\bf{n}}_k:\R^d \to \{1,\dots,d\}$ for $k=1,\ldots,d$, via
	\begin{equation} \label{eq_rank_name_def}
		\begin{aligned}
			\text{${\bf{r}}_i(x) = k$, where $k$ is such that $x_i = x_{(k)}$} \\
			\text{${\bf{n}}_k(x) = i$, where $i$ is such that $x_i = x_{(k)}$}
		\end{aligned}
	\end{equation}
	with ties broken by lexicographical ordering.
	\item The standard simplex in $\R^d$ and the ordered simplex in $\R^d$ play an important role. They are given by
	\begin{equation} \label{eq_simplex_notation}
		\Delta^{d-1} = \left\{x \in [0,\infty)^d \colon  x_1 + \cdots + x_d = 1\right\}, \quad
		\nabla^{d-1} = \left\{y \in \Delta^{d-1} \colon  y_1 \ge \cdots \ge y_d \right\}.
	\end{equation}
	We also define the interiors $\Delta^{d-1}_+ = \{x \in \Delta^{d-1} \colon x_{(d)} > 0\}$, $\nabla^{d-1}_+ = \{y \in \nabla^{d-1} \colon y_d > 0\}$ and boundaries $\partial \Delta^{d-1} = \Delta^{d-1}\setminus \Delta^{d-1}_+$, $\partial \nabla^{d-1} = \nabla^{d-1}\setminus \nabla^{d-1}_+$.
	\item $L_X^a$ is the local time process of a continuous semimartingale $X$ at level $a \in \R$, and we write $L_X$ for $L_X^0$. The local time is defined by $dL_X^a(t) := d|X(t) - a| - \text{sign}(X(t) - a)dX(t)$ and $L_X^a(0) = 0$.
\end{itemize}

\section{Hybrid Jacobi Models} \label{sec:hybrid_polynomial_models}

We now introduce a parametric family of models which we call \emph{hybrid Jacobi processes}. These models are constructed in a standard way using Dirichlet forms in Section~\ref{sec:hybrid_construction}, and we perform a detailed analysis of properties such as ergodicity, particle collisions, and boundary attainment in Sections~\ref{sec:ergodicity}, \ref{sec:collisions} and \ref{sec:boundary_attainment} respectively. In later sections we demonstrate their usefulness in the context of SPT. We refer the reader to \cite{Fukushima1994Dirichlet} for an in-depth treatment of the theory of Dirichlet forms.

For a vector $x \in \R^d$ we define the \emph{tail sum} starting from component $k = 1,\ldots,d$ by
\begin{equation} \label{eq_tail_sum_notation}
	\bar x_k = \sum_{l=k}^d x_l.
\end{equation}
We also write $\bar x_{(k)} = \sum_{l=k}^d x_{(l)}$ for the tail sums of $x_{()}$. This slight abuse of notation is to avoid cumbersome expressions like $\overline{x_{()}}_k$.

\begin{defn} \label{def:hybrid_Jacobi}
	Fix parameters $a = (a_1,\ldots,a_d) \in \R^d$, $\gamma = (\gamma_1,\ldots,\gamma_d) \in \R^d$, $\sigma > 0$. A \emph{hybrid Jacobi process} is a $\Delta^{d-1}$-valued weak solution of the SDE
	\begin{equation} \label{eqn:X_polynomial_dynamics}
		\begin{aligned}
			dX_i(t) &= \frac{\sigma^2}{2}\(\gamma_i + a_{{\bf r}_i(t)}- (\bar a_1 + \bar \gamma_1) X_i(t)\)dt \\
			&\quad + \sigma \sum_{j=1}^d \left(\delta_{ij} - X_i(t) \right) \sqrt{ X_j(t) } \, dW_j(t), \quad i = 1,\ldots,d,
		\end{aligned}
	\end{equation}
	where $W$ is a standard $d$-dimensional Brownian motion and ${\bf r}_i(t) = {\bf r}_i(X(t))$ is the rank identifying function as in \eqref{eq_rank_name_def}. When $\gamma = 0$ we call it a \emph{rank Jacobi process}.\end{defn}

\begin{remark}
	If $a = 0$, \eqref{eqn:X_polynomial_dynamics} reduces to the standard (multivariate) Jacobi process introduced in \cite{Gourie2006Multi}. This is also a special case of the Wright--Fisher diffusion model in population genetics; see e.g.\ \cite[Chapter~10]{Ethier1986Markov}.
\end{remark}

Hybrid Jacobi processes do not exist for arbitrary parameter values. As we show in Section~\ref{sec:hybrid_construction}, the process exists provided the following condition is satisfied.
\begin{assum} \label{ass:bar_a}
	The vectors $a$ and $\gamma$ are such that $ \bar a_k + \bar \gamma_{(k)} > 0$ for $k=2,\dots,d$.
\end{assum}

If $X$ is a hybrid Jacobi process with parameters $a,\gamma,\sigma$, a calculation shows that $d[ X_i, X_j](t) = c_{ij}(X(t))dt$ where the diffusion matrix is given by
\begin{equation} \label{eqn:c_def}
	c_{ij}(x)  = \sigma^2 x_i(\delta_{ij} -x_j), \quad  i,j=1,\dots,d, \quad x \in \Delta^{d-1}.
\end{equation}
Here, and in the sequel, we denote by $C^k(\Delta^{d-1})$ for any $k \in \N \cup \{\infty\}$ to be the set of all functions $u = v|_{\Delta^{d-1}}$, where $v \in C^k(\R^d)$.
Using this notation the generator of $X$ takes the form
\begin{equation} \label{eqn:generator_d}
	Lu(x) = \frac{1}{2}\sum_{i,j=1}^d c_{ij}(x)\partial_{ij}u(x) + \frac{\sigma^2}{2}\sum_{i=1}^d \(\gamma_i +  a_{{\bf r}_i(x)} - ( \bar a_1 + \bar \gamma_1)x_i \)\partial_i u(x)
\end{equation} 
for functions $u \in C^2(\Delta^{d-1})$ and $x \in \Delta^{d-1}$.  Moreover, a useful consequence of \eqref{eqn:c_def} is that the quadratic covariation between two linear functions of a hybrid Jacobi process takes on a particularly simple form. Indeed, for any $u,v \in \R^d$ we have 
\begin{equation} \label{eqn:qv_formula}
	d[u^\top X,v^\top X](t) = \big((u \circ v)^\top X(t) - (u^\top X(t))(v^\top X(t))\big)dt,
\end{equation}
where $u\circ v = (u_1v_1,\dots,u_dv_d)$ is the Hadamard componentwise product.

\subsection{Construction} \label{sec:hybrid_construction}

Throughout the rest of Section~\ref{sec:hybrid_polynomial_models} we fix parameters $a,\gamma,\sigma$ as in Definition~\ref{def:hybrid_Jacobi}. The construction of the hybrid Jacobi process involves a certain probability density $p(x)$ on $\Delta^{d-1}$, which will turn out to be the invariant density of the process. This density is understood with respect to the surface area measure on $\Delta^{d-1}$. More precisely, we use the pushforward of Lebesgue measure on $\R^{d-1}$ under the map
\[
(x_1,\ldots,x_{d-1}) \mapsto (x_1,\ldots,x_{d-1}, 1 - x^\top {\bf 1}_{d-1}).
\]
All integrals over $\Delta^{d-1}$ and $\nabla^{d-1}$ (the ordered simplex; see \eqref{eq_simplex_notation}) should be understood with respect to this measure, unless otherwise indicated. More details on this convention, along with several useful integral identities, are given in Appendix~\ref{app:integral}.

We now define
\begin{equation} \label{eqn:p_def}
	p(x) = Z^{-1}\prod_{k=1}^d x_{(k)}^{a_k + \gamma_{{\bf n}_k(x)}-1}, \quad x \in \Delta^{d-1}_+, 
\end{equation}
where the normalizing constant $Z$ is given by
\[
Z = \int_{\Delta^{d-1}} \prod_{k=1}^d x_{(k)}^{a_k + \gamma_{{\bf n}_k(x)}-1}dx.
\]
For later use we also define
\begin{equation} \label{eqn:q_general_def}
	q(y) = \frac{1}{d!}\sum_{ \tau \in \mathcal{T}_d} p(y_\tau), \quad y \in \nabla^{d-1}_+,
\end{equation}
which will turn out to be the invariant density of the ranked process. The next lemma clarifies the role of Assumption~\ref{ass:bar_a}. Its proof is contained in Appendix~\ref{app:hybrid_polynomial_models}.

\begin{lem} \label{lem:finite_Z}
	$Z < \infty$ if and only if Assumption~\ref{ass:bar_a} holds.
\end{lem}

Assume that Assumption~\ref{ass:bar_a} holds, so that $m(dx) := p(x)dx$ is a probability measure on $\Delta^{d-1}$. We define a symmetric bilinear form on $L^2(\Delta^{d-1},m)$ with domain $\Dcal = C^\infty(\Delta^{d-1})$ by 
\begin{equation} \label{eqn:X_Dirichlet_form}
	\mathcal{E}(u,v) := \frac{1}{2}\int_{\Delta^{d-1}} \nabla u^\top c \nabla v p. 
\end{equation}
The form $(\mathcal{E},\mathcal{D})$ is Markovian in the sense of \cite[Ch.~1, $(\Ecal.4)$]{Fukushima1994Dirichlet}. It is also closable as defined in \cite[Ch.~1, $(\Ecal.3)$]{Fukushima1994Dirichlet}. To show this we need the following integration by parts formula, whose proof is contained in Appendix~\ref{app:hybrid_polynomial_models}.

\begin{lem}[Integration by parts] \label{thm:IBP}
	Let Assumption~\ref{ass:bar_a} be satisfied. For any functions $v \in C^1(\Delta^{d-1})$ and $\xi \in C^1(\Delta^{d-1};\R^d)$ we have
	\begin{equation}\label{eqn:IBP}
		\frac{1}{2}\int_{\Delta^{d-1}} \nabla v^\top c \xi p = - \frac{1}{2}\int_{\Delta^{d-1}} v \diver(c\xi p).
	\end{equation}
	In particular, taking $\xi = \nabla u$ for some $u \in C^2(\Delta^{d-1})$, the identity \eqref{eqn:IBP} reads
	\[
	\mathcal{E}(u,v) = - \int_{\Delta^{d-1}} v \, Lu \, p
	\]
	where $L$ is the generator defined in \eqref{eqn:generator_d}. 
\end{lem}

Thanks to Lemma~\ref{thm:IBP} and \cite[Proposition~I.3.3]{Ma1992Intro}, $(\mathcal{E},\mathcal{D})$ is closable. We let $(\mathcal{E},D(\mathcal{E}))$ be the closure of $(\mathcal{E},\mathcal{D}).$ A standard application of \cite[Theorem~3.1.2]{Fukushima1994Dirichlet} then shows that $(\mathcal{E},D(\mathcal{E}))$ is a regular, strongly local Dirichlet form. Moreover, Lemma~\ref{thm:IBP} implies that the generator associated with $\mathcal{E}$ coincides with $L$ on functions in $\mathcal{D}$. With some abuse of notation we let $(L,D(L))$ denote the generator of $(\mathcal{E},D(\mathcal{E}))$. Existence of the hybrid Jacobi process now follows from the well-known correspondence between local regular Dirichlet forms and Hunt diffusions \cite[Theorem~7.2.2]{Fukushima1994Dirichlet}. Recall that a Hunt diffusion on $\Delta^{d-1}$ is a continuous adapted process $X$ with values in $\Delta^{d-1}$, defined on a filtered space $(\Omega,\F,(\F(t))_{t \geq 0})$, along with a family $(\P_x)_{x \in \Delta^{d-1}}$ of probability measures such that under each $\P_x$, $X$ is a strong Markov process with $X(0)=x$. See \cite[Appendix~A.2]{Fukushima1994Dirichlet} for details.

We may now state the existence theorem for hybrid Jacobi processes. The proof is standard, but for the benefit of the reader we give a proof in Appendix~\ref{app:hybrid_polynomial_models}.

\begin{thm} \label{thm:process_existence}
	Let Assumption~\ref{ass:bar_a} be satisfied. Then there exists a Hunt diffusion $(\Omega,\F,(\F(t))_{t \geq 0}, X, (\P_x)_{x \in \Delta^{d-1}})$ and a Borel set $N \subset \Delta^{d-1}$ with $m(N) = 0$ such that for every $x \in \Delta^{d-1}\setminus N$ the following two properties hold:
	\begin{enumerate} [label = ({\roman*})]
		\item $\P_x(X(t) \in N \text{ for some } t > 0) = 0$, \label{item:existence_i}
		\item $u(X(t)) - u(X(0)) - \int_0^t Lu(X(s))ds$ is a $\P_x$-martingale for every $u \in D(L)$. \label{item:existence_ii}
	\end{enumerate}
	In particular, $X$ is a hybrid Jacobi process under $\P_x$ for every $x \in \Delta^{d-1}\setminus N$.
\end{thm}

\begin{remark}
	If the absolute continuity condition \cite[Equation~(4.2.9)]{Fukushima1994Dirichlet} for the Markov transition function holds, then the set $N$ in Theorem~\ref{thm:process_existence} can be taken to be the empty set. We do not pursue verification of the absolute continuity condition in this paper as the existence of an exceptional set $N$ is sufficient for our purposes.
\end{remark}

\begin{remark}
	In this paper we do not treat the question of uniqueness in law, or pathwise uniqueness, of solutions to \eqref{eqn:X_polynomial_dynamics}. As will be seen in Theorem~\ref{thm:boundary_attainment} below, the hybrid Jacobi process hits different regions of the boundary under different choices of the parameters $a$ and $\gamma$. As such, the laws of the processes when varying the parameters $a$ and $\gamma$ are not equivalent. Hence, a removal of drift approach to study \eqref{eqn:X_polynomial_dynamics} does not apply. Moreover, the volatility degenerates at the boundary of the simplex, causing additional difficulties in the study of uniqueness. As such, we leave the important question of both pathwise uniqueness and uniqueness in law for hybrid Jacobi models to future research.
\end{remark}
We define the class of probability measures 
\[\mathcal{P}_0 = \{\mu \in \mathcal{P}(\Delta^{d-1}): \  \mu(N) = 0 \text{ and } \mu(\partial \Delta^{d-1}) = 0 \},\]
with $N$ as in the statement of Theorem~\ref{thm:process_existence}. The elements of $\mathcal{P}_0$ can serve as initial laws for the hybrid Jacobi process $X$. We restrict attention to laws which do not charge $\partial \Delta^{d-1}$ to simplify some of the statements and proofs to come. Note, in particular, that since both $N$ and $\partial \Delta^{d-1}$ are $m$-nullsets we have that $m$, and every measure absolutely continuous to it, are members of $\mathcal{P}_0$. For $\mu \in \mathcal{P}_0$ we denote by $\P_\mu$ the law of the process $X$ constructed in Theorem~\ref{thm:process_existence} with initial law $\mu$. That is, $\P_\mu(A) = \int_{\Delta^{d-1}} \P_x(A) \mu(dx)$ for $A \in \Fcal$.

\subsection{Ergodicity} \label{sec:ergodicity}
The Dirichlet form construction implies ergodicity of the process $X$.

\begin{thm}[Ergodic property] \label{thm:ergodic}
	Let Assumption~\ref{ass:bar_a} be satisfied. The hybrid Jacobi process of Theorem~\ref{thm:process_existence} is ergodic with invariant density $p$ given by \eqref{eqn:p_def}. In particular, the following Birkhoff theorem holds: for any $f\colon \Delta^{d-1} \to \R$ such that $\int_{\Delta^{d-1}}|f|p < \infty$ we have
	\[\lim_{T \to \infty} \frac{1}{T} \int_0^T f(X(t))dt = \int_{\Delta^{d-1}}fp, \quad \P_\mu\text{-a.s.},\]
	for every $\mu \in \mathcal{P}_0$.
\end{thm}
\begin{proof}
	In view of \cite[Theorem~4.7.3]{Fukushima1994Dirichlet} it suffices to establish that $\mathcal{E}$ is a recurrent and irreducible Dirichlet form. Recurrence follows from \cite[Theorem~1.6.3]{Fukushima1994Dirichlet} since $1 \in \mathcal{D}$. To establish irreducibility consider the modified (pre-)Dirichlet form on the state space $\Delta^{d-1}_+$ given by
	\[\mathcal{E}_+(u,v) := \int_{\Delta^{d-1}_+}\nabla u^\top c \nabla v p, \quad u,v \in C^\infty_c(\Delta^{d-1}_+).\]
	The closure $(\mathcal{E}_+,D(\mathcal{E}_+))$ is a regular Dirichlet form and the conditions of \cite[Example~4.6.1]{Fukushima1994Dirichlet} are satisfied due to the positive definiteness of $c(x)$ and positivity of $p(x)$ on $\Delta^{d-1}_+$. Hence $\mathcal{E}_+$ is irreducible. Now suppose that $A \subseteq \Delta^{d-1}$ is an invariant set for the semigroup corresponding to $\mathcal{E}$.  By \cite[Theorem~1.6.1]{Fukushima1994Dirichlet} this is equivalent to the identity $\mathcal{E}(u,v) = \mathcal{E}(1_Au,1_Av)  + \mathcal{E}(1_{\Delta^{d-1}\setminus A}u,1_{\Delta^{d-1}\setminus A}v)$ for every $u,v \in D(\mathcal{E})$. Since $1_A = 1_{A \cap \Delta^{d-1}_+}$ in $L^2(\Delta^{d-1},m)$ we can assume without loss of generality that $A \subseteq \Delta^{d-1}_+$. By the local property of $\mathcal{E}$ and $\mathcal{E}_+$ we have that $\mathcal{E}(1_Au,1_Av) = \mathcal{E}_+(1_Au,1_Av)$ for $u,v \in C_c^\infty(\Delta^{d-1}_+)$ and an analogous expression holds for the term involving $1_{\Delta^{d-1}\setminus A}$. Thus we obtain
	\begin{equation} \label{eqn:irreducability} 
		\mathcal{E}_+(u,v) = \mathcal{E}_+(1_Au,1_Av)  + \mathcal{E}_+(1_{\Delta^{d-1}_+\setminus A}u,1_{\Delta^{d-1}_+\setminus A}v), \quad u,v \in C_c^\infty(\Delta^{d-1}_+).
	\end{equation}
	Since $C_c^\infty(\Delta^{d-1}_+)$ is a core for $\mathcal{E}_+$ it follows that \eqref{eqn:irreducability} holds for all $u,v \in D(\mathcal{E}_+)$. Thus, again by \cite[Theorem~1.6.1]{Fukushima1994Dirichlet}, we deduce that $A$ is invariant for the semigroup corresponding to $\mathcal{E}_+$. But then by irreducibility of $\mathcal{E}_+$ it must be that either $m(A) = 0$ or $m(\Delta^{d-1}_+\setminus A) = 0$. This establishes irreducibility of $\mathcal{E}$ and completes the proof. 
\end{proof}
Although the ranked process $X_{()}$ is in general not Markov, the ergodicity of $X$ implies that a Birkhoff-type theorem holds for $X_{()}$ as well. We state this fact as a corollary.
\begin{cor}[Rank ergodic property] \label{cor:rank_ergodic}
	Let Assumption~\ref{ass:bar_a} be satisfied and let $X$ be the hybrid Jacobi process of Theorem~\ref{thm:process_existence}. For any $f\colon \nabla^{d-1} \to \R$ such that $\int_{\nabla^{d-1}} |f|q < \infty$ we have
	\begin{equation}\label{eqn:rank_ergodic}\lim_{T \to \infty} \frac{1}{T} \int_0^T f(X_{()}(t))dt = \int_{\nabla^{d-1}} fq, \quad \P_\mu\text{-a.s.},
	\end{equation}
	for every $\mu \in \mathcal{P}_0$, where $q$ is given by \eqref{eqn:q_general_def}.
\end{cor}
Corollary~\ref{cor:rank_ergodic} follows from Theorem~\ref{thm:ergodic}, a change of variables and symmetry of $\Delta^{d-1}$.
Equation \eqref{eqn:rank_ergodic} is an important property for $X_{()}$ to possess when, in Section~\ref{sec:hybrid_market}, we use $X$ to model the market weights in a financial market. The ergodic property of $X_{()}$ then encapsulates the stability of the capital distribution curve.

\subsection{The Study of Collisions} \label{sec:collisions}
Next we tackle the question of particle collisions. The following proposition shows that boundary attainment and particle collisions only occur at a Lebesgue nullset of time points. The proof is contained in Appendix~\ref{app:hybrid_polynomial_models}. 

\begin{prop} \label{prop:Lebesgue_collision} Fix $\mu \in \mathcal{P}_0$ and let $X$ be the hybrid Jacobi process of Theorem~\ref{thm:process_existence}.
	\begin{enumerate}[label = ({\roman*})]
		\item For every $i=1,\dots,d$ we have that the set $\{t\geq 0: X_i(t) = 0\}$ is $\P_\mu$-a.s.\ a Lebesgue nullset. \label{item:boundary_zero}
		\item The set $\{t\geq 0: X_i(t) = X_j(t) \text{ for some } i \ne j\}$ is $\P_\mu$-a.s.\ a Lebesgue nullset. \label{item:collision_zero}
		\item Consider the event $B = \{x  \in \Delta^{d-1}_+: x_i = x_j = x_k  \text{ for some distinct indices } i,j,k\}.$ Then $\P_\mu(X(t) \in B \text{ for some } t > 0) = 0$. \label{item:triple_zero}
	\end{enumerate}
\end{prop}

Proposition~\ref{prop:Lebesgue_collision}\ref{item:triple_zero} establishes that triple collisions do no occur in $\Delta^{d-1}_+$, but, as we will see in Theorem~\ref{thm:boundary_attainment}, they may occur on the boundary of the simplex.

We now use Proposition~\ref{prop:Lebesgue_collision} to study the semimartingale decomposition of the ranked process $X_{()}$. This process plays an important role both in our applications and subsequent results on boundary attainment.

\begin{prop}[Ranked process dynamics] \label{prop:ranked_dynamics}
	Let Assumption~\ref{ass:bar_a} be satisfied, and let $X$ be the hybrid Jacobi process of Theorem~\ref{thm:process_existence}. Omitting time arguments for brevity, the dynamics of $X_{()}$ are given by 
	\begin{equation} \label{eqn:X_()_dynamics}
		\begin{aligned}
			dX_{(k)} &= \frac{\sigma^2}{2}\(\gamma_{{\bf n}_k(\cdot)} + a_k - (\bar a_1 + \bar \gamma_1) X_{(k)}\)dt \\
			&\quad + \sigma \sum_{l=1}^d \left(\delta_{kl} - X_{(k)} \right) \sqrt{X_{(l)}} \, d\tilde W_l - \frac{1}{4}dL_{k-1,k} + \frac{1}{4}dL_{k,k+1}
		\end{aligned}
	\end{equation}
	for $k=1,\dots,d$ where $\tilde W$ is a $d$-dimensional standard Brownian motion and we define $L_{k,k+1} = L_{X_{(k)} - X_{(k+1)}}$ with the convention that $L_{0,1} = L_{d,d+1} = 0$. In particular, the diffusion matrix of $X_{()}$ is given by
	\begin{equation}\label{eqn:kappa_def} 
		\kappa_{kl}(y)  = \sigma^2 y_k(\delta_{kl} - y_l), \quad  k,l=1,\dots,d, \quad   y \in \nabla^{d-1}. 
	\end{equation}
\end{prop}

\begin{proof}
	By the results of \cite{Banner2008Local}, $X_{()}$ is a semimartingale with semimartingale decomposition given by \cite[Theorem~2.3]{Banner2008Local}. Since, by Proposition~\ref{prop:Lebesgue_collision}\ref{item:collision_zero}, collisions only occur at a Lebesgue nullset of time points, in order to obtain \eqref{eqn:X_()_dynamics} it suffices to show that $L_{X_{(k)} - X_{(l)}} = 0$ whenever $l-k \geq 2$. By Proposition~\ref{prop:Lebesgue_collision}\ref{item:triple_zero} triple collisions do not occur in the open simplex, and hence $dL_{X_{(k)} - X_{(l)}} = 1_{\{X_{(k)} = 0\}}dL_{X_{(k)} - X_{(l)}}$ whenever $l-k \geq 2$. Using the approximation for the local time process (see e.g.\ \cite[Corollary~1.9]{Revuz1999conti}) we then have for every $T \geq 0$ that
	\begin{align*}
		\int_0^T&1_{\{X_{(k)}(t) = 0\}}dL_{X_{(k)} - X_{(l)}}(t)\\
		& = \lim_{\epsilon \downarrow 0} \frac{1}{\epsilon}\int_0^T 1_{\{X_{(k)}(t) = 0\}}1_{\{0 \leq X_{(k)}(t) - X_{(l)}(t) < \epsilon\}}d [X_{(k)} - X_{(l)},X_{(k)} - X_{(l)}](t) \\
		& = \lim_{\epsilon \downarrow 0} \frac{1}{\epsilon}\int_0^T1_{\{X_{(k)}(t) = X_{(l)}(t) = 0\}}  (X_{(k)}(t) + X_{(l)}(t)- (X_{(k)}(t) - X_{(l)}(t))^2)ds \\
		& = 0,
	\end{align*}
	where in the second equality we used that $\{X_{(k)}(t) = 0\} = \{X_{(k)}(t) = X_{(l)}(t) = 0\}$ and that $d [X_{(k)} - X_{(l)},X_{(k)} - X_{(l)}](t) = (X_{(k)}(t) + X_{(l)}(t)- (X_{(k)}(t) - X_{(l)}(t))^2)dt$. Hence $L_{X_{(k)} - X_{(l)}} = 0$ whenever $l-k \geq 2$ which completes the proof.
\end{proof}

\begin{remark} \label{rem:rank_Dirichlet}
	It is clear from \eqref{eqn:X_()_dynamics} that $X_{()}$ is not a Markov process in general, as the drift depends on the name that occupies each rank. However, in the special case $\gamma = 0$ the process $X_{()}$ is Markov. This can be verified by studying well-posedness of the reflected stochastic differential equation \eqref{eqn:X_()_dynamics} that $X_{()}$ solves.
	Note that in this case $q$ becomes
	\begin{equation} \label{eqn:q_def}
		q(y)   = Q_a^{-1}\prod_{k=1}^d y_k^{a_k-1}, \quad  y \in \nabla^{d-1}_+, 
	\end{equation} where $Q_a$ is the normalizing constant; see also Lemma~\ref{lem:Q_finite}.
\end{remark}
We conclude this section by establishing formulas for the ergodic limit of the collision local times for the ranked particles.

\begin{prop}[Ergodic limit for the collision local times] \label{prop:ergodic_local} For $k=2,\dots,d$ let $L_{k-1,k}$ be the collision local times as in Proposition~\ref{prop:ranked_dynamics}. Then
	\[\lim_{T \to \infty} \frac{L_{k-1,k}(T)}{T} = 2\sigma^2(\bar a_k + \E_m[\sum_{l=k}^d \gamma_{{\bf n}_l(X)}]- (\bar a_1 + \bar \gamma_1)\E_m[\bar X_{(k)}]); \qquad \P_\mu\text{-a.s.,}\]
	for every $\mu \in \Pcal_0$, where $\E_m[\cdot]$ denotes expectation under the invariant measure $m$ of Section~\ref{sec:hybrid_construction}.
\end{prop}
\begin{proof}
	Note that from \eqref{eqn:X_()_dynamics} we have for $k=2,\dots,d$, that
	\[d\bar X_{(k)}(t) = \frac{\sigma^2}{2}(\bar a_k + \sum_{l=k}^d \gamma_{{\bf n}_l(t)} - (\bar a_1 + \bar \gamma_1)\bar X_{(k)}(t))\, dt + dM(t) - \frac{1}{4}dL_{k-1,k}(t),\]
	where $M$ is the martingale term. Writing this in integral form, dividing by $T$ and sending $T \to \infty$ yields the result courtesy of the ergodic theorem \eqref{thm:ergodic}, which handles the drift term, together with \cite[Lemma~1.3.2]{fernholz2002stochastic}, which ensures that $M_T/T \to 0$ as $T \to \infty$.
\end{proof}

\subsection{Boundary Attainment} \label{sec:boundary_attainment}
Next we describe when the name and rank processes $X_i$ and $X_{(k)}$ hit zero. This will be important in Section~\ref{sec:hybrid_market} where $X$ is used to model the market weight process for the open market in Section~\ref{sec:open_market_SPT}. For an index set $I \subseteq \{1,\dots,d\}$ we define the function $\Lambda_I:\R^d \to \R$ via  \begin{equation} \label{eqn:Lambda_def}
	\Lambda_I(x) = \sum_{i \in I} x_i.
\end{equation} By convention we set $\Lambda_{\emptyset} \equiv 0$. 
\begin{thm}[Boundary Attainment] \label{thm:boundary_attainment}
	Let Assumption~\ref{ass:bar_a} be satisfied, and let $X$ be the hybrid Jacobi process of Theorem~\ref{thm:process_existence} with initial law $\mu \in \mathcal{P}_0$.
	\begin{enumerate}
		\item \label{item:bound_attain_rank} \textnormal{(Boundary Attainment of Ranks)} For each $k \in \{2,\dots,d\}$ the following holds:
		\smallskip
		\begin{enumerate}[label =({\alph*})]
			\item \label{item:bound_attain_rank_hit} We have $\P_{\mu}(X_{(k)}(t) = 0 \text{ for some } t > 0) = 0$ if and only if
			\[
			\text{$\bar a_l + \bar \gamma_{(l)} \geq 1$ for every $l =2,\dots,k$.}
			\]
			\item \label{item:bound_attain_rank_above} If $\bar a_k + \bar\gamma_{(k)}\geq 1$ then $\P_{\mu}(X_{(k)}(t) = 0 \text{ and } X_{(k-1)}(t) > 0 \text{ for some } t > 0) = 0.$
		\end{enumerate}
		\smallskip
		\item \label{item:bound_attain_name} \textnormal{(Boundary Attainment of Names)}
		Fix a nonempty index set $I \subseteq \{1,\dots,d\}$ with cardinality $N = |I| \le d-1$. Let $\gamma^{-I}_{()}$ denote the $(d-N)$-dimensional vector obtained by ordering the elements $\{\gamma_j\}_{j \not \in I}$ in decreasing order. The following holds:
		\smallskip
		\begin{enumerate}[label = ({\alph*})]
			\item \label{item:bound_attain_name_hit} We have $\P_{\mu}\(\Lambda_I(X(t)) = 0 \text{ for some } t > 0\)  = 0$ if and only if
			\begin{equation} \label{eqn:param_cond_I}
				\text{$\bar a_l + \sum_{i \in I} \gamma_i + \sum_{k=l}^{d-N} \gamma^{-I}_{(k)} \geq 1$ \quad  for every  $l = 2,\dots,d-N+1$.}
			\end{equation}
			\item \label{item:bound_attain_name_above} If $\bar a_{d-N+1} + \sum_{i \in I} \gamma_i \geq 1$ then
			\[
			\text{$\P_\mu(\Lambda_I(X(t)) = 0 \text{ and } \Lambda_J(X(t)) > 0 \text{ for all } J \supsetneq I \text{ and some } t > 0) = 0.$}
			\]
		\end{enumerate}
	\end{enumerate}
\end{thm}

In the important special case of rank Jacobi processes (i.e.\ when $\gamma  = 0$) the conditions on the parameters in Theorem~\ref{thm:boundary_attainment} simplify. We state this special case as a corollary.

\begin{cor} \label{cor:boundary_attainment}
	Let Assumption~\ref{ass:bar_a} be satisfied, let $\gamma=0$, and let $X$ be the rank Jacobi process of Theorem~\ref{thm:process_existence} with initial law $\mu \in \mathcal{P}_0$. Fix $k \in \{2,\dots,d\}$. 
	\begin{enumerate}
		\item \label{item:cor_rank_hit} The following are equivalent:
		\smallskip
		\begin{enumerate}[label = ({\alph*})]
			\item $\P_\mu(X_{(k)}(t) = 0 \text { for some } t > 0) = 0$.
			\item There exists an index set $I \subset \{1,\dots,d\}$ with cardinality $|I| = d-k+1$ such that 
			\begin{equation} \label{eq_cor_boundary_attainment_0}
				\P_\mu(\Lambda_I(X(t))= 0 \text{ for some } t > 0) = 0.
			\end{equation}
			\item Every index set $I \subset \{1,\dots,d\}$ with cardinality $|I| = d-k+1$ satisfies \eqref{eq_cor_boundary_attainment_0}.
			\item $\bar a_l \geq 1$ for every $l = 2,\dots,k$.
		\end{enumerate}
		\smallskip
		\item \label{item:cor_rank_above} Suppose $\bar a_k \geq 1$. Then
		\smallskip
		\begin{enumerate} [label = ({\alph*})]
			\item $\P_{\mu}(X_{(k)}(t) = 0 \text{ and } X_{(k-1)}(t) > 0 \text{ for some } t > 0) = 0.$
			\item There exists an index set $I \subset \{1,\dots,d\}$ with cardinality $|I| = d-k+1$ such that
			\begin{equation} \label{eq_cor_boundary_attainment_1}
				\P_\mu(\Lambda_I(X(t)) = 0 \text{ and } \Lambda_J(X(t)) > 0 \text { for every } J \supsetneq I \text{ and some } t > 0) = 0.
			\end{equation}
			\item Every index set $I \subset \{1,\dots,d\}$ with cardinality $|I| = d-k+1$ satisfies \eqref{eq_cor_boundary_attainment_1}.
		\end{enumerate}
	\end{enumerate}
\end{cor}
\begin{remark}
	\begin{enumerate}[label = ({\roman*})]
		\item Note that $X_{(1)}$ can never hit zero. Indeed since the diffusion $X$ takes values in $\Delta^{d-1}$ we have the lower bound $X_{(1)}(t) \geq 1/d$ for every $t \geq 0$.
		\item Theorem~\ref{thm:boundary_attainment} and Corollary~\ref{cor:boundary_attainment} show that triple collisions may occur at the boundary of the simplex. This is in stark contrast to the behaviour in the interior, as guaranteed by Proposition~\ref{prop:Lebesgue_collision}\ref{item:triple_zero}.
		\item Corollary~\ref{cor:boundary_attainment}\ref{item:cor_rank_above} shows that the tail sum $\bar a_k$ determines the behaviour of $X_{(k)}$ near the boundary in the rank based case. Indeed, if $\bar a_k \geq 1$ then $X_{(k)}$ cannot hit zero on its own volition; it can only do so if it is `pushed' to zero by a larger market weight.
	\end{enumerate} 
\end{remark} 
Theorem~\ref{thm:boundary_attainment} is proved using Lyapunov function techniques and is contained in Appendix~\ref{app:hybrid_polynomial_models}. The simple dynamics of $\Lambda_I(X)$ allow for a very precise analysis of the boundary behaviour. We summarize the dynamics of $\Lambda_I(X)$ in the following lemma.
\begin{lem} \label{lem:Lambda}
	Let index sets $I,J \subseteq \{1,\dots,d\}$ be given and let $X$ be a hybrid Jacobi process. Then $d[\Lambda_I(X),\Lambda_J(X)](t) = \sigma^2(\Lambda_{I \cap J}(X(t)) - \Lambda_I(X(t))\Lambda_J(X(t)))dt$. Moreover, the dynamics of $\Lambda_I(X)$ are given by
	\begin{equation} \label{eqn:Lambda_dynamics}
		\begin{aligned}
			d\Lambda_I(X(t)) &= \frac{\sigma^2}{2}\(\sum_{i \in I} (\gamma_i + a_{{\bf r}_i(t)}) - (\bar a_1 + \bar \gamma_1)\Lambda_I(X(t))\)dt \\
			&\quad +\sigma \sqrt{\Lambda_I(X(t))(1-\Lambda_I(X(t)))}dB(t)
		\end{aligned}
	\end{equation} 
	for some one-dimensional Brownian motion $B$.
\end{lem}
\begin{proof}
	The expression for the quadratic covariation of $\Lambda_I$ with $\Lambda_J$ follows from \eqref{eqn:qv_formula} with $u = v = \sum_{i\in I} e_i$. The expression for \eqref{eqn:Lambda_dynamics} then follows from \eqref{eqn:X_polynomial_dynamics} and Lévy's characterization of Brownian motion.
\end{proof}

We conclude this section with a few examples which exhibit a variety of different behaviours regarding boundary attainment.
\begin{eg} Let Assumption~\ref{ass:bar_a} be satisfied, and let $X$ be the hybrid Jacobi process of Theorem~\ref{thm:process_existence} with initial law $\mu \in \mathcal{P}_0$.
	\begin{enumerate}[label = ({\roman*})]
		\item Assume that $\gamma = \gamma_*\boldsymbol{1}_d$ for some $\gamma_* > 0$ and that $a = 0$. Then for $k=2,\dots,d$,
		\[\P_\mu(X_{(k)}(t)= 0 \text{ for some } t > 0) = 0 \quad \iff \quad \gamma_* \geq \frac{1}{d-k+1}.\]
		In the case that $\gamma_* < 1/(d-k+1)$, for any subset $I \subset \{1,\dots,d\}$ with $|I| = d-k+1$ we have that $\Lambda_I(X)$ hits zero with positive probability. This is summarized in the following number line for $\gamma_*$, which is on a geometric scale.
		\vspace{0.2cm}
		
		\begin{tikzpicture}[font=\Large]
			\draw (0,0) -- (6,0);
			\draw[shift={(0,-0.2)},color=black] node[below] 
			{$\bf 0$};
			
			\draw[shift={(0.05,0)},color=black] node 
			{$($};
			
			\draw[shift={(1.25,0.2)},color=black] node
			{\scriptsize $X_{(2)} \text{ hits } 0$};
			
			\draw[shift={(2.5,-0.2)},color=black] node[below] 
			{$\bf \frac{1}{d-1}$};
			
			\draw[shift={(2.5,0)},color=black] node 
			{$)\hspace{-0.1cm}[$};
			
			\draw[shift={(3.75,0.2)},color=black] node
			{\scriptsize $X_{(3)} \text{ hits } 0$};
			
			\draw[shift={(3.75,-0.3)},color=black] node
			{\scriptsize $X_{(2)} > 0$};
			
			\draw[shift={(5,-0.2)},color=black] node[below] 
			{$\bf \frac{1}{d-2}$};
			
			\draw[shift={(5,0)},color=black] node 
			{$)\hspace{-0.1cm}[$};
			
			\draw[shift={(6,0)},color=black]  node[right]  {$\dots$};
			
			\draw[->] (7,0) -- (13,0);
			
			\draw[shift={(8,-0.2)},color=black] node[below] 
			{$\bf \frac{1}{2}$};
			\draw[shift={(8,0)},color=black] node 
			{$)\hspace{-0.1cm}[$};
			\draw[shift={(9.25,-0.3)},color=black] node
			{\scriptsize $X_{(d-1)} > 0$};
			
			\draw[shift={(9.25,0.2)},color=black] node
			{\scriptsize 
				$X_{(d)} \text{ hits } 0$};

			\draw[shift={(10.5,-0.2)},color=black] node[below] 
			{$\bf 1$};
			\draw[shift={(10.5,0)},color=black] node 
			{$)\hspace{-0.1cm}[$};
			\draw[shift={(11.7,-0.3)},color=black] node
			{\scriptsize $X_{(d)} > 0$};

		\end{tikzpicture}

		\item Let $\gamma = 0$ and assume that $a = \eta e_d$ for some $\eta > 0$ so that Assumption~\ref{ass:bar_a} is satisfied. If $\eta \geq 1$ then no component of $X$ hits zero $\P_\mu$-a.s. On the other hand if $\eta < 1$ then for any collection $I \subset \{1,\dots,d\}$ such that $1 \leq |I| \leq d-1$ we have that $\Lambda_I(X)$ hits zero with positive probability. In particular $X_{(2)}$ hits zero with positive probability.
		\item Assume $d \geq 3$, let $\gamma = 0$ and set $a = - 10^6e_1 - e_2 + \frac{3}{2}e_d$.
		We have 
		\[\bar a_k = \begin{cases}
			-10^6 + 1/2, & k =1, \\
			1/2, & k =2,\\
			3/2,  & k=3,\dots,d.
		\end{cases}\]
		Then $\P_\mu(X_{(2)}(t) = 0 \text{ for some } t > 0) > 0$ since $\bar a_2 < 1$ and so, consequently, the smaller ranked components will hit zero with positive probability as well. However, since $\bar a_k \geq 1$ for $k \geq 3$ the smaller ranks will only hit zero when $X_{(2)}$ does. In particular we have 
		\begin{align*}
			&\P_\mu(X_{()}(t) = e_1 \text{ for some } t > 0) > 0,  \\
			&\P_\mu(X_{()}(t) \in \partial \Delta^{d-1}\setminus \{e_1\} \text{ for some } t > 0) = 0.
		\end{align*}
		Note that the size of $a_1$ does not in any way affect the boundary attainment behaviour.
		\item Let $d =3$ and choose $\gamma = (\frac{1}{2},\frac{1}{3},\frac{1}{4})$ and $a = (0,-\frac{1}{3},\frac{1}{2})$.
		Then 
		$\gamma_1 + a_3 =1$  and $\gamma_i + a_3 < 1$ for $i=2,3$. Consequently $X_{2}$ and $X_{3}$ hit zero with positive probability. Next note that for $i < j$ we have
		\[ \gamma_i + \gamma_j  + a_2 + a_3 = \begin{cases}
			1, & i =1, j =2, \\ 
			\frac{11}{12}, & i = 1, j = 3, \\
			\frac{3}{4}, & i = 2, j = 3.
		\end{cases} \] 
		Hence 
		\begin{align*}
			\P_\mu(X_1(t) + X_2(t) = 0 \text{ for some } t > 0) = 0, \\
			\P_\mu(X_1(t) + X_3(t) = 0 \text{ for some } t > 0) > 0, \\
			\P_\mu(X_2(t) + X_3(t) = 0 \text{ for some } t > 0) > 0.
		\end{align*}
		In particular we see $\P_\mu$-a.s.\ that $X_1(t) = 0 \implies X_3(t)  = 0$ and that $X_{(2)}(t) = 0 \implies X_3(t) = 0$. That is, $X_1$ can hit zero, but not on its own and $X_{(2)}$ can hit zero, but not when simultaneously $X_3$ is the largest weight.
	\end{enumerate}
\end{eg}

\section{Open Markets in Stochastic Portfolio Theory} \label{sec:open_market_SPT}

In this section we switch focus and, in a general nonparametric setting, introduce stochastic portfolio theory and open markets. We present a new setup for open markets, relaxing the setup considered in \cite{karatzas2020open}, and under a condition on the covariation between small and large capitalization stocks \eqref{eqn:kappa_def} obtain results for growth-optimal portfolios in the open market. Later, in Section~\ref{sec:hybrid_market}, we will use the hybrid Jacobi processes introduced in the previous section to model the market weights in the open market setting developed in this section.
Throughout this section we work on a filtered probability space $(\Omega,\Fcal,(\Fcal(t))_{t \ge 0}, \P)$ satisfying the usual conditions. 

\subsection{Stochastic Portfolio Theory}
We consider a financial market with $d \ge 2$ assets. The market capitalization (share price times total number of outstanding shares) of the $i$th asset is modelled by a nonnegative continuous semimartingale $S_i$. We assume that the stocks do not pay dividends. As is commonly done in SPT, we immediately switch numeraire to $S_1 + \cdots + S_d$, the total capitalization of the entire market. This quantity is assumed to stay strictly positive at all times, although we do allow individual asset values to reach zero. In this numeraire, the asset prices are the \emph{market weights},
\[
X_i = \frac{S_i}{S_1 + \cdots + S_d}, \quad i=1,\ldots,d.
\]
The market weight process $X = (X_1,\ldots,X_d)$ thus takes values in the standard simplex $\Delta^{d-1}$.
We assume $X$ is an It\^o process with dynamics
\begin{equation} \label{eq_market_weight_dynamics}
	dX(t) = b(t)dt + \sigma(t)dW(t)
\end{equation}
for some standard $d$-dimensional Brownian motion $W$ and progressively measurable processes $b$ and $\sigma$ taking values in $\R^d$ and $\R^{d \times d}$ respectively. We write
\[
c = \sigma \sigma^\top
\]
and note that $c \bm1_d = 0$ and $\bm1_d^\top b = 0$ since $X$ takes values in $\Delta^{d-1}$. In what follows, only the market weight process $X$ will appear in the analysis. The capitalization process $S$ will not.


We next discuss trading strategies, which can be parameterized in a number of ways. An $X$-integrable process $H = (H_1,\ldots,H_d)^\top$ representing the number of shares held of each asset yields the wealth process $V = H_1 X_1 + \ldots + H_d X_d$. Since no other assets are available for trading (in particular, there is no bank account), the self-financing condition states that $dV(t) = H(t)^\top dX(t)$. Note that wealth, like prices, is expressed in the numeraire $S_1 + \cdots + S_d$ and thus measures performance relative to the market.

We only allow strategies whose wealth stays strictly positive. This permits us to describe trading strategies in terms of $\theta_i = H_i / V$, which is the number of shares, per unit of wealth, that the investor holds in the $i$th asset. These quantities satisfy
\begin{equation} \label{eq_fully_invested}
	\theta_1 X_1 + \cdots + \theta_d X_d = 1.
\end{equation}
Conversely, any $X$-integrable process $\theta = (\theta_1,\ldots,\theta_d)$ that satisfies \eqref{eq_fully_invested} yields a strictly positive wealth process $V^\theta$ given by
\begin{equation} \label{eq_wealth_dynam_NEW}
	\frac{dV^\theta(t)}{V^\theta(t)} = \sum\limits_{i=1}^d \theta_i(t)dX_i(t), \quad V^\theta(0) = 1.
\end{equation}
This turns out to be the most convenient way to parameterize trading strategies for the purposes of this paper. The initial wealth is not relevant for our analysis, so we always take it to be one. If a strategy is given by $\theta(t) = \theta(X(t))$ for some function $\theta \colon \Delta^{d-1} \to \R^d$, we say that it is in \emph{feedback form}.

A particularly important strategy is the \emph{market portfolio} $\theta_\Mcal = \bm1_d$, which is the strategy that holds the entire market. It satisfies \eqref{eq_fully_invested} and its wealth process is $V^{\theta_\Mcal} = X_1 + \ldots + X_d = 1$. The market portfolio is a risk-free investment with respect to the current numeraire, and as such it plays a role similar to the bank account in classical models. It is the benchmark against which the performance of other strategies will be measured.

\begin{remark}
	Let $\theta$ be an $X$-integrable process that does not satisfy \eqref{eq_fully_invested}, and hence is not a self-financing trading strategy. Shift it by a multiple $1 - X^\top \theta$ of the market portfolio to obtain $\theta' = \theta + (1 - X^\top \theta) \bm1_d$. The process $\theta'$ is $X$-integrable and satisfies \eqref{eq_fully_invested}. Moreover, replacing $\theta$ by $\theta'$ does not affect the right-hand side of \eqref{eq_wealth_dynam_NEW} since $\bm1_d^\top dX(t) = 0$. In this way, \emph{any} $X$-integrable process $\theta$ yields a trading strategy.
\end{remark}

\begin{remark}
	In SPT, and elsewhere, it is common to describe trading using $\pi_i = \theta_i X_i$, the fraction of wealth invested in asset $i=1,\ldots,d$. This description is less convenient in our context since $X_i$ may reach zero, in which case $\theta_i$ cannot be recovered from $\pi_i$.
\end{remark}

\subsection{Growth Optimality}
We now study the problem of maximizing growth of wealth. Following \cite{karatzas2021portfolio}, we define growth optimality in terms of the semimartingale decomposition $\log V^\theta = A^\theta + M^\theta$ of the log-wealth process corresponding to a strategy $\theta$. We stress that wealth is measured relative to the market, and hence growth is too.

\begin{defn}
	Let $\Xi$ be a collection of strategies. A strategy $\hat \theta \in \Xi$ is said to be \emph{growth optimal in $\Xi$} if $A^{\hat \theta} - A^\theta$ is a non-decreasing process for every strategy $\theta \in \Xi$. 
\end{defn}

As discussed in the introduction, many models in SPT that capture the stability of the ranked market weights yield growth optimal strategies with undesirable properties. The following specific example of this phenomenon serves as a motivating example for the analysis to come.


\begin{eg}[A Motivating Closed Market Example] \label{sec:vol_stabilized_example}
	Fix a parameter $\gamma_* > 0$ and consider the Markov process given by
	\[ 
	dX_i(t) = \frac{\gamma_*d}{2}\(\frac{1}{d}- X_i(t)\)dt + \sum_{j = 1}^d (\delta_{ij} - X_i(t))\sqrt{X_j(t)} \, dW_j(t), \quad i=1,\ldots,d.
	\] 
	This is the market weight dynamics in the volatility-stabilized market introduced in \cite{fernholz2005relative}, which is a single-parameter model in SPT that is tractable since the market weight process $X$ is a polynomial diffusion; see \cite{fernholz2005relative,pickova2014generalized} for an in-depth discussion. We note that this is a particular case of the hybrid Jacobi model introduced in Section~\ref{sec:hybrid_polynomial_models}, obtained by setting $a =0$, $\sigma^2 = 1$ and $\gamma_i= \gamma_*$ for every $i$. In particular, Theorem~\ref{thm:ergodic} yields that $X$ is ergodic and its invariant measure is a symmetric Dirichlet distribution with parameter $\gamma_*$. As such, this market model satisfies property (1) outlined in the introduction.
	
	Next, it can be shown, and indeed follows as a special case of Theorem~\ref{thm:hybrid_growth_optimal} to come, that a growth optimal strategy $\hat\theta$ exists if and only if $\gamma_* \geq 1$. It is then given by
	\[
	\hat\theta_i(t) = \frac{1}{2}\( \frac{\gamma_*}{X_i(t)} + 2-d\gamma_* \), \quad i=1,\ldots,d,
	\]
	which is well-defined because the condition $\gamma_* \ge 1$ is necessary and sufficient for the market weights to remain strictly positive; see Theorem~\ref{thm:boundary_attainment}. This model falls under the framework of \cite{kardaras2021ergodic} and, from the results of that paper, $\hat \theta$ is not only growth optimal in the volatility-stabilized market, but possesses a \emph{robust} growth optimality property over a certain class of admissible models. Thus property (3) in the introduction is also satisfied.
	
	However, $\hat \theta$ is rather poorly behaved. In terms of the portfolio weights $\hat\pi_i = \hat\theta_i X_i$ we have
	\[
	\hat \pi_i(t) = \frac{1}{2}(\gamma_* + X_i(t)(2-d\gamma_*)).
	\]
	This strategy prescribes significant short-selling of stocks with large capitalizations. Indeed, whenever $X_i(t) > \gamma_* / (d\gamma_* - 2)$, the portfolio takes a short position in asset $i$. Thus, essentially independently of $\gamma_*$, one would have to short-sell approximately the 100 largest stocks in the S\&P~500 (under current capitalizations) to implement this strategy. Moreover, the leverage ratio can be substantial. As a concrete example suppose that $X_i(t) = 0.05$ and that $d=500$. Then $\hat \pi_i(t) = 0.05 - 12\gamma_*$. Since $\gamma_* \geq 1$, this strategy requires the investor to short at least $12$ times their wealth in some of the large capitalization stocks. Indeed, the stability of the market weights in this model ensures that the largest stocks will eventually lose value, which leads to this artificially leveraged growth-optimal strategy. Moreover, the strategy itself has an explicit and strong dependence on the dimension parameter $d$. Hence, properties (2) and (4) are not met.
	%
	%

\end{eg}

\subsection{Open Market Strategies}
A key goal of this paper is to argue that the undesirable behaviour in Example~\ref{sec:vol_stabilized_example} can be avoided by working with \emph{open markets} as in \cite{fernholz2018numeraire,karatzas2020open}. More precisely, we consider strategies that are only allowed to take positions in large capitalization stocks along with positions in the market portfolio. Allowing investments in the market portfolio, thereby allowing certain restricted positions in small capitalization stocks, is essential for tractability and distinguishes our setup from previous work \cite{fernholz2018numeraire,karatzas2020open}.

Recall the rank and name identifying functions ${\bf r}_i(x)$ and ${\bf n}_k(x)$ in \eqref{eq_rank_name_def}. We write
\begin{equation} \label{eq_rank_name_at_X}
	\begin{aligned}
		{\bf r}_i(t) &= \text{${\bf r}_i(X(t)) =$ the rank occupied by name $i$ at time $t$,} \\
		{\bf n}_k(t) &= \text{${\bf n}_k(X(t)) =$ the name that occupies rank $k$ at time $t$.}
	\end{aligned}
\end{equation}
For any trading strategy $\theta$ we then write
\[
\theta_{\bf n}(t) = (\theta_{{\bf n}_1(t)}(t),\dots,\theta_{{\bf n}_d(t)}(t))^\top, \quad t \ge 0.
\]
This is the permutation of $\theta(t)$ which orders the investments in the assets by their rank. Lastly, we introduce the truncated vector of ranked weights,
\[
X^N_{()} = (X_{(1)},\dots,X_{(N)}).
\]
We can now define the class of open market strategies.


\begin{defn} [Open Market Strategies]  \label{def:open_market} 	
	Fix $N \leq d$. A trading strategy $\theta$ is called an \textit{open market strategy} if it admits the representation
	\begin{equation} \label{eqn:pi_canonical}
		\theta_{\bf{n}}(t) = \begin{pmatrix}
			h(t) \\ 0
		\end{pmatrix} + (1 - h(t)^\top X^N_{()}(t))\boldsymbol{1}_d
	\end{equation}
	for some $N$-dimensional process $h$ such that the $d$-dimensional process $h_{{\bf{r}}_i(t)}(t) 1_{\{{\bf r}_i(t) \leq N \}}$, $i=1,\dots,d$, $t \ge 0$, is $X$-integrable. We denote by $\mathbb{O}^N$ the set of all such strategies. 
\end{defn}

Here $h_k(t)$ is the number of shares, per unit of wealth, that the investor \emph{directly} invests in the $k^{\text{th}}$ ranked asset at time $t$, for $k=1,\ldots,N$. To finance this strategy the investor buys or sells $1 - h(t)^\top X^N_{()}(t)$ number of shares, per unit of wealth, of the market portfolio. One should view strategies of this kind as investments in an \emph{open market} of size $N$ embedded in a larger market of size $d$. This is a market where at any time $t$ there are $N+1$ assets available for investment, namely
\begin{itemize}
	\item $X_{{\bf{n}}_k(t)}(t)$ for $k=1,\dots,N$ (the assets occupying the top $N$ ranks at time $t$); and
	\item the market portfolio.
\end{itemize}
With $N=d$ one recovers the standard setup (the \emph{closed market} or \emph{full market}) where investments in small capitalization stocks are unrestricted. This is also achieved with $N = d-1$, since the smallest asset can be created synthetically by investing in the top $d-1$ assets and the market portfolio. A proper open market is obtained with $N < d-1$.


Finally, for later use we define the \emph{open market portfolio} $\theta^N_{\mathcal{M}}$ for the open market of size $N$. It is defined by choosing $h(t) = (\sum_{k=1}^N X_{(k)}(t))^{-1}\boldsymbol{1}_N$ in the representation \eqref{eqn:pi_canonical}. Thus
\[(\theta^N_{\mathcal{M}})_{{\bf{n}}_k(t)}(t)= \frac{1_{\{k \leq N\}}}{\sum_{k=1}^N X_{(k)}(t)}, \quad k=1,\dots,d.\]

\begin{remark}
	The integrability condition on $h$ in Definition~\ref{def:open_market} is necessary and sufficient for $\theta$ to be a trading strategy. Thus the set of such processes $h$ fully parameterizes the set of open market strategies.
\end{remark}

\subsection{Growth Optimality in the Open Market}
We now consider an open market of size $N < d$, and study the growth optimal strategy in the class of open market strategies under two structural conditions on the market. This allows us to relate the growth optimal strategy in the open market to the growth optimal strategy in the closed market---although, as we will see, the latter will not necessarily satisfy the integrability conditions required to be a bona fide trading strategy.

Our first condition is classical. Apart from the absence of integrability requirements, this condition is sometimes known simply as the \emph{structure condition}. Recall the processes $b$ and $\sigma$ appearing in the market weight dynamics \eqref{eq_market_weight_dynamics} as well as $c = \sigma \sigma^\top$.

\begin{assum} \label{ass:closed_viable}
	There exists a $d$-dimensional progressively measurable process $\ell$ such that $c \ell = b$ up to $\P \otimes dt$-nullsets.
\end{assum}

\begin{remark}
	By \cite[Proposition~2.4]{karatzas2021portfolio}, Assumption~\ref{ass:closed_viable} together with the condition
	\begin{equation} \label{eqn:finite_growth}
		\int_0^T \ell(t)^\top c(t) \ell(t)dt < \infty \text{ for every $T \geq 0$}
	\end{equation} is equivalent to the existence of a growth optimal strategy $\hat \theta^{d}$ in the closed market, given by
	\begin{equation} \label{eqn:pi_ell_opt}
		\hat \theta^d(t) = \ell(t) +(1 - \ell(t)^\top X(t))\boldsymbol{1}_d.
	\end{equation}
	In this paper we do not impose \eqref{eqn:finite_growth} and hence do not require the existence of a growth optimal strategy in the closed market. However, as examples in Section~\ref{sec:hybrid_market} will show, a growth optimal strategy in the open market may exist even if \eqref{eqn:finite_growth} is violated.
\end{remark}

Now, observe that the log-wealth process for any strategy $\theta$ is given by
\begin{equation} \label{eqn:logV_arb_pi}
	\log V^{\theta}(T) = \int_0^T \(\theta(t)^\top c(t)\ell(t) - \frac{1}{2}\theta(t)^\top c(t) \theta(t)\)dt + \int_0^T \theta(t)^\top \sigma(t)dW(t).
\end{equation} 
To find the growth optimal strategy in the open market (if it exists) we need to maximize the integrand appearing in the drift term of \eqref{eqn:logV_arb_pi} over all $\theta \in \mathbb{O}^N$, the class of open market strategies defined in Definition~\ref{def:open_market}. To do this we introduce some additional notation. Define
\begin{align*}
	\kappa_{kl}(t) & = c_{{\bf{n}}_k(t){\bf{n}}_l(t)}(t), \quad \varrho_k(t)  = \ell_{{\bf{n}}_k(t)}(t), \quad k,l=1,\dots,d, \quad t\geq 0.
\end{align*}
In particular, $\kappa_{kl}$ is the instantaneous covariation between the assets occupying the $k^{\text{th}}$ and $l^{\text{th}}$ ranks. 
To maximize the drift in \eqref{eqn:logV_arb_pi}, consider any open market strategy $\theta \in \mathbb{O}^N$ in the canonical form of \eqref{eqn:pi_canonical}. Using that $c(t) \boldsymbol{1}_d = 0$, or equivalently $\kappa(t) \bm1_d = 0$, we have that the drift is equal to
\[
h(t)^\top (\kappa(t)\varrho(t))^N - \frac{1}{2} h(t)^\top \kappa^N(t) h(t),
\]
where $\kappa^N(t)$ is the upper left $N\times N$ block of $\kappa(t)$ and $(\kappa(t)\varrho(t))^N$ is the $N$-dimensional truncated vector $((\kappa(t)\varrho(t))_1,\dots,(\kappa(t)\varrho(t))_N)$. Maximizing this expression over $h(t)$ is an unconstrained concave quadratic maximization problem. The solution is characterized by the first order condition, which is
\begin{equation} \label{eqn:FOC}
	(\kappa(t)\varrho(t))^N = \kappa^N(t) h(t).
\end{equation}
Our second structural condition is designed to allow for an explicit solution of \eqref{eqn:FOC} in terms of $\varrho$.

\begin{assum} \label{ass:kappa_condition}
	There exist progressively measurable processes $f_1,\ldots,f_d$ and $g$ such that, up to $\P \otimes dt$-nullsets, $\sum_{k=N+1}^d f_k = 1$ and
	\begin{equation} \label{eqn:kappa_condition}
		\kappa_{kl} = - f_k f_l g \quad \text{for all $k=1,\dots,N$ and $l=N+1,\dots,d$.}
	\end{equation}
\end{assum}

This condition requires that the instantaneous covariation between large capitalization assets and small capitalization assets be of product form. We stress that this does not restrict the instantaneous covariations \emph{within} the collection of top $N$ assets, nor within the collection of bottom $d-N$ assets. The condition only involves cross-interaction between the two groups of assets. Moreover, the requirement $\sum_{k=N+1}^d f_k = 1$ is a normalization that can always be achieved as long as $\sum_{k=N+1}^d f_k > 0$. We also note that the covariation matrix for the market weight process of Example~\ref{sec:vol_stabilized_example} satisfies this condition and, as we will see in Section~\ref{sec:hybrid_market}, when $X$ is taken to be the hybrid Jacobi process of Section~\ref{sec:hybrid_polynomial_models}, Assumption~\ref{ass:kappa_condition} holds as well.

\begin{thm}[Growth optimality in the open market]  \label{thm:main_theoretical}
	Fix $N < d$ and let Assumptions~\ref{ass:closed_viable} and~\ref{ass:kappa_condition} be satisfied.
	\begin{enumerate} 
		\item \label{item:hat_h} A solution to \eqref{eqn:FOC} exists and is given by 
		\begin{equation} \label{eqn:hat_h}
			\hat h(t) = \varrho^N(t) - \zeta(t)\boldsymbol{1}_N,
		\end{equation}
		where $\varrho^N(t) = (\varrho_1(t),\dots,\varrho_N(t))$ and $\zeta(t) = \sum_{l=N+1}^d f_l(t)\varrho_l(t)$.
		\item \label{item:growth_optimal}There exists a growth optimal strategy $\hat \theta$ in the open market if and only if 
		\begin{equation} \label{eqn:finite_growth_open}
			\int_0^T \hat h(t)^\top \kappa^N(t) \hat h(t)dt < \infty, \text{ $\P$-a.s., for every $T \geq 0$},
		\end{equation}
		where $\hat h$ is given by \eqref{eqn:hat_h}.
		In this case $\hat \theta$ is determined by $\hat h$ via \eqref{eqn:pi_canonical}. Moreover, if $\kappa^N$ is invertible up to $\P  \otimes dt$-nullsets then the growth optimal strategy is unique up to $\P \otimes dt$-nullsets. 
	\end{enumerate} 
\end{thm}
\begin{proof} This theorem can be proved by applying the results of \cite{karatzas2020open}. However, for the reader's benefit we present a short self-contained proof. Throughout, we argue up to $\P \otimes dt$-nullsets. Fix $k \in \{1,\ldots,N\}$. Since $\kappa \bm1_d = 0$ we have from Assumption~\ref{ass:kappa_condition} that 
	\[
	(\kappa^N \boldsymbol{1}_N)_k = \sum_{l=1}^N \kappa_{kl} = - \sum_{l=N+1}^d \kappa_{kl} = \sum_{l=N+1}^d f_k f_l g = f_k g.
	\]
	It follows that 
	\[(\kappa\varrho)_k = \sum_{l=1}^N \kappa_{kl}\varrho_l + \sum_{l=N+1}^d \kappa_{kl} \varrho_l  = (\kappa^N \varrho^N)_k - f_kg\zeta = (\kappa^N(\varrho^N - \zeta\boldsymbol{1}_N))_k. \]
	This shows that $\hat h$ given by \eqref{eqn:hat_h} solves \eqref{eqn:FOC}, thus proving \ref{item:hat_h}. 
	
	We now prove \ref{item:growth_optimal}. Suppose that \eqref{eqn:finite_growth_open} holds. Then it is easily verified that $\hat \theta$ is $X$-integrable. Indeed, up to a factor of $1/2$, both the finite variation and quadratic variation parts of the log-wealth process are given precisely by \eqref{eqn:finite_growth_open}. Hence $\hat \theta$ is a trading strategy, and by the preceding discussion it is growth optimal.
	
	Conversely, suppose that \eqref{eqn:finite_growth_open} fails. Then the stopping times
	\[
	\tau_M := \inf\left\{T \geq 0: \int_0^T \hat h(t)^\top \kappa^N(t)\hat h(t)dt \geq M\right\}, \quad M \in \N,
	\]
	satisfy $\sup_M \tau_M < \infty$ with positive probability. For each $M$, consider the strategy $\hat \theta^{M}$ characterized by
	\[ 	\hat \theta^M_{\bf n}(t) = \(\begin{pmatrix}
		\hat h(t) \\ 0
	\end{pmatrix} + (1 - \hat h(t)^\top X^N_{()}(t))\boldsymbol{1}_d\)1_{\{t \leq \tau_M\}} + \boldsymbol{1}_d 1_{\{t > \tau_M\}}.\]
	This strategy imitates $\hat \theta$ up to time $\tau_M$ and invests in the market portfolio thereafter. It satisfies the required integrability condition and so is indeed a trading strategy in the open market. Since the finite variation part of the log-wealth of $\hat \theta^{M}$ will be at least $M$ after $\tau_M$, and since $\sup_{M} \tau_M < \infty$ with positive probability, no growth optimal strategy can exist.
	
	Finally, if a growth optimal strategy exists and $\kappa^N$ is invertible, then \eqref{eqn:hat_h} is the unique solution to \eqref{eqn:FOC} and, consequently, $\hat \theta$ is the unique growth optimal strategy.
\end{proof}

Assuming its conditions are satisfied, Theorem~\ref{thm:main_theoretical} shows that the growth optimization problem in the open market is just as easy, or difficult, as the growth optimization problem in the closed market. This is because the optimizer $\hat h$ is an explicit function of the process $\ell$ in Assumption~\ref{ass:closed_viable} and the given processes $f_k$, $k=N+1,\ldots,d$, in Assumption~\ref{ass:kappa_condition}.

Furthermore, the growth optimal strategy $\hat \theta$ admits an intuitive decomposition. At any time $t \geq 0$ this strategy prescribes the investor to hold, per unit of wealth,
\begin{itemize}
	\item  $\varrho_k(t)$ number of shares in the asset occupying the $k^{\text{th}}$ rank, for $k=1,\dots,N$,
	\item $-\zeta(t) \boldsymbol{1}_N^\top X_{()}^N(t)$ number of shares in the open market portfolio $\theta^N_{\mathcal{M}}$,
	\item  $1 - \varrho^N(t)^\top X^N_{()}(t) + \zeta(t) \boldsymbol{1}_N^\top X_{()}^N(t)$ number of shares in the full market portfolio $\theta_{\mathcal{M}}$.
\end{itemize} 
For an investor with unrestricted access to all $d$ assets, investing $\varrho_k(t)$ in the asset occupying the $k^{\text{th}}$ rank is growth optimal for all $k=1,\dots,d$, and this strategy is financed by investing the remaining wealth (perhaps negative) in the market portfolio. Indeed, this is an equivalent way of describing \eqref{eqn:pi_ell_opt}. Under the open market constraint, the investor is no longer allowed to execute this strategy and so must reallocate their wealth. The above decomposition shows that this amounts to reallocating their investment from the small capitalization stocks into the open market portfolio and financing the difference using the full market portfolio.

\section{Modelling the Market Weights} \label{sec:hybrid_market}

In this section we use the hybrid Jacobi process as a model for the market weights, and study growth optimality in the open market of size $N \le d-1$. We fix parameters $a,\gamma,\sigma$ as in Definition~\ref{def:hybrid_Jacobi} satisfying Assumption~\ref{ass:bar_a}, and we let $X$ be the hybrid Jacobi process of Theorem~\ref{thm:process_existence}. 

\subsection{Growth Optimization Problem}
Theorem~\ref{thm:main_theoretical} characterizes growth optimality in the open market. We now make use of this result for the hybrid Jacobi process. In particular, Theorem~\ref{thm:boundary_attainment} allows us to characterize the parameter values for which \eqref{eqn:finite_growth_open} holds. This leads to a complete description of growth optimality in the open market for hybrid Jacobi processes. Recall the tail sum notation \eqref{eq_tail_sum_notation}.

\begin{thm}[Growth optimal strategy in hybrid Jacobi markets] \label{thm:hybrid_growth_optimal}
	There exists a growth optimal strategy in the open market if and only if $\bar a_k + \bar \gamma_{(k)} \geq 1$ for $k=2,\dots,N+1$. In this case the growth optimal strategy $\hat \theta$ is characterized by
	\begin{equation} \label{eqn:hat_pi_hybrid}
		\hat \theta_{{\bf{n}}_k(t)}(t) = \begin{cases}
			1 - \frac{1}{2}(\bar a_1 + \bar \gamma_1) + \dfrac{a_k + \gamma_{{\bf{n}}_k(t)}}{2X_{(k)}(t)}, & k=1,\dots,N, \\[3ex]
			1 - \frac{1}{2}(\bar a_1 + \bar \gamma_1) + \dfrac{\bar a_{N+1} + \sum_{l=N+1}^d\gamma_{{\bf{n}}_l(t)}}{2\bar X_{(N+1)}(t)}, & k = N+1,\dots,d,
		\end{cases}
	\end{equation}
	and it is the unique growth optimal strategy up to $\P_\mu \otimes dt$-nullsets for every $\mu \in \mathcal{P}_0$.
\end{thm}
Before proving the theorem we make a few remarks.
\begin{remark}
	Note that the optimal strategy $\hat \theta$ does hold a position in the small capitalization stocks through the market portfolio. Indeed, equation~\eqref{eqn:hat_pi_hybrid} prescribes the same number of shares to be held when $k \in \{N+1,\dots,d\}$. The equation for $\hat \theta$  can instead be written in the canonical form \eqref{eqn:pi_canonical}, which represents the position via its contributions to each investible asset, by taking
	\begin{equation} \label{eqn:hat_h_hybrid}
		\hat h_k(t) =  \frac{a_k + \gamma_{{\bf n}_k(t)}}{2X_{(k)}(t)} - \frac{\bar a_{N+1} + \sum_{l=N+1}^d\gamma_{{\bf n}_l(t)}}{2\bar X_{(N+1)}(t)}, \quad k = 1,\dots,N.
	\end{equation}
\end{remark}
\begin{remark} 
	It may seem odd that \eqref{eqn:hat_pi_hybrid} does not involve the volatility parameter $\sigma$. This is only due to the parametrization used in Definition~\ref{def:hybrid_Jacobi} where $\sigma$ appears both in the drift and the volatility. In the financial literature it is more common to view the drift and volatility parameters as independently defined. We can accomplish this by setting $\alpha = \sigma^2 a / 2$ and $\beta = \sigma^2 \gamma / 2$. Then \eqref{eqn:X_polynomial_dynamics} becomes
	\[
	\begin{aligned}
		dX_i(t) &= \(\alpha_{{\bf r}_i(t)} + \beta_i - (\bar \alpha_1 + \bar \beta_1) X_i(t)\)dt + \sigma \sum_{j=1}^d \left(\delta_{ij} - X_i(t) \right) \sqrt{ X_j(t) } \, dW_j(t)
	\end{aligned}
	\]
	for $i=1,\ldots,d$ and, correspondingly, the optimal strategy \eqref{eqn:hat_pi_hybrid} becomes
	\[\hat \theta_{{\bf{n}}_k(t)}(t) = \begin{cases}
		1 - \frac{1}{\sigma^2}(\bar \alpha_1 + \bar \beta_1) + \dfrac{\alpha_k + \beta_{{\bf{n}}_k(t)}}{\sigma^2X_{(k)}(t)}, & k=1,\dots,N, \\[3ex]
		1 - \frac{1}{\sigma^2}(\bar \alpha_1 + \bar \beta_1) + \dfrac{\bar \alpha_{N+1} + \sum_{l=N+1}^d\beta_{{\bf{n}}_l(t)}}{\sigma^2\bar X_{(N+1)}(t)}, & k = N+1,\dots,d.
	\end{cases}\]
	This strategy depends on $(\alpha_k + \beta_{{\bf{n}}_k(t)})/\sigma^2$ for $k=1,\dots,N$, which admits an interpretation of risk-adjusted return for the asset that is occupying the $k^{\text{th}}$ rank at time $t$.
\end{remark}
\begin{remark}
	The condition $\bar a_k + \bar \gamma_{(k)} \geq 1$ for $k=2,\dots,N+1$ appearing in the statement of Theorem~\ref{thm:hybrid_growth_optimal} has a clear financial interpretation.  Note that using the dynamics of the ranked weights obtained in Proposition~3.12 we have that the dynamics of $\bar X_{(k)} = \sum_{l=k}^d X_{(l)}$ are given by
	\[d\bar X_{(k)} = \frac{\sigma^2}{2}(\bar a_k + \sum_{l=k}^d \gamma_{{\bf n}_l(t)} - (\bar a_1 + \bar \gamma_1)\bar X_{(k)}(t))\, dt + dM(t) - \frac{1}{4}dL_{k-1,k}(t),\]
	where $M$ is the martingale term. Consequently, the term $\bar a_k + \sum_{l=k}^d \gamma_{{\bf n}_l(t)}$ appearing in the drift can be interpreted as a cumulative growth parameter for the $d-k+1$ smallest capitalization stocks. The term depends on which names occupy the smallest ranks and the smallest value this growth parameter can have is $\bar a_k + \bar \gamma_{(k)}$ precisely when the stocks with the smallest individual (named) growth parameters occupy the smallest ranks. The condition $\bar a_k + \bar \gamma_{(k)} \geq 1$ then requires that even in this ``worst case" arrangement of the assets, the cumulative growth parameter of the $d-k+1$ smallest stocks is sufficiently large. 
\end{remark}
\begin{remark}
	It is worth pointing out that, akin to \cite{cuchiero2019polynomial}, the hybrid Jacobi market weight process is consistent with the following specification of the capitalization process:
	\[
	dS_i(t) = \frac{\sigma^2}{2} (\gamma_i + a_{{\bf r}_i(t)})\Sigma(t)dt + \sigma\sqrt{S_i(t)\Sigma(t)}dW_i(t)
	\]
	for $i=1,\dots,d$, where $\Sigma = S_1 + \cdots + S_d$. Then $X_i = S_i / \Sigma$ has dynamics given by \eqref{eqn:X_polynomial_dynamics}. 
\end{remark}

\begin{proof}[Proof of Theorem~\ref{thm:hybrid_growth_optimal}]
	We first verify the conditions of Theorem~\ref{thm:main_theoretical}. To check Assumption~\ref{ass:closed_viable} we define $\ell \colon \Delta^{d-1} \to \R^d$ by
	\begin{equation} \label{eqn:ell_def}
		\ell_i(x) = \frac{\gamma_i + a_{{\bf r}_i(x)}}{2x_i},  \quad i=1,\dots,d,
	\end{equation}
	for $x \in \Delta^{d-1}_+$ and arbitrarily set $\ell(x) = 0$ for $x \in \partial \Delta^{d-1}$. We then have the identity
	\[
	\(c(x)\ell(x)\)_i = \frac{\sigma^2}{2}\(\gamma_i + a_{{\bf r}_i(x)} - (\bar a_1 + \bar \gamma_1)x_i\)
	\]
	for all $x \in \Delta^{d-1}_+$ and $i=1,\dots,d$, where $c(x)$ is given by \eqref{eqn:c_def}. Proposition~\ref{prop:Lebesgue_collision}\ref{item:boundary_zero} shows that $X$ almost surely spends a Lebesgue nullset of time at the boundary $\partial \Delta^{d-1}$. We deduce that Assumption~\ref{ass:closed_viable} is satisfied with the process $\ell(X)$. 
	
	Next, Assumption~\ref{ass:kappa_condition} follows from the form \eqref{eqn:kappa_def} of the diffusion matrix $\kappa(X_{()})$ of the ranked market weights. Indeed, condition \eqref{eqn:kappa_condition} is satisfied with 
	\[
	g(t) = \sigma^2 \bar X_{(N+1)}(t)^2, \qquad  f_k(t) = \frac{X_{(k)}(t)}{\bar X_{(N+1)}(t)}, \quad  k =1,\dots,d,
	\]
	and Proposition~\ref{prop:Lebesgue_collision}\ref{item:boundary_zero} ensures that $f_k$ is well-defined up to $\P_\mu \otimes dt$-nullsets.
	
	We may now apply Theorem~\ref{thm:main_theoretical}\ref{item:hat_h} to see that a solution to \eqref{eqn:FOC} is given by $\hat h$ defined in \eqref{eqn:hat_h_hybrid}. By Theorem~\ref{thm:main_theoretical}\ref{item:growth_optimal} this gives the unique growth optimal strategy \eqref{eqn:hat_pi_hybrid} if and only if \eqref{eqn:finite_growth_open} holds. It remains to argue that the latter condition is equivalent to $\bar a_k + \bar \gamma_{(k)} \geq 1$ for $k=2,\dots,N+1$.
	
	A calculation using \eqref{eqn:hat_h_hybrid} shows that the left-hand side of \eqref{eqn:finite_growth_open} is given by
	\begin{equation} \label{eqn:local_growth_jacobi}
		\begin{aligned}
			\hat h(t)^\top \kappa^N(t)\hat h(t) &= \frac{\sigma^2}{4}\bigg(\sum_{k=1}^N \frac{(a_k + \gamma_{{\bf{n}}_k(t)})^2}{X_{(k)}(t)} \\
			&\qquad + \frac{(\bar a_{N+1} + \sum_{k=N+1}^d \gamma_{{\bf{n}}_k(t)})^2}{\bar X_{(N+1)}(t)} - (\bar a_1 + \bar \gamma_1)^2\bigg).
		\end{aligned}
	\end{equation}
	Assume first that $\bar a_l + \bar \gamma_{(l)} \geq 1$ for $l=2,\dots,N+1$. Then by Theorem~\ref{thm:boundary_attainment}\ref{item:bound_attain_rank}\ref{item:bound_attain_rank_hit} we have that $\P_\mu(X_{(N+1)}(t) = 0 \text{ for some } t > 0) = 0$. Consequently, from \eqref{eqn:local_growth_jacobi},	we see that \eqref{eqn:finite_growth_open} holds.
	
	Conversely, assume there exists some $l \in \{2,\dots,N+1\}$ such that $\bar a_l + \bar \gamma_{(l)} < 1$. Set $I = \{i \in \{1,\dots,d\}: {\bf{n}}_k(\gamma) = i \text{ for some } k \geq l\}$ and let $\Lambda_I$ be as in \eqref{eqn:Lambda_def}. Note that $|I| = d-l+1$ and by Theorem~\ref{thm:boundary_attainment}\ref{item:bound_attain_name}\ref{item:bound_attain_name_hit} we have that $\Lambda_I(X)$ hits zero with positive probability. 
	For $n \in \N$ define the stopping times
	\begin{align*}
		\tau_n & := \inf\{t\geq0: \Lambda_I(X(t)) \leq 1/n\}, \\
		\tau & := \lim_{n \to \infty} \tau_n   = \inf\{t \geq 0: \Lambda_I(X(t)) = 0\}.
	\end{align*} Applying Itô's formula to \eqref{eqn:Lambda_dynamics} we obtain for every $n$ and $T \geq 0$ that
	\begin{equation} \label{eqn:log_Lambda_growth_proof}
		\begin{aligned}
			-\log \frac{\Lambda_I(X(T\land \tau_n))}{\Lambda_I(0)}  & = \frac{\sigma^2(\bar a_1 + \bar \gamma_1)(T \land \tau_n)}{2} - \sum_{i \in I}\int_0^{T \land \tau_n} \frac{ \gamma_i + a_{{\bf r}_i(t)}}{\Lambda_I(X(t))}dt \\
			& \quad - M(T \land \tau_n) + \frac{1}{2}[M,M](T \land \tau_n),
		\end{aligned}
	\end{equation} 
	where $M(T) = \sigma\int_0^T \sqrt{\frac{1-\Lambda_I(X(t))}{\Lambda_I(X(t))}}dB(t)$. Suppose for contradiction that $\int_0^T (\Lambda_I(X(t)))^{-1}dt$ is finite for every $T \geq 0$. This implies that $[M,M](T \land \tau) < \infty$. Thus, when sending $n \to \infty$ in \eqref{eqn:log_Lambda_growth_proof}, we obtain on the right  hand side a finite limit for every $T \geq 0$. However, on the left hand side we obtain $-\log (\Lambda_I(X(T \land \tau)) / \Lambda_I(0))$ which is infinite with positive probability for $T$ large enough, since $\{\tau < \infty\}$ has positive probability. This yields a contradiction and establishes that $\int_0^T (\Lambda_I(X(t)))^{-1}dt = \infty$ with positive probability for sufficiently large $T$. But $\bar X_{(N+1)} \leq \Lambda_I(X)$ so it follows that $\int_0^T \bar X_{(N+1)}^{-1}(t)dt = \infty$ with positive probability. Consequently, we see from \eqref{eqn:local_growth_jacobi} that \eqref{eqn:finite_growth_open} does not hold in this case. This completes the proof.
\end{proof}

\subsection{Examples}\label{sec:examples}
We consider a few examples.
\subsubsection{Name-Based Jacobi Markets}\label{sec:name_based_example} Consider the case $a = 0$. Then Assumption~\ref{ass:bar_a} is satisfied if and only if $\gamma_i > 0$ for $i=1,\dots,d$. Theorem~\ref{thm:hybrid_growth_optimal} guarantees that there exists a growth optimal strategy $\hat \theta$ in the open market if and only if $\bar \gamma_{(d-N+1)} \geq 1$. Using \eqref{eqn:hat_pi_hybrid} it is easily checked that $\hat \theta$ is long-only if and only if we additionally have $\sum_{k=1}^{d-1} \gamma_{(k)} \leq 2$.

Consider the further special case $\gamma = \gamma_*\boldsymbol{1}_d$ for some $\gamma_* > 0$. In  this case the dynamics of $X$ in \eqref{eqn:X_polynomial_dynamics} reduce to the volatility stabilized market in Example~\ref{sec:vol_stabilized_example}. The growth optimal strategy in the open market, $\hat \theta$, exists if and only if $ \gamma_* \geq 1/(d-N)$.  In this case \eqref{eqn:hat_pi_hybrid} becomes
\[ 	\hat \theta_{{\bf{n}}_k(t)}(t) = \begin{cases}
	1 - \frac{d\gamma_*}{2} + \frac{\gamma_*}{2X_{(k)}(t)}, & k=1,\dots,N, \\
	1 - \frac{d\gamma_*}{2} + \frac{(d-N)\gamma_*}{2\bar X_{(N+1)}(t)}, & k = N+1,\dots,d.
\end{cases}\]
This is long-only if and only if $\gamma_* \in [\frac{1}{d-N}, \frac{2}{d-1}]$. When $N \leq (d+1)/2$ this interval is nonempty.  As such, the open market in the volatility-stabilized setup helps resolve the high-leverage problem that appears in Example~\ref{sec:vol_stabilized_example}. However, the issue of stability with respect to the size of the market is not resolved here, since $\hat \theta$ depends strongly on $d$, the size of the closed market.
\subsubsection{Rank-Based Jacobi Markets} \label{sec:rank_based}
Now consider the purely rank-based case $ \gamma = 0$. Then Assumption~\ref{ass:bar_a} is satisfied if and only if $\bar a_k > 0$ for $k=2,\dots,d$. By Theorem~\ref{thm:hybrid_growth_optimal} there exists a growth optimal strategy in the open market of size $N$ if and only if $\bar a_{k} \geq 1$ for $k = 2,\dots,N+1$. In this case \eqref{eqn:hat_pi_hybrid} becomes
\[	\hat \theta_{{\bf{n}}_k(t)}(t)= \begin{cases}
	1 - \frac{\bar a_1}{2} + \frac{a_k}{2X_{(k)}(t)}, & k=1,\dots,N, \\
	1 - \frac{\bar a_1}{2} + \frac{\bar a_{N+1}}{2\bar X_{(N+1)}(t)}, & k = N+1,\dots,d.
\end{cases}\]
It is straightforward to find admissible parameter values such that this strategy is long-only. For example, this will be the case whenever $\bar a_1 \le 2$ and $a_k \ge 0$ for all $k=1,\ldots,N$.


The above strategy has the attractive property that it does not directly depend on the parameter $d$. Indeed, there is only an implicit dependence on $d$ through $\bar a_{N+1} = \sum_{k=N+1}^d a_k$ and $\bar X_{(N+1)} = \sum_{k = N+1}^d X_{(k)}$. Both of these quantities admit interpretations in terms of cumulative statistics of the small capitalization assets: $\bar X_{(N+1)}$ is the cumulative size of the small capitalization stocks, while $\bar a_{N+1}$ is the cumulative growth parameter for the small capitalizations stocks. As such, rank-based Jacobi markets exhibit property (4) discussed in the introduction.

\subsubsection{Jacobi Atlas Models}
Next we explore a particular rank-based model, which, although simplistic, is illustrative of the effect that the open market has in comparison to the classically studied closed market. Fix $\eta \geq 1$ and choose $a \in \R^d$ satisfying Assumption~\ref{ass:bar_a} such that $a_1 = a_2 = \dots = a_N = 0$ and $\bar a_{N+1} = \eta$. We call this specification the \emph{Jacobi Atlas model} as all the growth comes from the small capitalization assets through the parameter $\eta$. This is analogous to the Atlas model of \cite{banner2005atlas}. In this case the growth optimal strategy is
\begin{equation} \label{eqn:hat_pi_Atlas}
	\hat \theta_{{\bf{n}}_k(t)}(t) = \begin{cases}
		1 - \frac{\eta}{2} , & k=1,\dots,N, \\
		1 - \frac{ \eta}{2} + \frac{\eta}{2\bar X_{(N+1)}(t)}, & k = N+1,\dots,d.
	\end{cases}
\end{equation}
The investor holds the same number of units of each large capitalization asset and holds a larger number of shares (through the market portfolio) in the small capitalization assets due to their growth potential guaranteed by the parameter $\eta \geq 1$.

The parameter $\eta$ also directly determines how leveraged the strategy is. A larger value of $\eta$ leads to a higher degree of short-selling. Since $\eta$ represents the cumulative growth of the small capitalization stocks the interpretation is clear. For large values of $\eta$ the investor benefits by holding a larger position in the small capitalization stocks (through the market portfolio) and financing this strategy by short selling the large capitalization stocks. Conversely, when $\eta$ is small, the small capitalization stocks do not provide enough growth to justify the risk of holding a highly leveraged position. 

Note that the optimal strategy is insensitive to the exact specification of the parameters $a_{N+1},\dots,a_d$ as long as $\bar a_{N+1} = \eta$ and Assumption~\ref{ass:bar_a} is satisfied. These choices, however, do affect the behaviour of the larger market in which the open market is embedded. To illustrate, consider the following two situations:
\begin{enumerate} [label = ({\roman*}),noitemsep] 
	\item $a = \eta e_d$,
	\item $a = (\eta -\epsilon)e_{N+1} + \epsilon e_d$ for some choice of $\epsilon \in (0, 1)$.
\end{enumerate}
Both satisfy $\bar a_{N+1} = \eta$ and Assumption~\ref{ass:bar_a}. However, under (i), Theorem~\ref{thm:hybrid_growth_optimal} implies that there exists a growth optimal strategy in the open market of any size $M \in \{1,\dots,d-1\}$. In particular there exists a growth optimal strategy in the closed market. By contrast, the specification (ii) does not admit a growth optimal strategy in the open market of any size $M > N$. This is because under (ii), $X_{(N+2)}$ will hit zero with positive probability, while under (i) none of the assets will hit zero. The open market of size $N$, however, is impervious to these differences.
\subsubsection{The Hybrid Case Revisited} 
We saw in the name-based case, Example~\ref{sec:name_based_example}, that although there are specifications for which the growth optimal strategy is not artificially leveraged, there is no stability with respect to the parameter $d$ due to the name-based dependency. The rank-based case of Example~\ref{sec:rank_based} solves this issue, but it has other imperfections. Under the rank-based setup each asset asymptotically spends the same amount of time occupying each rank. Indeed, the ergodic property (Theorem~\ref{thm:ergodic}) yields $\lim_{T \to \infty}\frac{1}{T}\int_0^T 1_{\{{\bf{n}}_k(t) = i\}}dt = 1/d$ for every $i,k =1,\dots,d$. This is an unrealistic property that does not hold in real world equity markets.

The general hybrid specification is more flexible as one can allow a particular named asset to occupy some ranks more frequently than others. As an example, let $\eta \geq 1$ and $\gamma_* > 0$ be given and set $a = \eta e_d$ and $\gamma = \gamma_*e_1$.
In this case the first named asset will occupy the highest ranks more frequently than the lowest ranks. Indeed, using the ergodic property one can show that
	\[
	\lim_{T \to \infty} \frac{1}{T}\int_0^T 1_{\{{\bf r}_1(t) = k\}}dt > \lim_{T \to \infty} \frac{1}{T}\int_0^T 1_{\{{\bf r}_1(t) = l\}}dt, \quad k < l.
	\]

	\subsection{Functional Generation of the Growth Optimal Strategy} \label{sec:func_gen}
	In this section we show that the growth optimal strategy in the rank-based case is functionally generated in the sense of \cite[Chapters 3 \& 4]{fernholz2002stochastic}. 
	We adopt the setting of Section~\ref{sec:open_market_SPT}, where $X$ serves as the market weight process. Following \cite{fernholz2002stochastic,karatzas2017trading} we introduce the notion of functionally generated strategies and the master formula.
	
	\begin{defn}
		Let $\theta$ be a strategy and let $G:\R^d \to (0,\infty)$ be a function that is continuous on $\Delta^{d-1}_+$ and such that $G(X)$ is a semimartingale. If we have the representation 
		\begin{equation}
			\label{eqn:func_gen_wealth}\log V^\theta(T) = \log G(X(T)) + \Gamma(T)
		\end{equation}
		for some finite variation process $\Gamma$ with $\Gamma(0) = 0$ then we call $\theta$  a \emph{functionally generated strategy} with generating function $G$ and drift process $\Gamma$. In this case we write $\theta^G$ for $\theta$.
	\end{defn}
	
	The representation \eqref{eqn:func_gen_wealth} shows that the wealth process of a functionally generated strategy can be found without computing a stochastic integral.  This will allow us to show that functionally generated strategies have a built-in robustness property, which is crucial in the analysis of the robust growth optimization problem we consider in Section~\ref{sec:robust}. 
	
	\begin{thm}[Master Formula, Proposition~4.7 in \cite{karatzas2017trading}] \label{thm:master_formula}
		Let $G:\R^d \to (0,\infty)$ be a function that is continuous on $\Delta^{d-1}_+$ and such that $G(X)$ is a semimartingale. Assume there exist locally bounded measurable functions $g_i: \R^d \to \R$ for $i = 1,\dots,d$ and a finite variation process $Q$ such that 
		\begin{equation}\label{eqn:func_gen_dynamics}
			d\log G(X(t)) = \sum\limits_{i=1}^d g_i(X(t))dX_i(t) + dQ(t).
		\end{equation}
		Then $G$ functionally generates the strategy $\theta^G$ with components given by
		\begin{equation}\label{eqn:pi_func_gen}
			\theta^G_i(t) = g_i(X(t)) + 1 - \sum\limits_{j=1}^d X_j(t)g_j(X(t)), \quad i = 1,\dots,d.
		\end{equation}
	\end{thm}

	\begin{remark} \label{rem:func_gen} If the function $G$ from the previous theorem is $C^2$ on an open neighbourhood of $\Delta^{d-1}_+$ then the assumptions of the theorem are satisfied with $g_i(x) = \partial_i \log G(x)$ and $dQ(t) =\frac{1}{2}\sum_{i,j=1}^d \partial_{ij}\log G(X(t))d[X_i,X_j](t)$. The corresponding strategy becomes
		\[
		\theta^G_i(t) = \partial_i \log G(X(t)) + 1 - \sum\limits_{j=1}^d X_j(t)\partial_j \log G(X(t)), \quad i = 1,\dots,d
		\]
		and we have the representation \[\log V^{\theta^G}(T) = \log G(X(T)) -\frac{1}{2}\sum_{i,j=1}^d\int_0^T\frac{ \partial_{ij} G(X(t))}{G(X(t))}d[X_i,X_j](t).\]
	\end{remark}

	We now establish that the growth optimal strategy in the open market for rank Jacobi models is functionally generated and determine its generating function.

	\begin{prop} \label{prop:func_gen_open}
		Let $X$ be the rank Jacobi model of Section~\ref{sec:hybrid_market} with parameter $a \in \R^d$ and $\sigma^2 > 0$. Assume that $\bar a_k \geq 1$ for every $k=2,\dots,N+1$ so that the growth optimal strategy $\hat \theta$ of Theorem~\ref{thm:hybrid_growth_optimal} exists. Define $\hat F: \nabla^{d-1}_+ \to (0,\infty)$ and $\hat G: \Delta^{d-1}_+ \to (0,\infty)$ via
		\begin{equation} \label{eqn:hat_G}
			\hat F(y) = \bar y_{N+1}^{\bar a_{N+1}/2}\prod_{k=1}^N y_k^{a_k/2}, \qquad \hat G(x) = \hat F(x_{()}). 
		\end{equation}  Then $\hat \theta$ given by \eqref{eqn:hat_pi_hybrid} is functionally generated by $\hat G$. Moreover, its wealth process has the representation
		\begin{equation} \label{eqn:hat_pi_log_wealth}
			\begin{aligned}
				\log V^{\hat \theta} (T) &=  \log \hat F(X_{()}(T)) - \int_0^T \frac{L_\kappa \hat F}{\hat F}(X_{()}(t))dt - \frac{1}{8}\sum_{k=1}^{N-1} (a_k-a_{k+1})\tilde L_{k,k+1}(T) \\
				&\quad - \frac{1}{8}\int_0^T\(a_N - \bar a_{N+1} \frac{X_{(N)}(t)}{\bar X_{(N+1)}(t)}\)d \tilde L_{N,N+1}(t),
			\end{aligned}
		\end{equation}
		%
		where $L_\kappa$ is the operator given by
		\begin{equation} \label{eqn:L_kappa}
			L_\kappa f(y) = \frac{1}{2}\sum_{k,l=1}^d \kappa_{kl}(y)\partial_{kl}f(y), \quad f \in C^2(\nabla^{d-1}),
		\end{equation}
		with $\kappa$ given by \eqref{eqn:kappa_def} and $\tilde L_{k,k+1} = L_{\log X_{(k)} - \log X_{(k+1)}}$ are the local times of gaps of the ranked log-weights.
	\end{prop}
	
	\begin{proof}
		This result is a consequence of the semimartingale decomposition for $X_{()}$ given by \eqref{eqn:X_()_dynamics} along with the master formula, Theorem~\ref{thm:master_formula}.
	\end{proof}

	\begin{remark} \label{rem:discont_G}
		In the general hybrid case there is a candidate generating function for $\hat \theta$ given by
		\[
		G(x) = \bar x_{(N+1)}^{\frac{1}{2}(\bar a_{N+1} + \sum_{k=N+1}^d \gamma_{{\bf{n}}_k(x)})}\prod_{k=1}^N x_{(k)}^{\frac{1}{2}(a_k + \gamma_{{\bf{n}}_k(x)})}.
		\]
		This function is differentiable almost everywhere and we have $\hat \theta(t) = \nabla \log G(X(t)) - (1 + X(t)^\top \nabla \log G(X(t)))\boldsymbol{1}_d$ whenever $X(t)$ is at a point of differentiability of $G$. However, when $N < d-1$ and $\gamma$ is not a multiple of $\boldsymbol{1}_d$, $G$ is discontinuous on a nonempty subset of $\{x_{(N)} = x_{(N+1)}\}$. In particular, we do not expect $G(X)$ to be a semimartingale. We do not know whether there is any meaningful way to interpret $\hat \theta$ as being generated by $G$. In the case $N = d-1$ the formula for $G$ simplifies to 
		\begin{equation} \label{eqn:hat_G_closed}
			G(x) = \(\prod_{k=1}^d x_{(k)}^{a_k/2}\)\(\prod_{i=1}^d x_i^{\gamma_i/2}\),
		\end{equation}
		which is a continuous function. In this special case it can be shown that $\hat \theta$ is functionally generated by $G$.
	\end{remark}
	
	%
		
		\section{Robust Asymptotic Growth} \label{sec:robust}
		The rank and hybrid Jacobi processes are parametric models which should only be viewed as idealizations of real-world market weight dynamics. Moreover, even if one accepts the structural form of the model, parameter estimation is not straightforward. In particular, drift processes are notoriously difficult to estimate for the low signal-to-noise ratios typically seen in financial data; see \cite{kardaras2012robust,kardaras2021ergodic}.
		
		We now study an \emph{asymptotic} and \emph{robust} growth optimization problem in the open market, with only two inputs fixed a priori: (i) the covariation matrix $\kappa$ of the ranked market weights, and (ii) their invariant density $q$. Following \cite{kardaras2021ergodic,itkin2020robust} we restrict to models with the prescribed covariation structure and invariant density, but allow for drift uncertainty. We maximize asymptotic growth in this class, and show that the optimal solution is the growth optimal strategy in the rank Jacobi model.
		
		\begin{remark}
			The results of this section only cover the rank Jacobi model ($\gamma = 0$). The difficulty in passing to the hybrid case arises from the fact that $\hat \theta$ in \eqref{eqn:hat_pi_hybrid} is not functionally generated in that case, as discussed in Remark~\ref{rem:discont_G}. 
		\end{remark}

		We now introduce the robust asymptotic growth optimization problem, motivated by the ergodic robust maximization problems studied in \cite{kardaras2021ergodic,itkin2020robust}. We work on the canonical path space $\Omega = C([0,\infty);\Delta^{d-1})$ with the topology of locally uniform convergence and Borel $\sigma$-algebra $\F$. The coordinate process is denoted by $X$ and we write $(\Fcal(t))_{t\ge0}$ for the right-continuous filtration generated by $X$.

		Fix $N \leq d-1$, $\sigma \in (0,\infty)$, and $a \in \R^d$, and assume that
		\begin{equation} \label{eqn:finite_growth_ass}
			\bar a_k > 1 \text{ for every } k=2,\dots,N+1.
		\end{equation}
		Throughout this section we let $\hat \P$ denote the law of the rank Jacobi process of Theorem~\ref{thm:process_existence} with parameters $a, \sigma$ and arbitrary initial law $\mu \in \mathcal{P}_0$ (the choice of initial law does not matter in the analysis to come.) Under $\hat \P$, Theorem~\ref{thm:hybrid_growth_optimal} yields a growth optimal strategy in the open market. Note that \eqref{eqn:finite_growth_ass} is stronger than the condition of Theorem~\ref{thm:hybrid_growth_optimal} in that it contains a strict inequality. This ensures that $\hat \theta$ has finite asymptotic growth rate as defined in \eqref{eqn:growth_rate_def} below.

		Recalling the matrix function $\kappa \colon \nabla^{d-1} \to \mathbb{S}_+^d$ in \eqref{eqn:kappa_def} and the density $q \colon \nabla^{d-1}_+ \to (0,\infty)$ in \eqref{eqn:q_def}, we now define the set of admissible models that will be used to define the robust optimization problem.
		
		\begin{defn}\label{def:Pi_geq}
			Let $\Pi_{\geq}$ denote the set of all probability measures $\P$ on $(\Omega,\F)$ such that the following hold: $X_0 \sim \mu$ for some $\mu \in \mathcal{P}_0$, $X$ is a continuous semimartingale with canonical decomposition of the form $dX(t) = b^{\P}(t)dt + dM^\P(t)$,  and we have
			\begin{enumerate}[label = (\roman*)]
				\item  \label{item:non_explosion} $X_{(N+1)}(t) > 0$ for $t \geq 0$, $\P$-a.s.
				\item  \label{item:QV_rank_condition} $\int_0^T d[X_{(k)},X_{(l)}](t) = \int_0^T \kappa_{kl}(X_{()}(t))dt$ for $k,l=1,\dots,d$ and $T \ge 0$, $\P$-a.s.
				\item \label{item:ergodic_rank}  For every measurable function $h \colon \nabla^{d-1} \to \R$ with $\int_{\nabla^{d-1}} |h|q < \infty$, 
				\[
				\lim_{T \to \infty} \frac{1}{T}\int_0^T h(X_{()}(t))dt = \int_{\nabla^{d-1}} hq, \quad \text{$\P$-a.s.}
				\]
				\item \label{item:drift_condition_rank} We have
				\[
				\limsup\limits_{T \to \infty}\frac{1}{T}\int_{0}^T |b^\P_i(t)|^{r'}dt< \infty, \quad i=1,\dots,d, \quad \P\text{-a.s.}
				\]
				for some $r' > s'$ where $s'$ is the conjugate exponent of $s = \min\{\bar a_2,\dots,\bar a_{N+1}\}$, meaning that $1/s + 1/s'=1$.
			\end{enumerate}
		\end{defn}  
		Condition~\ref{item:non_explosion} guarantees that $X_{(N+1)}$ does not hit zero so that the candidate robust growth-optimal strategy $\hat \theta$ is actually a trading strategy (i.e.\ $\hat \theta$ is $X$-integrable). Conditions~\ref{item:QV_rank_condition} and \ref{item:ergodic_rank} encode the covariation structure and ergodic property of $X_{()}$ respectively. Condition~\ref{item:drift_condition_rank} is a technical condition on the asymptotic growth of the drift process of $X$. Our proof of the main result in this section, Theorem~\ref{thm:robust} below, relies on this condition; however, it is unclear to us if it can be relaxed or not. Also, note that the constant $s$ appearing in \ref{item:drift_condition_rank} is strictly greater than one by condition \eqref{eqn:finite_growth_ass}.
		\begin{remark}
			Note that $\hat \P \in \Pi_{\geq}$. Thus the law of the rank Jacobi process is admissible. This follows from \eqref{eqn:X_()_dynamics}, Corollary~\ref{cor:rank_ergodic} and the fact that drift of $X$ in \eqref{eqn:X_polynomial_dynamics} is bounded.
		\end{remark}

		The \emph{asymptotic growth rate} of a wealth process $V^\theta$ and an admissible law $\P \in \Pi_{\geq}$ is defined by
		\begin{equation} \label{eqn:growth_rate_def}
			g(V^\theta;\P) = \sup\left\{\eta \in \R \colon \liminf_{T \to \infty} \frac{1}{T} \log V^\theta(T)\geq \eta, \text{ $\P$-a.s.} \right\}.
		\end{equation}
		We aim to compute the \emph{robust optimal growth rate} in the open market of size $N$, 
		\begin{equation} \label{eqn:lambda}
			\hat \lambda := \sup_{\theta \in \mathbb{O}^N}\inf_{\P \in \Pi_{\geq}} g(V^\theta;\P).
		\end{equation}
		We can now state the main result of this section.
		
		\begin{thm}[Robust growth optimality]\label{thm:robust}
			The robust optimal growth rate $\hat \lambda$ is achieved by $\hat \theta$ given by \eqref{eqn:hat_pi_hybrid} and we have 
			\begin{equation} \label{eqn:robust_growth} 
				\hat\lambda = \frac{\sigma^2}{8}\int_{\nabla^{d-1}} \(\sum_{k=1}^N\frac{a_k^2}{y_k} + \frac{\bar a_{N+1}^2}{\bar y_{N+1}}\)q(y)dy - \frac{\sigma^2}{8}\bar a_1^2.
			\end{equation}
			Moreover, $g(V^{ \hat \theta}; {\P}) = \hat\lambda$ for every $\P \in \Pi_{\geq}$.
		\end{thm}

		To prove this result we establish that the expression on the right hand side of \eqref{eqn:robust_growth} is both a lower and upper bound for $\hat\lambda$. By definition of $\hat \lambda$ we have the bounds
		\begin{equation} \label{eqn:lambda_bound}
			\inf_{\P \in \Pi_{\geq}} g(V^{\hat \theta};\P) \leq \hat\lambda \leq \sup_{\theta \in \mathbb{O}^N} g(V^{\theta};\hat \P).
		\end{equation}
		From the growth optimality of $\hat \theta$ under $\hat \P$, established in Theorem~\ref{thm:hybrid_growth_optimal}, the upper bound is equal to $g(V^{\hat \theta};\hat \P)$. Moreover, Lemma~\ref{lem:hybrid_model_asymptotic_growth} below establishes that $g(V^{\hat \theta};\hat \P)$ matches the right hand side of \eqref{eqn:robust_growth}. Then, in view of \eqref{eqn:lambda_bound}, to complete the proof it suffices to show that $\hat \theta$ achieves the same asymptotic growth rate under every measure $\P \in \Pi_{\geq}$. Lemma~\ref{lem:func_gen_growth_ivnariance} below establishes that strategies functionally generated by permutation invariant smooth functions bounded away from zero indeed have this growth rate invariance property. By approximating the generating function of $\hat \theta$ by such functions we are able to show that $\hat \theta$ also possesses this property. This last step is technical and is carried out in Appendix~\ref{app:robust}.

		We start with a technical lemma, whose proof is also located in Appendix~\ref{app:robust}.
		\begin{lem} \label{lem:technical_robust} Fix $\P \in \Pi_{\geq}$.  
			\begin{enumerate}[label = ({\roman*})]
				\item \label{item:collision} For $k=1,\dots,N$ the set $\{t: X_{(k)}(t) = X_{(k+1)}(t)\}$ is $\P$-a.s.\ a Lebesgue nullset.  
				\item \label{item:QV} We have the identity $[X_i,X_j](T) = \int_0^T\sum_{k,l=1}^d 1_{\{{\bf{n}}_k(t)=i,{\bf{n}}_l(t)=j\}}\kappa_{kl}(X_{()}(t))dt$ for every $i,j=1,\dots,d$ and every $T \geq 0$.
			\end{enumerate}
		\end{lem}
		
		Next we characterize the asymptotic growth rates of strategies under the law $\hat \P$ of the rank Jacobi process.
		
		\begin{lem}\label{lem:hybrid_model_asymptotic_growth} Consider a strategy $\theta(t) = \theta(X(t))$ in feedback form satisfying the integrability condition 
			$\int_{\Delta^{d-1}}|\theta^\top c\ell| p< \infty$
			where $c$, $p$ and $\ell$ are given by \eqref{eqn:c_def}, \eqref{eqn:p_def} and \eqref{eqn:ell_def} respectively. Then
			\begin{equation} \label{eqn:worst_case_ergodic}
				g(V^{\theta};\hat \P) = \int_{\Delta^{d-1}}\( \theta^\top c\ell - \frac{1}{2}\theta^\top c\theta\)p.
			\end{equation}
			In particular 
			\begin{equation} \label{eqn:worst_case_growth}
				g(V^{\hat \theta};\hat \P) = \frac{1}{8}\int_{\nabla^{d-1}} \(\sum_{k=1}^N\frac{a_k^2}{y_k} + \frac{\bar a_{N+1}^2}{\bar y_{N+1}}\)q(y)dy - \frac{1}{8}\bar a_1^2,
			\end{equation}
			where $\hat \theta$ is given by \eqref{eqn:hat_pi_hybrid}.
		\end{lem}
		\begin{proof}
			We work under $\hat \P$.
			Using \eqref{eqn:X_polynomial_dynamics}, the wealth process for any strategy $\theta$ is given by
			\begin{align}\log V^{\theta}(T)= \int_0^T \theta(X(t))^\top c(X(t))\ell(X(t))dt + M(T) - \frac{1}{2} [M,M](T), \label{eqn:ergodic_lemma1}
			\end{align}
			where $dM(t) = \theta(X(t))^\top \sigma(X(t))dW(t)$ for $\sigma_{ij}(x) = \sigma \sqrt{x_j} (\delta_{ij} - x_i)$. By the integrability assumption and the ergodic property, Theorem~\ref{thm:ergodic}, we have
			\begin{equation} \label{eqn:ergodic_lemma2}
				\lim_{T \to \infty} \frac{1}{T} \int_0^T \theta(X(t))^\top c(X(t))\ell(X(t))dt = \int_{\Delta^{d-1}}\theta^\top c\ell p.
			\end{equation}
			Also by the ergodic property, $\lim_{T \to \infty} T^{-1}[M,M](T) = \int_{\Delta^{d-1}} \theta^\top c\theta p$, where the right-hand side may be infinite. By the Dambis--Dubins--Schwarz theorem there exists a Brownian motion $B$ (possibly on an extended probability space) such that $M(T) = B([M,M](T))$. The strong law of large numbers for Brownian motion, along with the ergodic property, yield
			\begin{equation} \label{eqn:ergodic_lemma3}
				\begin{aligned}
					\lim_{T \to \infty} \frac{M(T) - \frac{1}{2}[M,M](T)}{T} &= \lim_{T \to \infty} \frac{[M,M](T)}{T}\(\frac{B([M,M](T))}{[M,M](T)} - \frac{1}{2}\) \\
					&= - \frac{1}{2}\int_{\Delta^{d-1}}\theta^\top c\theta p.
				\end{aligned}
			\end{equation}
			Combining \eqref{eqn:ergodic_lemma1}, \eqref{eqn:ergodic_lemma2} and \eqref{eqn:ergodic_lemma3} gives \eqref{eqn:worst_case_ergodic}. Moreover, a direct calculation shows that  
			\[  \hat \theta(x)^\top c(x)\ell(x) - \frac{1}{2}\hat \theta(x)^\top c(x)\hat \theta(x) = \frac{1}{8}\(\sum_{k=1}^N\frac{a_k^2}{x_{(k)}} + \frac{\bar a_{N+1}^2}{\bar x_{(N+1)}} - \bar a_1^2\),\]
			which yields \eqref{eqn:worst_case_growth}. The integral in \eqref{eqn:worst_case_growth} is finite by virtue of \eqref{eqn:finite_growth_ass} and Lemma~\ref{lem:Q_finite}.
		\end{proof}
		
		The next step is to establish the growth rate invariance property for functionally generated strategies with sufficiently regular generating functions.
		
		\begin{lem}[Growth rate invariance] \label{lem:func_gen_growth_ivnariance}
			Let  $F\in C^2(\nabla^{d-1};(0,\infty))$ be such that $\log F$ is bounded. Set $G(x) = F(x_{()})$ for $x \in \Delta^{d-1}$ and assume that $G \in C^2(\Delta^{d-1};(0,\infty))$. Then
			\begin{equation} \label{eqn:func_gen_growth_invariance} g(V^{\theta^G};\P) = \int_{\nabla^{d-1}} \frac{-L_\kappa F}{F}q = \int_{\nabla^{d-1}}\( \nabla \log F^\top \kappa \varrho - \frac{1}{2}\nabla \log F^\top \kappa \nabla\log F\)q
			\end{equation}
			for every $\P \in \Pi_{\geq}$, where $\varrho(y) = \ell(y)$ for $y \in \nabla^{d-1}_+$ and $L_\kappa$ is given by \eqref{eqn:L_kappa}.
		\end{lem}
		\begin{proof}
			By Remark~\ref{rem:func_gen} we have the representation 
			\[\log V^{\theta^G}(T) = \log G(X(T)) - \frac{1}{2}\int_0^T \sum_{i,j=1}^d \frac{\partial_{ij}G}{G} (X(t)) d[X_i,X_j](t).\]
			Moreover, we have $\partial_i G(x) = \partial_k F(x_{()})$ for all $x \in \Delta^{d-1}$ and all $i,k$ such that ${\bf{n}}_k(x) = i.$
			Hence, using Lemma~\ref{lem:technical_robust}\ref{item:QV}, we see that 
			\begin{equation} \label{eqn:smooth_permute}
				\begin{aligned}
					\frac{1}{2}\sum_{i,j=1}^d & \frac{\partial_{ij}G}{G} (X(t)) d[X_i,X_j](t) \\
					& = \frac{1}{2}\sum_{i,j=1}^d \frac{\partial_{ij}G}{G} (X(t)) \sum_{k,l=1}^d 1_{\{{\bf{n}}_k(t)= i, {\bf{n}}_l(t) = j\}}\kappa_{kl}(X_{()}(t))dt \\
					& = \frac{1}{2} \sum_{k,l=1}^d \frac{\partial_{kl} F}{F} (X_{()}(t)) \kappa_{kl}(X_{()}(t))dt = \frac{L_\kappa F }{F} (X_{()}(t)) dt.
				\end{aligned} 
			\end{equation}
			It follows that the wealth process can be written as
			\[\log V^{\theta^G}(T) = \log F(X_{()}(T)) + \int_0^T \frac{-L_\kappa F}{F} (X_{()}(t)) dt.\]
			Since $F \in C^2(\nabla^{d-1};(0,\infty))$ and $\log F$ is bounded, $-L_\kappa F/F$ is bounded. Hence, dividing by $T$, sending $T \to \infty$ and using the ergodic property yields
			\[
			\lim_{T \to \infty} \frac{1}{T} \log V^{\theta^G}(T) = \int_{\nabla^{d-1}} \frac{-L_\kappa F}{F}q, \quad \P\text{-a.s.}
			\]
			for every $\P \in \Pi_{\geq}$. This proves the first equality in \eqref{eqn:func_gen_growth_invariance}. The second equality in \eqref{eqn:func_gen_growth_invariance} follows by comparing this expression with the expression for $g(V^{\theta^G};\hat \P)$ given by Lemma~\ref{lem:hybrid_model_asymptotic_growth} and using the permutation invariance of $G$. Note that $\int_{\Delta^{d-1}} \nabla \log G^\top c \ell p< \infty$ since $\nabla \log G$ is bounded, so that Lemma~\ref{lem:hybrid_model_asymptotic_growth} is indeed applicable.
		\end{proof}
		
		\begin{remark}
			The second equality in \eqref{eqn:func_gen_growth_invariance} can also be proved by integration by parts.
		\end{remark}
		
		In view of Lemma~\ref{lem:func_gen_growth_ivnariance} and $\hat \theta$ being functionally generated by virtue of Proposition~\ref{prop:func_gen_open}, we expect that the asymptotic growth rate of $\hat \theta$ is the same under each $\P \in \Pi_{\geq}$. This is indeed true, but the difficulty in proving this is that $\hat G$ in \eqref{eqn:hat_G} is not $C^2$. As such, the drift process $\Gamma$ in \eqref{eqn:func_gen_wealth} will contain local time terms as in \eqref{eqn:hat_pi_log_wealth}, whose ergodic averages are difficult to analyze for arbitrary $\P \in \Pi_{\geq}$. The solution is to mollify $\hat G$ to obtain smooth approximations $\hat G_n$ for which Lemma~\ref{lem:func_gen_growth_ivnariance} applies. We then approximate, uniformly in $T \in [0,\infty)$, the wealth process of $\hat \theta$ by the wealth processes of $\theta^{\hat G_n}$ induced by the mollified functions. This uniform approximation crucially uses Definition~\ref{def:Pi_geq}\ref{item:drift_condition_rank}. As a result we obtain the growth rate invariance property for $\hat \theta$ and complete the proof of Theorem~\ref{thm:robust}. The details are located in Appendix~\ref{app:robust}.
		
		\section{Conclusion} \label{sec:conclusion}
		In this paper we introduced hybrid Jacobi processes in Section~\ref{sec:hybrid_polynomial_models} and studied in detail their ergodic, collision and boundary attainment properties. In Section~\ref{sec:open_market_SPT} we shifted focus to study a relaxed version of an open market in a general non-parametric framework. We proposed a structural condition \eqref{eqn:kappa_condition} on the covariation between small and large-cap stocks under which the growth-optimal strategy in the open market becomes comparable to the growth-optimal strategy in the closed market it is embedded in. Then, in Section~\ref{sec:hybrid_market}, we combined the insights from the two previous sections to study an open market setup under which the market weight process is a hybrid Jacobi process. We observed that the hybrid Jacobi markets satisfy the desirable properties (1) and (2) laid out in the introduction. Additionally, Section~\ref{sec:robust} and Section~\ref{sec:rank_based} showed that the subclass of rank Jacobi markets satisfied properties (3) and (4) respectively.
		
		Although we do not study the empirical calibration of hybrid Jacobi models in this paper we would like to point out a few features of the model, which make it amenable to calibration. The explicit expression for the invariant density \eqref{eqn:p_def} may allow for \emph{moment matching}. That is to compute certain theoretical moments in terms of the parameters $a,\gamma$ and $\sigma$ and match these to the capital distribution curves obtained from data. Additionally, the ergodic expressions for the collision local times in Proposition~\ref{prop:ergodic_local} allow one to calibrate parameters to the \emph{turnover} in the stock ranks. Indeed, the methods of \cite{fernholz2013second} are applicable to the setting of this paper (note that when the volatility parameters in \cite{fernholz2013second} are all equal to each other then the invariant density of the market weights in \cite{fernholz2013second} is the same as in this paper with $\bar a_1 + \bar \gamma_1 = 0$). 
		
		One important limitation of the hybrid Jacobi models we introduce here is the appearance of only one single volatility parameter $\sigma^2$. Taking into consideration the condition \eqref{eqn:kappa_condition} a natural extension for the covariation matrix in the hybrid Jacobi models would be 
		\begin{equation} \label{eqn:extended_c}
			c_{ij}(x) = -\sigma^2_{kl}x_ix_j \quad \text{ for } i \ne j, \qquad  \text{ where } k = r_i(x), \ l = r_j(x)
		\end{equation} and setting $c_{ii}(x) = -\sum_{j \ne i} c_{ij}$ for rank-based volatility parameters $\sigma_{kl}$. However, unless all of the volatility parameters are the same (which brings up back precisely to the volatility structure in the hybrid Jacobi model), the covariation matrix \eqref{eqn:extended_c} becomes discontinuous. This raises many challenges including existence of the corresponding market weight process, its ergodic properties, etc. We leave this important extension to future research.
		
		
		\begin{appendix}

			\section{Integration over $\Delta^{d-1}$ and $\nabla^{d-1}$} \label{app:integral}
			In this section we discuss in detail the conventions regarding integration over the simplex, the ordered simplex and its extensions as defined below. We also establish some useful identities and a technical lemma, which will be used in the next sections.
			
			We first extend the definition of the ordered simplex to a larger collection of sets. Given $\alpha,\beta \in \R$ set  $\nabla^{d-1}(\alpha,\beta) = \{y \in \R^{d}: y_1 \geq y_2 \geq \dots \geq y_{d} \geq \beta, \ \sum_{k=1}^{d} y_k = \alpha \}$ and note that $\nabla^{d-1}(1,0) = \nabla^{d-1}$. All integrals over $\nabla^{d-1}(\alpha,\beta)$ are understood with respect to the pushforward of the Lebesgue measure on $\R^{d-1}$ under the transformation $T_\alpha:\R^{d-1} \to \R^d$ given by $T_\alpha(x_1,\dots,x_{d-1}) = (x_1,\dots,x_{d-1},\alpha-\sum_{k=1}^{d-1}x_k)$.
			That is for a measurable function $f:\nabla^{d-1}(\alpha,\beta) \to \R$ we have that 
			\[\int_{\nabla^{d-1}(\alpha,\beta)} f(y)dy = \int_{T_\alpha^{-1}(\nabla^{d-1}(\alpha,\beta))} f(x_1,\dots,x_{d-1},\alpha-\sum_{k=1}^{d-1}x_k)dx_1,\dots,dx_{d-1}.\]  Integrals over $\Delta^{d-1}$ are defined analogously using the map $T_1$.
			
			Now we establish some useful identities involving integrals of functions of product form. To this end let $b \in \R^d$ be given and for $\alpha,\beta > 0$ define
			\[Q_b(\alpha,\beta) := \int_{\nabla^{d-1}(\alpha,\beta)} \prod_{k=1}^d y_k^{b_k-1}dy.\]
			Note that $Q_b(\alpha,\beta)$ is finite for every $b \in \R^d$ since $\beta > 0$. Additionally, by the definition of $\nabla^{d-1}(\alpha,\beta)$ and the integral change of variables formula, we have the homogeneity property 
			\begin{equation} \label{eqn:homogeneity_property}
				Q_b(\lambda\alpha,\lambda\beta) = \lambda^{\bar b_1-1}Q_b(\alpha,\beta)
			\end{equation} for every $\lambda > 0$. Recall that $\bar b_k = \sum_{l=k}^d b_l$ for every $k=1,\dots,d$. We now establish a useful lemma providing a recursive formula for $Q_b(\cdot,\cdot)$ as well as a limiting formula.
			\begin{lem} \label{lem:Q_lim}
				Let $b \in \R^d$ be given. 
				\begin{enumerate}
					\item \label{item:Q_b_iter}For any $\alpha > 0$ and $0 < \beta \leq \alpha/d$ we have the identity
					\begin{equation} \label{eqn:Q_b_iterative}
						Q_b(\alpha,\beta) = \int_{\beta}^{\alpha/d}y_d^{b_d-1}Q_{b'}(\alpha-y_d,y_d)dy_d
					\end{equation}
					where $b' = (b_1,\dots,b_{d-1}) \in \R^{d-1}$.
					\item \label{item:Q_b_lim}  $\lim_{\epsilon \downarrow 0} \epsilon^\eta Q_b(1-c\epsilon,\epsilon) = 0$ for any $c \in \R$ and $\eta > \max \{0,-\bar b_2,\dots,-\bar b_d\}$.
				\end{enumerate}
			\end{lem}
			\begin{proof} By our convention regarding integrals over $\nabla^{d-1}(\alpha,\beta)$ we have that 
				\begin{equation} \label{eqn:Q_b_expanded}
					Q_b(\alpha,\beta) = \int_{T_\alpha^{-1}(\nabla^{d-1}(\alpha,\beta))} (\alpha-\sum_{k=1}^{d-1}x_k)^{b_d-1}\prod_{k=1}^{d-1}x_k^{b_k-1}dx_1\dots dx_{d-1},
				\end{equation}
				where $T_\alpha^{-1}(\nabla^{d-1}(\alpha,\beta)) = \{x \in \R^{d-1}: x_1 \geq \dots \geq x_{d-1} \geq \alpha-\sum_{k=1}^{d-1}x_k \geq \beta\}$. Now we make the change of variables $y_k = x_{k}$ for $k=1,\dots,d-2$ and $y_d = \alpha-\sum_{k=1}^{d-1}x_k$. Note that $y_{d-1}$ is not defined, while $y_d$ is; this parametrization is convenient as will become apparent in the following equations. We then obtain
				\[Q_b(\alpha,\beta) = \int_{E} y_d^{b_d-1}(\alpha -y_d - \sum_{k=1}^{d-2}y_k)^{b_{d-1}-1}\prod_{k=1}^{d-2}y_k^{b_k-1}dy_1\dots dy_{d-2}dy_d,\]
				where $E = \{y_1 \geq \dots \geq y_{d-2} \geq \alpha -y_d -\sum_{k=1}^{d-2}y_k \geq y_d \geq \beta\} \subset\R^{d-1}$. Using Fubini we integrate over the last component $y_d$ to obtain
				\[Q_b(\alpha,\beta)\int_{\beta}^{\alpha/d}y_d^{b_d-1}\(\int_{\tilde E} (\alpha -y_d - \sum_{k=1}^{d-2}y_k)^{b_{d-1}-1}\prod_{k=1}^{d-2}y_k^{b_k-1}dy_1\dots dy_{d-2}\)dy_d,\]
				where $\tilde E = \{y_1 \geq \dots \geq  y_{d-2} \geq \alpha - y_d - \sum_{k=1}^{d-2}y_k \geq y_d\} \subset \R^{d-2}$. We now recognize from \eqref{eqn:Q_b_expanded} that the inner integral is precisely $Q_{b'}(\alpha-y_d,y_d)$, which proves \eqref{eqn:Q_b_iterative}.
				
				Now we prove \ref{item:Q_b_lim}. Fix $\eta, c$ as in the statement of the lemma and in what follows we always take $\epsilon$ small enough so that $1-c\epsilon > 0$. We prove the claim by induction on $d$. First take $d=2$ and note that $Q_{b}(1,\frac{\epsilon}{1-c\epsilon}) = \int_{\frac{\epsilon}{1-c\epsilon}}^{1/2} y^{b_2-1}(1-y)^{b_1-1}dy.$ If $b_2 > 0$ then we have that $Q_{b}(1,0) < \infty$ so we can directly send $\epsilon \downarrow 0$ to obtain the result. If not then a calculation involving L'Hopital's rule yields  
				\[\lim_{\epsilon \downarrow 0} \frac{Q_b(1,\frac{\epsilon}{1-c\epsilon})}{\epsilon^{-\eta}} = \lim_{\epsilon\downarrow 0} \frac{\epsilon^{\eta+b_2}}{\eta}\frac{(1-(c+1)\epsilon)^{b_1-1}}{(1-c\epsilon)^{b_1+b_2}} = 0,\]
				which proves the base case. 
				
				Next we assume the claim has been proven for $d-1$ and aim to prove it for $d$. If the vector $b$  satisfies $\bar b_k > 0$ for every $k \geq 2$ then we see by Lemma~\ref{lem:Q_finite} that $Q_{b}(1,0) < \infty$ so that directly sending $\epsilon \downarrow 0$ proves the claim. If $b$ does not satisfy this condition then we need to apply L'Hopital's rule to evaluate the limit. Using \eqref{eqn:Q_b_iterative} and the homogeneity property \eqref{eqn:homogeneity_property} we obtain that
				\[\frac{\partial}{\partial \beta}Q_b(1,\beta) = -\beta^{b_d-1}Q_{b'}(1-\beta,\beta) = -\beta^{b_d-1}(1-\beta)^{\bar b'_1-1}Q_{b'}\(1,\frac{\beta}{1-\beta}\),\]
				where $\bar b'_1 = \sum_{k=1}^{d-1}b_k$. Now by applying L'Hopital's rule, the chain rule and simplifying we obtain
				\[\lim_{\epsilon \downarrow 0} \frac{Q_b(1,\frac{\epsilon}{1-c\epsilon})}{\epsilon^{-\eta}} = \lim_{\epsilon\downarrow 0} \frac{\epsilon^{\eta+b_d}}{\eta}Q_{b'}\(1,\frac{\epsilon}{1-(c+1)\epsilon}\)\frac{(1-(c+1)\epsilon)^{\bar b_{1}'-1}}{(1-c\epsilon)^{\bar b_1}} = 0,\]
				where the last equality followed from the inductive hypothesis since $\eta + b_d > \max\{0,$ $-\bar b'_2,\dots,-\bar b'_{d-1}\}$.
			\end{proof}
			We close out this section by noting a final integrability property. Let $a \in \R^d$ be given satisfying Assumption~\ref{ass:bar_a} (with $\gamma = 0$) and fix $k \in \{2,\dots,d\}$ and $p > 0$. Note that the function $x \mapsto |\log (Cx)|^p x^\delta$ is bounded on $(0,1)$ for any $C,\delta > 0$. Hence by choosing $\delta > 0$ small enough so that $\tilde a := a -\delta e_k$ also satisfies Assumption~\ref{ass:bar_a} we see from Lemma~\ref{lem:Q_finite} that 
			\begin{equation} \label{eqn:log_bound}
				\int_{\nabla^{d-1}} |\log \bar y_k|^p \prod_{k=1}^{d}y_l^{a_l-1}dy \leq \int_{\nabla^{d-1}}|\log ((d-k+1)y_k)|^p y_k^\delta \prod_{k=1}^{d}y_l^{\tilde a_l-1}dy < \infty.
			\end{equation} 
			\section{Proofs of Results from Section~\ref{sec:hybrid_polynomial_models}} \label{app:hybrid_polynomial_models}
			In this section we prove some of the results from Section~\ref{sec:hybrid_polynomial_models} that were not included in the main body of the text. 
			We first set our sights on proving Lemma~\ref{lem:finite_Z}. To show this we first establish a separate lemma, which is also used elsewhere, and then obtain Lemma~\ref{lem:finite_Z} as a consequence.
			\begin{lem} \label{lem:Q_finite}
				For a vector $b \in \R^d$ we have that \[Q_b := \int_{\nabla^{d-1}} \prod_{k=1}^d y_k^{b_k-1}dy < \infty \quad  \iff \quad \bar b_k > 0 \quad \text{for} \quad k=2,\dots,d.\]
			\end{lem}
			\begin{proof} Note that 
				\[\prod_{k=1}^d y_k^{b_k-1}  = y_1^{\bar b_1} \prod_{k=2}^{d} \(\frac{y_{k-1}}{y_{k}}\)^{-\bar b_k}\(\prod_{k=1}^d y_k^{-1}\).\]
				Since $1/d \leq y_1 \leq 1$ we have that the size or sign of $\bar b_1$ will not affect whether or not $Q_b$ is finite. Hence, we assume without loss of generality that $\bar b_1 = 0$ or, equivalently, that $b_1 = -\sum_{k=2}^d b_k$. Next consider the change of variables $z_k = \log (y_{k-1}) - \log(y_k)$ for $k=2,\dots,d$. This transformation maps the ordered simplex onto $\R_+^{d-1}$ and its Jacobian is determined by $dz = \prod_{k=1}^{d} y_k^{-1}dy$. Thus we obtain
				\[Q_b = \int_{\R^{d-1}_{+}} \exp\(-\sum_{k=2}^d\bar b_kz_k\)dz = \prod_{k=2}^d \int_0^\infty e^{-\bar b_kz}dz.\]
				This expression is finite if and only if $\bar b_k > 0$ for every $k =2,\dots,d$ which completes the proof. 
			\end{proof}
			\begin{proof}[Proof of Lemma~\ref{lem:finite_Z}]
				Note that by a change of variables we equivalently have the representation
				\[Z = \sum_{\tau \in \mathcal{T}_d} \int_{\nabla^{d-1}} \prod_{k=1}^d y_k^{a_k + \gamma_{\tau{(k)}}-1}dy.\]
				By Lemma~\ref{lem:Q_finite} it follows that $Z$ is finite if and only if $\bar a_k + \bar \gamma_{\tau(k)} > 0$ for $k=2,\dots,d$ and every $\tau \in \mathcal{T}_d$. Since we have the inequality $\bar a_k + \bar\gamma_{\tau(k)} \geq \bar a_k + \bar \gamma_{(k)}$ for every $k$ and every $\tau \in \mathcal{T}_d$ with equality achieved for some $\tau \in \mathcal{T}_d$ it follows that $Z < \infty$ if and only if Assumption~\ref{ass:bar_a} holds. This completes the proof.
			\end{proof}
			Next we establish the integration by parts formula Lemma~\ref{thm:IBP}. 
			\begin{proof}[Proof of Lemma~\ref{thm:IBP}]
				To carry out the integration by parts we will make the transformation from $ \Delta^{d-1} \subset \R^d$ to its associated region $E \subset \R^{d-1}$ given by
				\[E = \{x \in [0,\infty)^{d-1}: x_1 + \dots + x_{d-1} \leq 1\}.\] We also define the projection map $\pi:\Delta^{d-1} \to E$ and its inverse $\pi^{-1}:E \to \Delta^{d-1}$ via 
				\[\pi(x) = (x_1,\dots,x_{d-1}), \qquad \pi^{-1}(z) = (z_1,\dots,z_{d-1},1-z^\top {\bf 1}_{d-1}) \]
				respectively. To avoid confusion we will generically denote elements of $\Delta^{d-1}$ by $x$ and elements of $E$ by $z$.
				
				Now recall $c,\xi,p,v$ as in the statement of the lemma and set $\eta(x) = c(x)\xi(x)p(x)$ for $x \in \Delta^{d-1}$. Next for $z \in E$ define $\psi(z) = v(\pi^{-1}(z))$. Then for any $z \in E$ we have
				\[\partial_i \psi(z) = \partial_i v(\pi^{-1}(z)) - \partial_d v(\pi^{-1}(z)); \quad i=1,\dots,d-1.\]
				Consequently, since $\eta(x)^\top {\bf 1}_d = 0$ for every $x \in \Delta^{d-1}$, we see that
				\begin{align*} \nabla v(x)^\top \eta(x) & = \sum_{i=1}^{d-1} \eta_i(x)\partial_i v(x) + \eta_d(x) \partial_d v(x) \\
					& = \sum_{i=1}^{d-1}\eta_i(x)(\partial_i v(x) - \partial_d v(x)) = \sum_{i=1}^{d-1}\eta_i(x)\partial_i\psi(\pi(x)).
				\end{align*}
				Setting $\tilde \eta_i(z) = (\eta_i \circ \pi^{-1})(z)$ for $i=1,\dots,d-1$ and $z \in E$ we have that
				\begin{equation} \label{eqn:projection_int}
					\int_{U} \nabla v^\top c\xi p(x)\, dx = \int_{\pi(U)} \nabla \psi(z)^\top \tilde \eta(z)\, dz
				\end{equation}
				for any $U \subseteq \Delta^{d-1}$.

				Now for $\epsilon > 0$ define $U_\epsilon = \{x \in \Delta^{d-1}: x_{(d)} > \epsilon\}$. The projected region is 
				\[E_\epsilon := \pi(U_\epsilon) = \{z \in E: \min\{z_1,\dots,z_{d-1},1-z^\top{\bf 1}_{d-1}\} > \epsilon\}.\] Note that $\partial E_\epsilon = \{z \in E: = \min\{z_1,\dots,z_{d-1},1-z^\top{\bf 1}_{d-1}\} = \epsilon\}$ and that for $\mathcal{H}^{d-2}$-a.e $z \in \partial E_\epsilon$ the outward unit normal is given by \begin{equation} \label{eqn:outward_pointing_normal}
					\nu(z) = \begin{cases}  -e_i, & \text{if } z_i = \min\{z_1,\dots,z_{d-1},1-z^\top{\bf 1}_{d-1}\}, \\
						\frac{1}{\sqrt{d-1}}{\bf 1}_{d-1},& \text{if } 1-z^\top{\bf 1}_{d-1} = \min\{z_1,\dots,z_{d-1},1-z^\top{\bf 1}_{d-1}\}
					\end{cases}
				\end{equation} where $\mathcal{H}^{d-2}$ is the $d-2$ dimensional Hausdorff measure. In particular the unit normal vector does not depend on $\epsilon$. Since the functions $\psi$ and $\tilde \eta$ are Lipschitz continuous on $\bar E_\epsilon$ we can apply the standard integration by parts formula (see \cite[Corollary~9.66]{Leoni2017first}) for each $\epsilon > 0$ to get
				\begin{equation} \label{eqn:IBP_with_boundary}
					\int_{E_\epsilon} \nabla \psi(z)^\top \tilde \eta(z)\, dz = - \frac{1}{2}\int_{E_\epsilon} \psi(z)\diver \tilde \eta(z)\, dz + \int_{\partial E_\epsilon} \psi(z) \nu(z)^\top \tilde \eta(z)d\mathcal{H}^{d-2}(z).
				\end{equation}
				Now for $z \in \partial E_\epsilon$ write $x$ for $\pi^{-1}(z)$. Then using the form of the outward pointing normal \eqref{eqn:outward_pointing_normal} and the explicit expression for $c$ and $p$ given by \eqref{eqn:c_def} and \eqref{eqn:p_def} respectively we obtain for $\mathcal{H}^{d-2}$-a.e $z \in \partial E_\epsilon$, 
				\begin{align*} |\psi(z) \nu(z)^\top \tilde \eta(z)|  & 
					=\sigma^2\epsilon p(x)|v(x)|\left|x^\top \xi(x) - \xi_{{\bf n}_d(x)}(x)\right|\left(1_{\{x_d \ne x_{(d)}\}} + \frac{1_{\{x_d = x_{(d)}\}}}{\sqrt{d-1}}\right)  \\
					&  \leq C \epsilon^{a_d + \gamma_{(d)}}\prod_{k=1}^{d-1}x_{(k)}^{a_k + \gamma_{{\bf{n}}_k(x)}-1},
				\end{align*}
				where $C = \sigma^2 Z^{-1}\sup_{x \in \Delta^{d-1}} \{|v(x)||x^\top \xi(x) - \xi_{{\bf n}_d(x)}(x)|\}$. Next let $j = {\bf{n}}_d(\gamma)$ and set $\gamma' = (\gamma_1,\dots,\gamma_{j-1},\gamma_{j+1},\dots,\gamma_d) \in \R^{d-1}$.
				Then by noting that $\pi^{-1}(\partial E_\epsilon) = \{x \in \Delta^{d-1}: x_{(d)} = \epsilon\}$ and using the notation of Appendix~\ref{app:integral} we see that 
				\begin{align}
					\left|\int_{\partial E_\epsilon} \right. & \left. \psi(z) \nu(z)^\top \tilde \eta(z)d\mathcal{H}^{d-2}(z)\right|
					\nonumber \\
					& 	\leq C \epsilon^{a_d + \gamma_{(d)}}\int_{\{x\in\Delta^{d-1}:\,  x_{(d)} = \epsilon\}}  \prod_{k=1}^{d-1}x_{(k)}^{a_k + \gamma_{{\bf{n}}_k(x)}-1} d(\mathcal{H}^{d-2} \circ \pi^{-1})(x) \nonumber \\
					& = C \epsilon^{a_d + \gamma_{(d)}}\sum_{\tau \in \mathcal{T}_{d-1}}\int_{\nabla_{d-2}(1-\epsilon,\epsilon)}  \prod_{k=1}^{d-1} y_k^{a_k + \gamma'_{\tau(k)}-1}dy    \nonumber
					= C  \epsilon^{a_d + \gamma_{(d)}} \sum_{\tau \in \mathcal{T}_{d-1}} Q_{a' + \gamma'_{\tau}}(1-\epsilon,\epsilon) , \label{eqn:IBP_estimate}\\
					\tiny\hphantom{m}
				\end{align}
				where $a' = (a_1,\dots,a_{d-1})$. Since $\bar a_k + \bar \gamma_{\tau(k)} \geq \bar a_k + \bar \gamma_{(k)} > 0$ for every $k=2,\dots,d$ and $\tau \in \mathcal{T}_d$ we have that \[a_d + \gamma_{(d)} > \max\{0,-\bar a'_2 - \bar \gamma'_{\tau'(2)},\dots,-\bar a'_{d-1} - \bar \gamma'_{\tau'(d-1)}\}\]
				for every $\tau' \in \mathcal{T}_{d-1}$, where $\bar a'_k + \bar \gamma'_{(k)} = \sum_{l=k}^{d-1} a_l + \sum_{l=k}^{d-1}\gamma_{(l)}$ for $k=2,\dots,d-1$. Hence by Lemma~\ref{lem:Q_lim}\ref{item:Q_b_lim} we have that the right hand side of \eqref{eqn:IBP_estimate} tends to zero as $\epsilon \downarrow 0$. Thus, when sending $\epsilon \downarrow 0$ in \eqref{eqn:IBP_with_boundary} we have that the boundary integral vanishes. Moreover a direct calculation, again using the explicit form of $c$ and $p$, shows that for every $z \in E$,
				$|\diver \tilde \eta(z)| \leq 3d\tilde Cp(\pi^{-1}(z))$, where $\tilde C = \sup_{x \in \Delta^{d-1}}\max_{i=1,\dots,d} (|\partial_i \xi_i(x)| + |\xi_i(x)|)$.
				Hence, by sending $\epsilon \to 0$ in \eqref{eqn:IBP_with_boundary} we obtain
				by dominated convergence that
				\begin{equation} \label{eqn:IBP_proj}
					\int_E \nabla \psi(z)^\top \eta(z)\, dz = - \int_{E} \psi(z)\, \diver \tilde \eta(z)\, dz.
				\end{equation} 
				Finally to obtain \eqref{eqn:IBP} we manipulate $\diver \tilde \eta$. To this end note that for every $z \in E$,
				\begin{equation} \label{eqn:proj}
					\begin{split}
						\diver\tilde \eta(z) = \sum_{i=1}^{d-1}\partial_i\tilde \eta_i(z) & = \sum_{i=1}^{d-1} \partial_i \eta_i(\pi^{-1}(z)) - \partial_d \sum_{i=1}^{d-1} \eta_i(\pi^{-1}(z))  \\
						& = \sum_{i=1}^{d} \partial_i\eta_i(\pi^{-1}(z)) = \diver\eta(\pi^{-1}  z),
					\end{split}
				\end{equation}
				where in the second to last equality we used the fact that $\eta(x)^\top {\bf 1}_d = 0$ for every $x \in \Delta^{d-1}$.
				Consequently, from \eqref{eqn:projection_int} and \eqref{eqn:proj} we see that \eqref{eqn:IBP_proj} is exactly the same as \eqref{eqn:IBP} once we change coordinates.
				This proves the general integration by parts formula. The final claim in the statement of the theorem now follows from the fact that for a function $u \in C^2(\Delta^{d-1})$ we have that $pLu = \frac{1}{2}\diver(c \nabla u p)$ almost everywhere.
			\end{proof}
			We now prove the existence of hybrid Jacobi models.
			\begin{proof}[Proof of Theorem~\ref{thm:process_existence}]
				The existence of a Hunt diffusion $X$ corresponding to the Dirichlet form is guaranteed by \cite[Theorem~7.2.2]{Fukushima1994Dirichlet}. Now 
				In particular, from \cite[Equations~(7.2.22)]{Fukushima1994Dirichlet} we have that the transition kernel $p_t(\cdot, dy)$ of the Hunt diffusion $X$ satisfies 
				\[ \label{eqn:quasi_cont_version}
				p_tu :=\int_{\Delta^{d-1}}p_t(\cdot,dy)u(y) \text{ is a quasi-continuous version of } T_tu
				\]
				for every $u \in C(\Delta^{d-1})$. Here $(T_t)_{t \geq 0}$ are the strongly continuous semigroup operators on $L^2(\Delta^{d-1},m)$ corresponding to the generator of the Dirichlet form $(L,D(L))$. From the Kolmogorov Backward equation (see e.g.\ \cite[Theorem~4.7]{eberle2009markov}) we have the identity
				\begin{equation} \label{eqn:Kolmogorov_backward}
					T_tu - u = \int_0^t (T_sLu)ds, \quad u \in D(L),
				\end{equation}
				where the integral is a Bochner integral and the equality is understood to hold on $L^2(\Delta^{d-1},m)$. Using the fact that $p_tu$ is a version of $T_tu$ we obtain from \eqref{eqn:Kolmogorov_backward} that 
				\begin{equation} \label{eqn:transition_function_relation}
					p_tu(x)  - u(x) - \int_0^t p_sLu(x)ds= 0, \quad m\text{-a.e.}\ x\in \Delta^{d-1}.
				\end{equation}  But by quasi-continuity of $p_tu$ we have for that for every $\epsilon > 0$ there exists an open set $G_\epsilon(u) \subset \Delta^{d-1}$ (depending on $u$) with $\text{Cap}(G_\epsilon(u)) < \epsilon$ and such that $p_tu$ is continuous on $\Delta^{d-1}\setminus G_\epsilon(u)$. Set $\tilde N(u) = \cap_{\epsilon > 0} G_\epsilon(u)$. By continuity of $p_tu$  we see that \eqref{eqn:transition_function_relation} holds for all $x \in \Delta^{d-1}\setminus \tilde N(u)$ and that $\text{Cap}(\tilde N(u)) = 0$. By \cite[Theorem~4.2.1(ii)]{Fukushima1994Dirichlet} it follows that $\tilde N(u)$ is an exceptional set, since it has zero capacity, and so, by \cite[Theorem~4.1.1]{Fukushima1994Dirichlet} there exists a properly exceptional set $N(u) \subset \Delta^{d-1}$ with $\tilde N(u) \subseteq N(u)$. By definition of properly exceptional \ref{item:existence_i} is satisfied with this choice of set $N(u)$.
				
				Now we establish that $M^u(t):= u(X(t)) - u(X(0)) - \int_0^t Lu(X(s))ds$ is a $\P_x$-martingale for every $x \in \Delta^{d-1}\setminus N(u)$. Let such an $x \in \Delta^{d-1}\setminus N(u)$ be given. Using the Markov property and definition of transition function we have for every $0 \leq s \leq t$ that 
				\begin{equation} \label{eqn:mart_problem_proof}
					\begin{split}
						\E[M^u(t) - M^u(s)|\Fcal(s)] =	\E_x& \left[u(X(t)) - u(X(s)) - \int_s^t Lu(X(r))dr|\F(s)\right] \\
						&   = p_{t-s}u(X(s)) - u(X(s)) - \int_0^{t-s}p_rLu(X(s)))dr. \end{split} 
					\qquad \P_x\text{-a.s.}
				\end{equation} 
				But by property \ref{item:existence_i}, $\P_x(X(s) \in N(u)\text{ for some } s > 0) = 0$ so, by \eqref{eqn:transition_function_relation}, we have that \eqref{eqn:mart_problem_proof} is zero $\P_x$-a.s. This establishes the martingale property.
				
				Now to finish the proof we need to obtain a single properly exceptional set such that  \ref{item:existence_i} and \ref{item:existence_ii} hold. Consider the countable dense set $D_0 \subset D(L)$ consisting of all polynomials with rational coefficients. Since a countable union of properly exceptional sets is again properly exceptional we obtain a single properly exceptional set $N = \bigcup_{u \in D_0} N(u)$ such that \ref{item:existence_ii} holds and $M^u$ is a $\P_x$-martingale for every $x \in \Delta^{d-1}\setminus N$ and every $u \in D_0$. Then by approximating any $u \in D(L)$ by uniformly bounded $u_n \in D_0$ such that $u_n \to u $ and $Lu_n \to Lu$ pointwise we readily obtain by bounded convergence that \ref{item:existence_ii} holds for every $u \in D(L)$ completing the proof. 
			\end{proof}
			Next we prove the various properties of hybrid Jacobi processes. Recall that for an index set $I \subseteq \{1,\dots,d\}$, as in \eqref{eqn:Lambda_def}, we denote by $\Lambda_I$ the function given by $\Lambda_I(x) = \sum_{i \in I} x_i$ for $x \in \R^d$. We now prove the first two items of Proposition~\ref{prop:Lebesgue_collision}.
			\begin{proof}[Proof of Proposition~\ref{prop:Lebesgue_collision}\ref{item:boundary_zero} and \ref{item:collision_zero}]
				First, we will show that $L_{\Lambda_I(X)} \equiv 0$ for every index set $I$. By the occupation density formula and Lemma~\ref{lem:Lambda} we have for every $\epsilon > 0$ that 
				\begin{align*} 
					\frac{1}{\epsilon}\int_0^\epsilon &L^a_{\Lambda_I(X)}(T)da  = \int_0^T\frac{1_{\{\Lambda_I(X(t))\in (0,\epsilon)\}}}{\epsilon}d[\Lambda_I(X),\Lambda_I(X)](t)  \\
					& = \sigma^2 \int_0^T \frac{1_{\{\Lambda_I(X(t))\in (0,\epsilon)\}}}{\epsilon}\Lambda_I(X(t))(1-\Lambda_I(X(t)))dt \leq \sigma^2 \int_0^T 1_{\{\Lambda_I(X(t))\in (0,\epsilon)\}}dt.
				\end{align*} 
				Sending $\epsilon \downarrow 0$ and using the right continuity of $a \mapsto L^a_{\Lambda_I(X)}$ on the left hand side and the Lebesgue dominated convergence theorem on the right hand side yields $L_{\Lambda_I(X)}(T) = 0$, $\P_\mu$-a.s.\ for every $T \geq 0$.
				
				Next note that for any nonnegative semimartingale $Y$ satisfying $L_Y(T) = 0$ we have from Tanaka's formula that 
				\begin{align*}
					Y(T)  = |Y(T)| & = Y(0) + \int_0^T \text{sign}(Y(t))dY(t) + L_Y(T)  \\
					& = Y(0) + \int_0^T 1_{\{Y(t) > 0\}}dY(t) - \int_0^T 1_{\{Y(t) = 0\}} dY(t),
				\end{align*}
				Rearranging gives $1_{\{Y(t) = 0\}}dY(t) = 0$. Applying this identity with $Y = \Lambda_I(X)$ and expanding out the stochastic integral using \eqref{eqn:X_polynomial_dynamics} we obtain
				\begin{equation} \label{eqn:Lambda_Lebesgue_time}\int_0^T\(\sum_{i \in I}a_{{\bf r}_i(t)} + \gamma_i \)1_{\{\Lambda_I(X(t)) = 0\}}dt = 0.
				\end{equation}
				Now we claim that \begin{equation} \label{eqn:Lambda_induct}
					\int_0^T 1_{\{\Lambda_I(X(t)) = 0\}}dt = 0
				\end{equation} for every index set $I$. We prove this claim by backward induction on $|I|$, the size of $I$. The base case $|I|=d$ is trivial since in that case $\Lambda_I(X) = \sum_{i=1}^d X_i \equiv 1$.  
				Now let $l \in \{2,\dots,d\}$ be given and assume that the result holds for index sets of size $l$. We will show that the result holds for index sets $I$ of size $l-1$. For any such index set $I$ note that as a consequence of the inductive hypothesis we have for a.e.\ $t \geq 0$ that
				\begin{equation} \label{eqn:induction_tool}
					\{\Lambda_I(X(t))  = 0\} \subseteq \{{\bf{n}}_k(t) \in I \text{ for every } k =d-l+2,\dots,d\}.
				\end{equation}
				Indeed, fixing time $t$, if indices in $I$ do not occupy the smallest $l-1$ ranks when $\Lambda_I(X(t)) = 0$ then there would be another index $j$ (depending on $t$) such that $X_{j}(t) = 0$. But then we would have that $\Lambda_J(X(t)) = 0$ where $J = I \cup \{j\}$. The inductive hypothesis ensures that this can only happen at a Lebesgue null-set of time points.
				
				Hence, using  \eqref{eqn:Lambda_Lebesgue_time} and \eqref{eqn:induction_tool} we obtain that 
				$(\sum_{k =d-l+2}^{d}a_k + \sum_{i \in I}\gamma_i)\int_0^T 1_{\{\Lambda_I(X(t)) = 0\}}dt  = 0.$
				By Assumption~\ref{ass:bar_a} we see that $ \sum_{k =d-l+2}^{d}a_k + \sum_{i \in I}\gamma_i \geq \bar a_{d-l+2} + \bar \gamma_{(d-l+2)} > 0$ so we must have that $\int_0^T 1_{\{\Lambda_I(X(t) = 0)\}}dt = 0$ proving \eqref{eqn:Lambda_induct}. Since this holds for every $T \geq 0$ we establish that $\int_0^\infty 1_{\{\Lambda_I(X(t) = 0)\}}dt = 0$. The statement of the lemma is precisely the special case when $|I| = 1$ so this completes the proof of the first item.
				
				The proof of the second item is very similar to the proof of \cite[Lemma~4.1]{itkin2021class}. Fix indices $i \ne j$. The occupation density formula for the semimartingale $X_i - X_j$ then yields
				\[\int_0^T f(X_i(t) - X_j(t))(X_i(t) + X_j(t) - (X_i(t) - X_j(t))^2)dt = \int_{\R} f(a)L^a_{X_i - X_j}(T)da\]
				for every bounded measurable function $f$, where we used \eqref{eqn:qv_formula} with $u=v = e_i - e_j$ to compute the quadratic variation of $X_i - X_j$. Taking the function $f(a) = 1_{\{a=0\}}$ we obtain 
				\begin{equation} \label{eqn:occ_dens}
					\int_0^T 1_{\{X_i(t) = X_j(t)\}}(X_i(t) + X_j(t) - (X_i(t) - X_j(t))^2)dt = 0.
				\end{equation}
				Note that $X_i(t) + X_j(t) - (X_i(t) - X_j(t))^2$ is nonnegative and is equal to $0$ if and only if $X_i(t) = X_j(t) = 0$ or $X_i(t) = 1$ or $X_j(t) = 1$. However, by part \ref{item:boundary_zero} this only happens at a Lebesgue null-set of time points. Hence, from \eqref{eqn:occ_dens} we see that $1_{\{X_i(t) = X_j(t)\}} = 0$, $\P_\mu$-a.s.\ for almost every $t \geq 0$. Since $i$ and $j$ were arbitrary this completes the proof.
			\end{proof}
			Next we investigate the question of triple collisions and prove Proposition~\ref{prop:Lebesgue_collision}\ref{item:triple_zero}. To assist us in this analysis we will compare the process $X$ to a different, but related, process for which triple collision properties are known. In \cite{itkin2021class}, rank-based continuous semimartingales on the simplex were constructed using Dirichlet forms. However, rather than taking the closed simplex $\Delta^{d-1}$ as the state space, the processes constructed in that setting were on the open simplex with absorption at the boundary; that is the state space was the one point compactification of the open simplex $\Delta^{d-1}_{+} \cup \{\Theta\}$ where the point $\Theta$ acted as a cemetery state for the process. Moreover, the (pre-)Dirichlet form $\mathcal{E}$ as in \eqref{eqn:X_Dirichlet_form} with core $C_c^\infty(\Delta^{d-1}_+)$ rather than $\mathcal{D}$ is under the purview of \cite{itkin2021class}. It follows from \cite[Theorem~3.1]{itkin2021class} that the resulting Dirichlet form produces a process $X^{\Theta}$ with generator $L$ as in \eqref{eqn:generator_d}, but with a different domain $D^\Theta(L)$. The process $X^{\Theta}$ possesses the strong Feller property and for every $\nu \in \mathcal{P}(\Delta^{d-1}_+ \cup \{\Theta\})$  there exists a probability measure $\P_\nu^\Theta$ such that $\P_\nu^\Theta(X^{\Theta}(0) \in \Gamma) = \nu(\Gamma)$ for every Borel set $\Gamma \subset \Delta^{d-1}_+ \cup \{\Theta\}$ and $\P^\Theta_\nu$ solves the martingale problem for $(L,D^{\Theta}(L))$; i.e.\
			$
			u(X^{\Theta}(T)) - u(X^{\Theta}(0)) - \int_0^T Lu(X^\Theta(t))dt 
			$
			is a martingale under $\P_\nu^\Theta$ for every $u \in D^\Theta(L)$ where, by convention, we extend any function $v$ from the open simplex $\Delta^{d-1}_+$ to the one point compactification $\Delta^{d-1}_+ \cup \{\Theta\}$ via $v(\Theta) = 0$. 
			
			Next we will show that the processes $X$ and $X^\Theta$ coincide on the interior of the simplex. To this end for every $n \in \N$ define the open sets 
			\begin{equation}\label{eqn:E_n}
				E_n = \{x \in \Delta^{d-1}: x_{(d)} > 1/n\}.
			\end{equation} Following \cite[Section~4.6]{Ethier1986Markov}, we say that a process $Z$ solves the \emph{stopped martingale problem corresponding to $L$ on $E_n$} with initial law $\xi \in \mathcal{P}(\Delta^{d-1})$ if $Z(0) \sim \xi$ and 
			\[u(Z(T \land \tau_n)) - u(Z(0)) - \int_0^{T \land \tau_n}Lu(Z(t))dt\]
			is a martingale for every $u \in D(L)$ where $D(L)$ is defined in Section~\ref{sec:hybrid_construction}. Here
			\begin{equation} \label{eqn:tau_n} 
				\tau_n = \inf\{t \geq 0: Z(t) \not \in E_n \}.
			\end{equation}
			Note that the sets \[\{u:E_n \to \R \ \big | \  u = v|_{E_n} \text{ for some } v \in D(L)\}, \quad \{u:E_n \to \R \ \big| \ u = v|_{E_n} \text{ for some } v \in D^\Theta(L)\}\] are identical for every $n$.
			Hence both $X(\cdot \land \tau_n)$ and $X^{\Theta}(\cdot \land \tau_n)$ solve the stopped martingale problem for $L$ on $E_n$ when their initial laws coincide.
			\begin{lem} \label{lem:law_lemma}
				For any $\mu \in \mathcal{P}_0$ we have that $\P_\mu|_{ \mathcal{F}(\tau_n)} = \P_\mu^\Theta|_{\mathcal{F}(\tau_n)}$ for every $n \in \N$. That is, the laws of $X(\cdot \land \tau_n)$ and $X^{\Theta}(\cdot \land \tau_n)$ are the same when the processes share the same initial distribution. 
			\end{lem}
			\begin{proof}
				Fix $n$ and recall that we view $\Delta^{d-1}$ as a subset of $\R^{d-1}$ via the transformation $(x_1,\dots,x_{d-1}) \mapsto (x_1,\dots,x_{d-1},1-\sum_{i=1}^{d-1}x_i)$. Thus, with some abuse of notation, in this proof we will view the drift and volatility coefficients appearing in \eqref{eqn:generator_d} as defined on $\R^{d-1}$ and $\mathbb{S}^{d-1}_+$ respectively. 
				
				Since both $X(\cdot \land \tau_n)$ and $X^{\Theta}(\cdot \land \tau_n)$ solve the stopped martingale problem with the common initial condition the result will follow if we can show that the stopped martingale problem is well-posed. To this end let $\psi \in C_c^\infty(\R^{d-1})$ be such that $\psi = 1$ on $E_n$ and $\psi = 0$ on $\R^{d-1} \setminus E_{n+1}$. Extend $c$ to all of $\R^{d-1}$ by zero and define $\hat c:\R^{d-1} \to \mathbb{S}^{d-1}_+$ and $\hat b:\R^{d-1} \to \R^{d-1}$ via
				\begin{align*}
					\hat b_i(x) &= (\gamma_i + a_{{\bf r}_i(x)} - (\bar a_1 + \bar \gamma_1) x_i)1_{E_n}(x), \quad i=1,\dots,d,  \\
					\hat c(x) & = c(x)\psi(x) + I_{d-1}(1-\psi(x))
				\end{align*} 
				where $I_{d-1}$ is the $(d-1) \times (d-1)$ identity matrix. Define the corresponding generator 
				$\hat L = \frac{1}{2}\sum_{i,j=1}^d \hat c_{ij}\partial_{ij} + \sum_{i=1}^d \hat b_i \partial_i.$ Since $\hat c$ is uniformly elliptic and both $\hat b$ and $\hat c$ are bounded it follows by \cite[Theorem 7.2.1]{Stroock1979Multi} that the martingale problem corresponding to $\hat L$ is well-posed. By \cite[Theorem 5.6.1]{Ethier1986Markov} it then follows that the stopped martingale problem corresponding to $\hat L$ on $E_n$ is well-posed. But $(\hat L u)|_{E_n}= L(u|_{E_n})$ for every $u \in D(\hat L)$ which completes the proof. 
			\end{proof}
			\begin{proof}[Proof of Proposition~\ref{prop:Lebesgue_collision}\ref{item:triple_zero}]
				Next let $E_n$ be as in \eqref{eqn:E_n} and set $B_n = B \cap E_n$. Define the sequence of stopping times by setting $\sigma_n^1 = 0$ and 
				\begin{align*}
					\tau^m_n & = \inf\{t \geq \sigma^m_n: X(t) \not \in E_{n}\} & m \geq 1, \\
					\sigma^m_n & = \inf\{t \geq \tau^{m-1}_n: X(t) \in \bar E_{n-1}\} & m \geq 2.
				\end{align*}
				Note that $\tau_n^1 = \tau_n$ where $\tau_n$ was given by \eqref{eqn:tau_n} with $Z$ replaced by $X$. We have that 
				\begin{equation} \label{eqn:stopping_time_inequality}
					\P_\mu(X(t) \in B_{n-1} \text{ for some } t >0)  \leq \sum_{m=1}^\infty \P_\mu(X(t) \in B_{n-1} \text{ for some } t \in (\sigma_n^m,\tau_n^m)).
				\end{equation}  
				Applying iterated conditioning and the strong Markov property we obtain for each $m \geq 1$ that
				\begin{align*}
					\P_\mu(X(t) \in B_{n-1} & \text{ for some } t \in  (\sigma_n^m,\tau_n^m)) \\
					& = \E_\mu[\P_\mu(X(t) \in B_{n-1}\text{ for some } t \in (\sigma_n^m,\tau_n^m) \ |\F(\sigma_n^m))] \\
					& = \E_\mu[\P_{X_{\sigma_n^m}}(X(t) \in B_{n-1} \text{ for some } t \in (0,\tau_n))1_{\{\sigma_n^m< \infty\}}].
				\end{align*}
				But by Lemma~\ref{lem:law_lemma} we have that the laws of $X$ and $X^{\Theta}$ agree on the time interval $(0,\tau_n)$. Hence for $\P_\mu$-a.e $\omega$\footnote{Specifically, we have that $\P_\mu(X_{\sigma_n^m}\not \in N, \sigma_n^m < \infty) = 1$, where $N$ is as in the statement of Theorem~\ref{thm:process_existence}.} we have that \begin{align*}
					\P_{X_{\sigma_n^m}(\omega)} (X(t) \in B_{n-1} \text{ for some } & t \in (0,\tau_n)) \\ & =  \P^{\Theta}_{X_{\sigma_n^m}(\omega)}(X^\Theta(t) \in B_{n-1} \text{ for some } t \in (0,\tau_n)) = 0.
				\end{align*} The final equality follows from \cite[Theorem 4.3]{itkin2021class} which guarantees for \emph{every} $x \in \Delta^{d-1}_+$ that $\P^\Theta_x(X^\Theta(t) \in B \text{ for some } t > 0) = 0$. It follows from \eqref{eqn:stopping_time_inequality} that $\P_\mu(X(t) \in B_{n-1} \text{ for some } t >0) = 0$. Since 
				\[\{X(t) \in B \text{ for some } t > 0\} = \bigcup_n \{X(t) \in B_n \text{ for some } t > 0\}\]
				the result follows.
			\end{proof}

			Finally, to close out this section we prove Theorem~\ref{thm:boundary_attainment}.
			\begin{proof}[Proof of Theorem~\ref{thm:boundary_attainment} ]
				We first prove the forward directions of \ref{item:bound_attain_name}\ref{item:bound_attain_name_hit} and \ref{item:bound_attain_rank}\ref{item:bound_attain_rank_hit}, and then prove the backward directions. To this end fix $I \subset \{1,\dots,d\}$ as in the statement of \ref{item:bound_attain_name}. Assume that there exists an $l \in \{2,\dots,d-N+1\}$ such that $\bar a_l + \sum_{i \in I} \gamma_i + \sum_{k=l}^{d-N} \gamma^{-I}_{(k)} < 1$. Let $J$ be an index set of size $d-l+1$, so that $\{\gamma_j\}_{j \in J} = \{\gamma_i\}_{i \in I} \cup \{\gamma_{(k)}^{-I}\}_{k=l}^{d-N}$; that is $J$ contains $I$ as well as the additional $d-l+1-N$ indices corresponding to smallest remaining values in the vector $\gamma$. 
				Note that $\sum_{j \in J} \gamma_j = \sum_{i \in I} \gamma_i + \sum_{k=l}^{d-N} \gamma_{(k)}^{-I}$.
				We will show that the process $\Lambda_J(X(t))$ hits zero with positive probability. Assume by way of contradiction that $\P_\mu(\Lambda_J(X(t)) = 0 \text{ for some } t > 0)  =0$. Then $-\log \Lambda_J(X)$ is a semimartingale and using Itô's formula we obtain from \eqref{eqn:Lambda_dynamics} that
				\[
				-d\log \Lambda_J(X) = \frac{\sigma^2}{2}\(\frac{1- \sum_{j \in J} (\gamma_j +  a_{{\bf r}_j(\cdot)})}{\Lambda_J(X)} + \bar a_1 + \bar \gamma_1 -1\)dt +  \sigma \sqrt{\frac{1- \Lambda_J(X)}{\Lambda_J(X)}}dB.
				\] Next, define the sets
				\begin{equation} \label{eqn:A_J}
					A_J := \{x \in \Delta^{d-1}: {\bf{n}}_k(x) \in J \text{ for every } k = l,\dots,d\}, \qquad A_J^c = \Delta^{d-1} \setminus A_J.
				\end{equation}
				The set $A_J$ corresponds to the region of the simplex for which the smallest components of $x$ are precisely those with indices in the set $J$. Fix $K > 0$ and set $A_J^K = \{ x \in A_J: \log x_{(l-1)} - \log x_{(l)} > K\}$. 
				Choose a function $\psi \in C_c^\infty(\Delta^{d-1})$ such that $0 \leq \psi \leq 1$,  $\psi = 1$ on $A_J^K$ and $\psi = 0$ on $\Delta^{d-1}\setminus A_J$. Also let $\epsilon  = 1 - \bar a_l - \sum_{i \in I} \gamma_i - \sum_{k=l}^{d-N} \gamma^{-I}_{(k)}$ and note by assumption that $\epsilon > 0$. 
				Then using the product rule we obtain
				\begin{align*}
					-d(\log \Lambda_J(X)\psi(X)) = & \frac{\sigma^2}{2}\psi(X)\(\bar a_1 + \bar \gamma_1 -1 + \frac{1- \sum_{j \in J} \gamma_j - \sum_{j \in J} \sum_{k=1}^d 1_{\{{\bf{n}}_k(t) = j\}}a_k}{\Lambda_J(X)}\)dt \\
					& - \frac{\sigma^2}{2} \log \Lambda_J(X)\sum_{i=1}^d\partial_i \psi(X) \(\gamma_i + \sum_{k=1}^d1_{\{{\bf{n}}_k(t) = i\}}a_k - (\bar a_1 + \bar \gamma_1) X_i\)dt \\ & 
					+ \psi(X)dM(t) + dN(t) +  \nabla \psi(X)^\top c(X) \frac{\boldsymbol{1}^J}{\Lambda_J(X)}dt,
				\end{align*}
				where $dM(t) = \sigma \sqrt{\frac{1- \Lambda_J(X(t))}{\Lambda_J(X(t))}}dB(t)$, $dN(t) = \log\Lambda_J\nabla \psi^\top c^{1/2}(X(t))dW(t)$, $\boldsymbol{1}^J_j = 1_{\{j \in J\}}$ for $j=1,\dots,d$ and we omitted the time index for notational clarity. Note that 
				$(c(x){\bf 1}^J/\Lambda_J(x))_i = -x_i  +x_i1_{\{i \in J\}}/\Lambda_J(x) \in [-1,0]$ for ever $i$.
				Hence, using the fact that $\psi = 0$ on $A_J^c$ and that $\Lambda_J(x) = \sum_{k=l}^d x_{(k)} = \bar x_{(l)}$ for $x \in A_J$ we obtain the estimate
				\begin{align*}
					-\log \Lambda_J(X(T))\psi(X(T))  \geq\,  & \epsilon\int_0^T \frac{\sigma^2\psi(X)}{2\bar X_{(l)}}1_{A_J}(X)dt + C \int_0^T \(-1  + 1_{A_J}(X)\log \bar X_{(l)}\)dt \\
					& + \int_0^T\psi(X)dM(t) + N(T),
				\end{align*} for some universal constant $C > 0$.  Then setting $d\tilde M(t)= \psi(X(t))dM(t)$ and using the fact that $0 \leq \psi \leq 1$ together with the fact that $\Lambda_J(x) = \bar x_{(l)}$ on $A_J$ we note that 
				\begin{equation} \label{eqn:tilde_M_bound}
					\left[\tilde M(T), \tilde M(t)\right] = \int_0^T\sigma^2 \psi^2(X)\frac{1-\Lambda_J(X)}{\Lambda_J(X)}\, dt \leq \int_0^T \frac{\sigma^2\psi(X)}{2\bar X_{(l)}}1_{A_J}(X)dt.
				\end{equation}
				Hence we obtain
			
			\begin{equation}\label{eqn:log_sigma_inequality}
				\begin{split}
					-\log \Lambda_J(X(T))\psi(X(T)) \geq  & \int_0^TC\(-1  + 1_{A_J}(X(t))\log \bar X_{(l)}(t)\)dt \\
					& + \frac{\epsilon}{2} [\tilde M,\tilde M](T) + \tilde M(T) + N(T).
				\end{split}
			\end{equation}
			By the ergodic property we have that 
			\begin{equation} \label{eqn:log_sigma_finite}
				\begin{split} 
					\lim_{T\to \infty} \frac{1}{T}\int_0^T1_{A_J}(X(t)) \log \bar X_{(l)}dt & = \int_{\Delta^{d-1}} 1_{A_J}(x)\log \bar x_{(l)}p(x)dx \\
					& = \frac{1}{d!}\sum_\tau \int_{\nabla^{d-1}} \log \bar y_{l}\prod_{k=1}^d y_k^{a_k + \gamma_{\tau(k)-1}}\, dy   > - \infty,
				\end{split}
			\end{equation}
			where the sum is taken over permutations $\tau$ such that $\tau(k) \in J$ for $k=l,\dots,d$ and finiteness is due to \eqref{eqn:log_bound}.
			Next note that \[[N,N](T) \leq C\int_0^T (\log \Lambda_J(X(t)))^21_{A_J}(X(t))dt = C\int_0^T (\log \bar X_{(l)}(t))^21_{A_j}(X(t))dt.\]
			Hence, by the ergodic property and \eqref{eqn:log_bound} we again have $\lim_{T \to \infty} \frac{[N,N](T)}{T} < \infty$, $\P_\mu$-a.s.\ so by \cite[Lemma~1.3.2]{fernholz2002stochastic} we conclude that \begin{equation} \label{eqn:mart_limit}
				\lim_{T\to \infty} \frac{N(T)}{T} = 0, \quad \P_\mu\text{-a.s.}
			\end{equation}
			Next we wish to perform a similar analysis for $\tilde M$. To this end
			let $\mathcal{T}_J= \{\tau \in \mathcal{T}: \tau(k) \in J \text{ for every } k =l,\dots,d\}$. 
			Note that \[\int_{\Delta^{d-1}} \psi(x)^2\frac{1_{A_J}(x)}{
				\bar x_{(l)}}p(x)dx \geq \frac{Z^{-1}}{(d-l+1)}\sum_{\tau \in \mathcal{T}_J} \int_{\nabla^{d-1} \cap \{\log \frac{y_{l-1}}{y_l} > K\}} \frac{1}{y_{l}}\prod_{k=1}^d y_k^{a_k + \gamma_{\tau(k)}-1}dy,\]
			where we used the fact that $\psi = 1$ on $A_J^K$ and that $\bar x_{(l)} \leq (d-l+1)x_{(l)}$.
			For $\tau \in \mathcal{T}_J$ set $b^\tau = a + \gamma_{\tau} -e_l$.
			The integrability of the right hand side is governed by Lemma~\ref{lem:Q_finite}; note that the intersection with the set $\{\log \frac{y_{l-1}}{y_l} > K\}$ in the domain of integration will not affect integrability. Indeed, arguing as in Lemma~\ref{lem:Q_finite}, we have that
			\[ \int_{\nabla^{d-1} \cap \{\log \frac{y_{l-1}}{y_{l}} > K\}} \frac{1}{y_{l}}\prod_{k=1}^d y_k^{a_k + \gamma_{\tau(k)}-1}dy \ \propto \int_{K}^\infty e^{- \bar b^\tau_{l} z}dz \prod_{k\ne l} \int_{0}^\infty e^{- \bar b^\tau_k z}dz = \infty,\] since $\bar b^\tau_l  = \bar a_l + \sum_{i \in I} \gamma_i + \sum_{k=l}^{d-N} \gamma^{-I}_{(k)} - 1 < 0$. Consequently, from the estimate \eqref{eqn:tilde_M_bound} we see that 
			\begin{equation} \label{eqn:main_blowup}
				\lim_{T\to \infty} \frac{[\tilde M,\tilde M](T)}{T} = \infty, \quad \P_\mu\text{-a.s.}
			\end{equation} 
			
			By the Dambis--Dubins--Schwarz theorem, \eqref{eqn:main_blowup} and the law of large numbers for Brownian motion we have that $\lim_{T\to \infty} \frac{ \tilde M(T)}{[\tilde M, \tilde M](T)} = 0$, $\P_{\mu}$-a.s. Dividing by $T$ and sending $T \to \infty$ in \eqref{eqn:log_sigma_inequality} shows that $\lim_{T \to \infty}-T^{-1}\psi(X(T))\log \Lambda_J(X(T)) = \infty,$  $\P_\mu$-a.s.\ by virtue of \eqref{eqn:log_sigma_finite}, \eqref{eqn:mart_limit} and  \eqref{eqn:main_blowup}. But then, by definition of $\Lambda_J$, we must have $\lim_{T\to \infty} \Lambda_J(X(T)) = 0$, $\P_{\mu}$-a.s.\ This contradicts the ergodicity of $X$. Hence $-\log \Lambda_J(X)$ cannot be a semimartingale and consequently we must have that 
			$
			\P_\mu(\Lambda_J(X(t)) = 0 \text{ for some } t > 0) > 0.
			$
			Since $\Lambda_J(X(t)) = 0 \implies \Lambda_I(X(t)) = 0$ we have proved the forward direction of \ref{item:bound_attain_name}\ref{item:bound_attain_name_hit}.
			
			We now prove the forward direction of \ref{item:bound_attain_rank}\ref{item:bound_attain_rank_hit}. To this end fix $k \in \{2,\dots,d\}$ and suppose that $\bar a_l + \bar \gamma_{(l)} < 1$ or some $l \in \{2,\dots,k\}$. Set 
			\[I = \{i \in \{1,\dots,d\}: {\bf{n}}_m(\gamma) = i \text{ for some } m \geq l\}\]
			and note that $N := |I| = d-l+1$.
			Then $\bar a_l + \sum_{i \in I}\gamma_i + \sum_{m=l}^{l-1}\gamma_{(m)}^{-I} = \bar a_l + \gamma_{(l)} < 1$, so that \eqref{eqn:param_cond_I} does not hold here (we evaluated \eqref{eqn:param_cond_I} at the largest index $d-N+1$).  Hence by \ref{item:bound_attain_name}\ref{item:bound_attain_name_hit} we have that $\P_\mu(\Lambda_I(X(t)) = 0 \text{ for some } t > 0) > 0$. But we have the set inclusions
			\[\{\Lambda_I(X(t)) = 0\} = \{X_i(t) = 0 \text{ for all } i \in I\} \subseteq \{X_{(l)}(t) = 0\} \subseteq \{X_{(k)}(t) = 0\}.\]
			This proves the forward direction of \ref{item:bound_attain_rank}\ref{item:bound_attain_rank_hit}.
			
			Now we prove the backward direction of \ref{item:bound_attain_name}\ref{item:bound_attain_name_hit}. Fix $I \subset \{1,\dots,d\}$ as in the statement of \ref{item:bound_attain_name} and suppose that $\bar a_l + \sum_{i \in I} \gamma_i + \sum_{k=l}^{d-N} \gamma^{-I}_{(k)} \geq 1$ for every $l = 2,\dots,d-N+1$. We will prove that $\Lambda_J(X)$ does not hit zero $\P_\mu$-a.s.\ for any index set $J$ containing $I$. We proceed by finite backward (strong) induction on the size of $J$ with base case $|J|=d$ and terminal case $|J| = N$, which is only possible if $J = I$. The base case is trivial since when $|J| = d$ we have that $\Lambda_J(X) = \sum_{i=1}^d X_i \equiv 1$.

			Before proceeding to the inductive step we introduce some notation. For some $l \leq d-1$ and an index set $K \subseteq \{1,\dots,d\}$ with $|K|= l$ define the stopping times
			\begin{align*}\tau^{K}_n &= \inf\{t \geq 0: \Lambda_{K}(X(t)) \leq 1/n\}, & n \in \N, \\
				\tau^{K} &= \inf\{t \geq 0: \Lambda_{K}(X(t)) = 0\}, \\
				\tau^{K,+}_m& = \min\{\tau_m^{J^+}: K \subsetneq J^+ \subseteq \{1,\dots,d\}\}, & m \in \N.
			\end{align*}
			The stopping time $\tau_m^{K,+}$ denotes the first time that $\Lambda_{K}(X(t))$ is less than or equal to $1/m$ for \emph{some} index set $J^+$ of size at least $l+1$ containing $K$. 
			Now we proceed with the finite induction by assuming for some $l \in \{N+1,\dots,d\}$ that $\Lambda_{J^+}(X)$ does not hit zero $\P_\mu$-a.s.\ for any set $J^+ \subset \{1,\dots,d\}$ with $|J^+| \geq l$ and $I \subseteq J^+$. Fix an arbitrary set $J\subset \{1,\dots,d\}$ with $|J| = l-1$ and $I \subseteq J$. We have to show that 
			$\P_\mu(\Lambda_{J}(X(t)) = 0 \text{ for some } t > 0) = 0.$
			
			
			
			We have for $t \leq \tau_m^{J,+}$ that $\Lambda_{J^+}(X(t)) \geq 1/m$ for every $J^+$ of size $l$ containing $I$. Moreover, if $X(t) \in A_{J}^c$ then necessarily $\{X_j(t)\}_{j \in J}$ do not occupy the smallest $l-1$ ranks. Consequently, there exists and index $j_* \not \in J$ such that $X_{j_*}(t) \leq X_j(t)$ for some $j \in J$. Set $J^+ = J \cup \{j_*\}$. We then have for $t \in [0,\tau_m^{J,+}]$ that  
			\begin{equation} \label{eqn:Lambda_m_bound}
				2\Lambda_J(X(t)) \geq \Lambda_{J^+}(X(t))1_{A_J^c}(X(t))\geq \frac{1}{m}1_{A_J^c}(X(t)).
			\end{equation}
			Next note that up to time $\tau_n^{J}$, the process $-\log \Lambda_{J}(X)$ is bounded so we can apply Itô's formula to obtain from \eqref{eqn:Lambda_dynamics} that
			\begin{equation}\label{eqn:Sigma_stopping_time_base}
				\begin{split}
					-\log & \Lambda_{J}(X(T\land \tau_n^{J} \land \tau_m^{J,+}))  =  -\log \Lambda_{J}(0)  + M(T \land \tau_n^{J} \land \tau_m^{J,+}) \\
					& + \frac{\sigma^2}{2}\int_0^{T \land \tau_n^{J} \land \tau_m^{J,+}} \(\bar a_1 + \bar \gamma_1-1 + \frac{1-\sum_{j \in J}(\gamma_j + a_{{\bf r}_j(t)})}{\Lambda_{J}(X(t))}\)dt
				\end{split} 
			\end{equation}
			for some martingale $M(\cdot \land \tau_n^{J} \land \tau_m^{J,+})$. We then estimate for any $t  \leq T \land \tau_n^{J} \land \tau_m^{J,+}$ that
			\begin{align*}
				\frac{1-\sum_{j \in J} (\gamma_j + a_{{\bf r}_j(t)})}{\Lambda_{J}(X(t))} = & \frac{1  - \sum_{j \in J} \gamma_j - \bar a_{d-l+1}}{\Lambda_{J}(X(t))}1_{A_{J}}(X(t)) \\
				& \hspace{1cm} + \frac{1- \sum_{j \in J} (\gamma_j + a_{{\bf r}_j(t)}) }{\Lambda_{J}(X(t))}1_{A_{J}^c}(X(t)) \leq  2m(1 + |\gamma|_1 + |a|_1),
			\end{align*}
			where $|\gamma|_1 = \sum_{i=1}^d |\gamma_i|$ and $|a|_1$ is defined analogously.
			To obtain the inequality we eliminated the first term by using the fact that $\bar a_{d-l+1} + \sum_{j \in J} \gamma_j \geq \bar a_{d-l+1} + \sum_{i \in I} \gamma_i + \sum_{k=d-l+2}^{d-N} \gamma^{-I}_{(k)} \geq 1$, where the last inequality follows by our assumption. We estimated the second term using \eqref{eqn:Lambda_m_bound}. Applying this estimate to \eqref{eqn:Sigma_stopping_time_base} yields
			
			\begin{align*}
				- \log\Lambda_{J}(X(T \land \tau_n^{J} \land \tau_m^{J,+}))  \leq & - \log\Lambda_{J}(0) + M(T \land \tau_n^{J} \land \tau_m^{J,+}) \\
				& + (m+1)\sigma^2T(1 + |a|_1 + |\gamma|_1).
			\end{align*}
			Taking expectation, sending $n \to \infty$ and using Fatou's lemma yields
			\[- \E_\mu[\log\Lambda_{J}(X(T \land \tau^{J} \land \tau_m^{J,+})] \leq  - \E_\mu[\log\Lambda_{J}(0)] + (m+1)\sigma^2T(1 + |a|_1 +|\gamma|_1).\]
			Since the right hand side is finite it follows that $\P_\mu(\tau^{J} < T  \land \tau_m^{J,+})$ = 0. But since $T$ and $m$ were arbitrary by sending them to infinity we obtain
			\begin{equation} \label{eqn:stopping_time_estimate}
				\P_\mu(\tau^{J} < \tau^{J,+}) = 0.
			\end{equation}
			where $ \tau^{J,+} := \lim_{m \to \infty} \tau_m^{J,+}$.
			By the inductive hypothesis $\tau^{J,+} = \infty$, $\P_\mu$-a.s.\ which completes the proof of the inductive step. Thus the result follows from finite backward induction, finishing the proof of \ref{item:bound_attain_name}\ref{item:bound_attain_name_hit}.
			
			To prove the backward direction of \ref{item:bound_attain_rank}\ref{item:bound_attain_rank_hit} fix $k \in \{2,\dots,d\}$ and assume that $\bar a_{(l)} + \bar \gamma_{(l)} \geq 1 $ for every $l = 2,\dots,k$. We have the  equality
			\begin{equation} \label{eqn:rank_set_inclusion}
				\{X_{(k)}(t) = 0 \text{ for some } t > 0\} = \bigcup_{I}\{\Lambda_I(X(t)) = 0 \text{ for some } t > 0\},
			\end{equation}
			where the union is taken over all sets $I \subseteq \{1,\dots,d\}$ with $|I| = d-k+1$. But note that  for any such $I$ and any $l = 2,\dots,k$ we have that $\bar a_l + \sum_{i \in I} \gamma_i + \sum_{j=l}^{d-N} \gamma^{-I}_{(j)}\geq \bar a_l + \bar \gamma_{(l)} \geq 1$. Hence, by part \ref{item:bound_attain_name}\ref{item:bound_attain_name_hit}, $\P_\mu(\Lambda_I(X(t)) = 0 \text{ for some } t > 0) = 0$ which, together with  \eqref{eqn:rank_set_inclusion}, proves the backward direction of \ref{item:bound_attain_rank}\ref{item:bound_attain_rank_hit}.
			
			Now we turn our attention to proving \ref{item:bound_attain_name}\ref{item:bound_attain_name_above}. Fix $I$ as in the statement of \ref{item:bound_attain_name} and assume that $\bar a_{d-N+1} + \sum_{i \in I} \gamma_i \geq 1$.
			Then arguing as in the inductive step of \ref{item:bound_attain_name}\ref{item:bound_attain_name_hit} we obtain \eqref{eqn:stopping_time_estimate} for the set $J$ replaced by $I$; that is,
			\[\P_\mu(\tau^I < \tau^{I,+}) = 0.\] But we always have the inequality $\tau^I \leq \tau^{I,+}$ since $\Lambda_J(X(t)) = 0 \implies \Lambda_I(X(t)) = 0$ for any index set $J$ containing $I$. We conclude that we must have $\tau^I = \tau^{I,+}$, $\P_\mu$-a.s.\ Next we will show that for every $n$
			\begin{equation} \label{eqn:hitting_above_n}
				\P_\mu(\Lambda_I(X(t)) = 0 \text{ and } \min_J\Lambda_J(X(t)) > 1/n \text{ for some } t > 0) = 0,
			\end{equation} 	where the minimum is taken over all $I \subsetneq J \subset \{1,\dots,d\}$. To this end we define the stopping times $\sigma_n^1 = 0$ and 
			\begin{align*}
				\tau_n^m & = \inf\left\{t \geq \sigma_n^m: \min_J \Lambda_J(X(t)) \leq \frac{1}{n+1}\right\},& m =1,2,\dots, \\
				\sigma_n^m & = \inf\left\{t \geq \tau_n^{m-1}:  \min_J \Lambda_J(X(t)) \geq \frac{1}{n}\right\},& m=2,3,\dots.
			\end{align*}
			
			Note that we cannot have $\min_J\Lambda_J(X(t)) > 1/n$ on the time interval $[\tau_n^m,\sigma_n^{m+1}]$ for any $m$. Hence we obtain the estimate
			\begin{align*} \P_\mu(\Lambda_I(X(t)) &= 0 \text{ and } \min_J\Lambda_J(X(t)) > 1/n \text{ for some } t > 0) \\ 
				&\quad \leq \sum_{m=1}^\infty \P_\mu(\Lambda_I(X(t)) = 0 \text{ and } \min_J\Lambda_J(X(t)) > 1/n \text{ for some } t \in (\sigma_n^m,\tau_n^{m})).
			\end{align*}
			By the strong Markov property we have for every $m$ that 
			\begin{align} \label{eqn:hitting_above_n_strong_markov}
				&\P_\mu( \Lambda_I(X(t)) =  0 \text{ and } \min_J\Lambda_J(X(t)) > 1/n \text{ for some } t \in (\sigma_n^m,\tau_n^{m})) \nonumber \\
				& = \E_\mu[\P_{X_{\sigma_n^m}}(\Lambda_I(X(t)) = 0 \text{ and } \min_J\Lambda_J(X(t)) > 1/n \text{ for some } t \in (0,\tau_n^{1}))1_{\{\sigma_n^m < \infty\}}].
			\end{align}
			For $\P_\mu$-a.e. $\omega \in \{\sigma_n^m < \infty\}$ we have that $\P_{X_{\sigma_n^m}(\omega)} \in \mathcal{P}_0$ so we see that for every such $\omega$ \[
			\inf\{t \geq 0: \Lambda_I(X(t)) = 0 \text{ and } \min_J \Lambda_J(X(t)) > 1/n\} \geq \tau^I = \tau^{I,+}, \quad  \P_{X_{\sigma_n^m}(\omega)}\text{-a.s.,}\]
			where we used the previously established fact that $\tau^{I} = \tau^{I,+}$ almost surely under any law in $\mathcal{P}_0$. But by definition $ \tau^{I,+} \geq \tau^1_n$. Consequently, we deduce that the expression in \eqref{eqn:hitting_above_n_strong_markov} is zero establishing \eqref{eqn:hitting_above_n}. Sending $n \to \infty$ now proves \ref{item:bound_attain_name}\ref{item:bound_attain_name_above}.
			
			It just remains to show \ref{item:bound_attain_rank}\ref{item:bound_attain_rank_above}. Fix $k \in \{2,\dots,d\}$ and assume that $\bar a_k + \bar \gamma_{(k)} \geq 1$. Note that 
			\begin{equation}\label{eqn:pushdown_set_inclusion} 
				\begin{split} \{&X_{(k)}(t)  = 0 \text{ and } X_{(k-1)}(t) > 0  \text{ for some } t > 0\} \\
					&  = \bigcup_{\substack{I \subseteq \{1,\dots,d\} \\ |I| = d-k+1}} \left\{\Lambda_I(X(t)) = 0 \text{ and } \Lambda_J(X(t)) > 0 \text { for every } I \subsetneq J \text{ and some } t > 0\right\}.
				\end{split}
			\end{equation}
			But for any index set $I$ admissible in the union we have that $\bar a_k + \sum_{i \in I} \gamma_i \geq \bar a_k + \bar \gamma_{(k)} \geq 1$. Hence the result follows by \ref{item:bound_attain_name}\ref{item:bound_attain_name_above} and \eqref{eqn:pushdown_set_inclusion} completing the proof.
		\end{proof}
		\section{Proof of Results from  Section~\ref{sec:robust}} \label{app:robust}
		The purpose of this section is to prove Theorem~\ref{thm:robust}. We first prove a technical lemma.
		\begin{proof}[Proof of Lemma~\ref{lem:technical_robust}]
			The proof of \ref{item:collision} is very similar to the proof of Proposition~\ref{prop:Lebesgue_collision}. Fix $\P \in \Pi_{\geq}$ and $k \in \{1,\dots,d-1\}$. Note that \[
			d[X_{(k)}-X_{(k+1)},X_{(k)}-X_{(k+1)}](t) 
			= (X_{(k)}(t) + X_{(k+1)}(t) - (X_{(k)}(t) - X_{(k+1)}(t))^2)dt.
			\] 
			Akin to \eqref{eqn:occ_dens}, we take the function $f(a) = 1_{\{a=0\}}$ in the occupation density formula for $X_{(k)} - X_{(k+1)}$ to obtain
			\begin{equation} \label{eqn:occ_dens_rank}
				\int_0^T 1_{\{X_{(k)}(t) = X_{(k+1)}(t)\}}(X_{(k)}(t) + X_{(k+1)}(t) - (X_{(k)}(t) - X_{(k+1)}(t))^2)dt = 0.
			\end{equation}
			Note that $X_{(k)}(t) + X_{(k+1)}(t) - (X_{(k)}(t) - X_{(k+1)}(t))^2$ is nonnegative and is equal to zero if and only if $X_{(k)}(t) = X_{(k+1)}(t) = 0$. But for $k=1,\dots,N$ by Definition~\ref{def:Pi_geq}\ref{item:non_explosion} this does not happen $\P$-a.s. Hence from \eqref{eqn:occ_dens_rank} we see that  $\P$-a.s.\ we must have that $1_{\{X_{(k)}(t) = X_{(k+1)}(t)\}} = 0$ for almost every $t \geq 0$ and every $k=1,\dots,N$ completing the proof of \ref{item:collision}. Note that for $k = N+1,\dots,d-1$ it follows from \eqref{eqn:occ_dens_rank} that the set $\{t: X_{(k)} = X_{(k+1)}, \ X_{(k)}(t) > 0\}$ is a Lebesgue null-set.
			
			Now to prove \ref{item:QV} we use \cite[Theorem~2.3]{Banner2008Local} to get that 
			\[dX_{(k)}(t) = \frac{1}{N_k(t)}\sum_{i=1}^d 1_{\{{\bf{n}}_k(t) = i\}}dX_i(t) + \frac{1}{2 N_k(t)}\sum_{l=k+1}^d dL_{k,l}(t) - \frac{1}{2N_k(t)}\sum_{l=1}^{k-1}dL_{l,k}(t)\]
			for $k=1,\dots,d$, where $N_k(t) = |\{i \in \{1,\dots,d\}: X_{(k)}(t) = i\}|$ and $L_{k,l} = L_{X_{(k)} - X_{(l)}}$. Hence it follows that
			\begin{equation}
				\label{eqn:QV_identity}
				\begin{split}
					\sum_{i,j=1}^d 1_{\{{\bf{n}}_k(t) = i,{\bf{n}}_l(t)= j\}} d[X_i,X_j](t) & =\sigma^2 N_k(t)N_l(t)X_{(k)}(t)(\delta_{kl} - X_{(l)}(t))dt
				\end{split}
			\end{equation} for $k,l=1,\dots,d,$ where we used Definition~\ref{def:Pi_geq}\ref{item:QV_rank_condition} and the form of $\kappa$ given by \eqref{eqn:kappa_def}. From the proof of Lemma~\ref{lem:technical_robust}\ref{item:collision} we have that $N_k(t) = N_l(t) =  1$ for a.e.\ $t$ on the set $\{X_{(k)}(t) > 0\} \cap \{X_{(l)}(t) > 0 \}$. But when either $X_{(k)}(t) = 0$ or $X_{(l)}(t) = 0$ the right hand side of \eqref{eqn:QV_identity} vanishes. Hence we have the identity $\sum_{i,j=1}^d 1_{\{{\bf{n}}_k(t) = i,{\bf{n}}_l(t)= j\}}d[X_i,X_j](t) = \kappa_{kl}(X_{()}(t))dt$ for $k,l=1,\dots,d$. Multiplying the integrand on the right hand side by $1_{\{{\bf{n}}_k(t) = i,{\bf{n}}_l(t)= j\}}$ we obtain 
			\[1_{\{{\bf{n}}_k(t) = i,{\bf{n}}_l(t)= j\}}d[X_i,X_j](t) = 1_{\{{\bf{n}}_k(t) = i,{\bf{n}}_l(t)= j\}}\kappa_{kl}(X_{()}(t))dt, \quad i,j,k,l=1,\dots,d.\]
			Summing over $k$ and $l$ yields the required identity.
		\end{proof}

		Now we are ready to prove Theorem~\ref{thm:robust}
		
		\begin{proof}[Proof of Theorem~\ref{thm:robust}] In view of the previous discussion carried out following the statement of Theorem~\ref{thm:robust} we just have to show that $g(V^{\hat\theta};\P)$ is independent of $\P \in \Pi_{\geq}$.
			To this end we will approximate the function $\hat G$, given by \eqref{eqn:hat_G}, by $C^2$ functions. Let $\psi \in C_c^\infty(\R)$ be such that $\psi \geq 0$, $\mathrm{supp}(\psi) = [0,1]$ and $\int_{\R}\psi = 1$. For $n \in \N$ define the $d$-dimensional mollifiers $\Psi_n:\R^d \to (0,\infty)$ via $\Psi_n(x) = n^d\prod_{i=1}^d \psi(nx_i)$. Now for $\epsilon > 0$ define the function $\hat G_\epsilon :(-\epsilon,\infty)^d \to (0,\infty)$ via $\hat G_\epsilon(x) = \hat G(x+\epsilon\boldsymbol{1}_d)$ and extend it to all of $\R^d$ by zero. For $n \in \N$ we define the smooth approximations $\hat G_{\epsilon,n}= \hat G_\epsilon* \Psi_n$.  By construction, $\hat G_{\epsilon,n}$ are $C^2$ and since $\hat G_\epsilon$ is continuous on $\Delta^{d-1}$ and differentiable almost everywhere on $\Delta^{d-1}$ we have by properties of convolution that $\lim_{n \to \infty} \hat G_{\epsilon,n} = \hat G_{\epsilon}$ pointwise and $\lim_{n \to \infty}\nabla \log \hat G_{\epsilon,n} = \nabla \log \hat G_{\epsilon}$ almost everywhere on $\Delta^{d-1}$. Additionally, $\hat G_\epsilon$ is $2\epsilon^{-1}$-Lipschitz continuous on $\Delta^{d-1}$ and $\hat G_{\epsilon,n}$ inherits this property. In particular $\nabla \log \hat G_{\epsilon,n}$ is uniformly bounded in $n$ for every $\epsilon > 0$. Moreover, the functions $\hat G_\epsilon$, $\hat G_{\epsilon,n}$ are permutation invariant for every $\epsilon$ and $n$. 
			
			The introduction of $\hat G_{\epsilon,n}$ above was done in two steps. First we introduced $\hat G_\epsilon$ to approximate $\hat G$ by explicit functions which are bounded away from zero (so that $\log \hat G_\epsilon$ is bounded for every $\epsilon$). The second step is to approximate $\hat G_\epsilon$ by the smooth functions $\hat G_{\epsilon,n}$. We will be able to show, by first sending $n \to \infty$ and then $\epsilon \to 0$, that the trading strategy generated by $\hat G$ has the same growth rate in every admissible measure.  
			
			To simplify the expressions to come we make the following definitions. For $y \in \R^d$ satisfying $y_1 \geq y_2 \geq \dots \geq y_d$ we define $\hat F_\epsilon(y) = \hat G_\epsilon(y)$ and $\hat F_{\epsilon,n}(y) = \hat G_{\epsilon,n}(y)$. Then partial derivatives of $\hat F_\epsilon$ and $\hat F_{\epsilon,n}$ will be with respect to the ordered vector coordinate $y_k$, whereas partial derivatives of $\hat G_\epsilon$ and $\hat G_{\epsilon,n}$ will be with respect to the unordered coordinates $x_i$. Next we denote by $\varphi$, $\varphi_\epsilon$ and $\varphi_{\epsilon,n}$ the functions $\log \hat G$, $\log \hat G_{\epsilon}$ and $\log \hat G_{\epsilon,n}$ respectively. Similarly, we write $\xi$, $\xi_\epsilon$ and $\xi_{\epsilon,n}$ for $\log \hat F$, $\log \hat F_{\epsilon}$ and $\log \hat F_{\epsilon,n}$. Lastly, write $\hat \theta, \hat \theta^\epsilon$ and $\hat \theta^{\epsilon,n}$ for the strategies generated by $\hat G, \hat G_\epsilon$ and $\hat G_{\epsilon,n}$ respectively.
			
			Now note that since $\varphi_{\epsilon,n}$ is continuously differentiable on $\R^d$ it follows in a similar way to \eqref{eqn:smooth_permute} that 
			\begin{equation} \label{eqn:robust_rank_identity1}
				\sum_{i,j=1}^d \partial_i\varphi_{\epsilon,n}\partial_j\varphi_{\epsilon,n}(X(t)) d[X_i,X_j](t) = \nabla \xi_{\epsilon,n}^\top \kappa\xi_{\epsilon,n}(X_{()}(t))dt.
			\end{equation}
			Next we have for each $\epsilon > 0$ that $\varphi_\epsilon$ is differentiable off of the set $\bigcup_{k=1,\dots,N}\{x_{(k)} = x_{(k+1)}\}$ and for every $x$ at which $\varphi_\epsilon$ is differentiable we have that 
			\[\partial_i \varphi_\epsilon (x) = \partial_k \xi_\epsilon(x_{()}), \quad {\bf{n}}_k(x) = i.\] Since by Lemma~\ref{lem:technical_robust}\ref{item:collision} we have that the top $N+1$ ranks only collide at a null-set of time points, we obtain using Lemma~\ref{lem:technical_robust}\ref{item:QV} the relationship
			\begin{equation} \label{eqn:robust_rank_identity2}
				\int_0^T\sum_{i,j=1}^d  \partial_i \varphi_\epsilon\partial_j\varphi_\epsilon(X(t))d[X_i,X_j](t) = \int_0^T\nabla \xi_\epsilon^\top \kappa\nabla\xi_\epsilon(X_{()}(t))dt,
			\end{equation}
			$\P\text{-a.s.}$ for every $\P \in \Pi_{\geq}$.
			An analogous identity for $\varphi$ and $\xi$, replacing $\varphi_\epsilon$ and $\xi_\epsilon$ respectively, holds as well.
			
			With these preliminaries in hand we are now ready to analyze the wealth processes. Since $\hat G_{\epsilon,n}$ is $C^2$ with bounded first and second derivatives (whenever $\epsilon > 1/n)$ and $ \xi_{\epsilon,n}$ is bounded it follows by Lemma~\ref{lem:func_gen_growth_ivnariance} that \begin{equation} \label{eqn:epsilon_n_growth}
				\lim_{T \to \infty} \frac{1}{T}\log V^{\hat \theta^{\epsilon,n}}(T) = \int_{\nabla^{d-1}} (\nabla \xi_{\epsilon,n}^\top \kappa \varrho - \frac{1}{2}\nabla \xi_{\epsilon,n}^\top \kappa \nabla \xi_{\epsilon,n})q =: \eta_{\epsilon,n}, \quad \P\text{-a.s.}
			\end{equation} for every $\epsilon,n$ such that $\epsilon > 1/n$. Next we will show that 
			\begin{equation} \label{eqn:epsilon_growth}
				\lim_{T \to \infty} \frac{1}{T}\log V^{\hat \theta^\epsilon}(T) =  \int_{\nabla^{d-1}} (\nabla \xi_{\epsilon}^\top \kappa \varrho - \frac{1}{2}\nabla \xi_{\epsilon}^\top \kappa \nabla \xi_{\epsilon})q =: \eta_\epsilon, \quad \P\text{-a.s.}
			\end{equation}
			for every $\P \in \Pi_{\geq}$ and $\epsilon > 0$. To this end fix $\P \in \Pi_{\geq}$ and let $b^{\P}$ and $M^{\P}$ denote the drift process and local martingale part of $X$ under $\P$, as in Definition~\ref{def:Pi_geq}, respectively. Next we compute that
			\begin{equation}\label{eqn:logV_n_estimate}
				\begin{split} & \left|\frac{\log V^{\hat\theta^{\epsilon}}(T)}{T} -\frac{\log V^{\hat\theta^{\epsilon,n}}(T)}{T}\right| \leq
					\frac{1}{T}\int_0^T  |(\nabla \varphi_\epsilon(X(t)) - \nabla \varphi_{\epsilon,n}(X(t)))^\top b^{\P}(t)|dt \\
					& \hspace{0.2cm} + \frac{1}{2T}\int_0^T\left|\sum_{i,j=1}^d \Big(\partial_i \varphi_\epsilon\partial_j\varphi_\epsilon(X(t))- \partial_i\varphi_{\epsilon,n}\partial_j\varphi_{\epsilon,n}(X(t))\Big) d[X_i,X_j](t)\right|  \\
					&  \hspace{0.2cm} +  \left|\frac{1}{T}\int_0^T   (\nabla \varphi_\epsilon(X(t)) - \nabla \varphi_{\epsilon,n}(X(t)))^\top dM^{\P}(t)\right|   \\
					& \leq \sum_{i=1}^d \(\frac{1}{T}\int_0^T|\partial_i \varphi_\epsilon(X(t)) - \partial_i \varphi_{\epsilon,n}(X(t))|^{r}dt\)^{1/r}\(\frac{1}{T}\int_0^T|b^{\P}_i(t)|^{r'}dt\)^{1/r'} \\
					&  \hspace{0.5cm} + \frac{1}{2T}\int_0^T|\nabla \xi_{\epsilon}^\top \kappa \nabla \xi_{\epsilon}(X_{()}(t)) -  \nabla \xi_{\epsilon,n}^\top \kappa \nabla \xi_{\epsilon,n}(X_{()}(t))|dt + N^{\P}_{\epsilon,n}(T),
				\end{split}
			\end{equation} where $N^{\P}_{\epsilon,n}(T) =  \int_0^T   (\nabla \varphi_\epsilon(X(t)) - \nabla \varphi_{\epsilon,n}(X(t)))^\top dM^{\P}(t)$, $r'$ is as in Definition~\ref{def:Pi_geq}\ref{item:drift_condition_rank} and we have $1/r + 1/r' = 1$.
			We used the fact that for a.e.\ $t$,  $\hat \theta^\epsilon(X(t)) = \nabla \varphi_\epsilon(X(t)) + (1-\nabla \varphi_\epsilon(X(t))^\top X(t))\boldsymbol{1}_d$ and the analogous expression for $\hat \theta^{\epsilon,n}$. We also used \eqref{eqn:robust_rank_identity1} and \eqref{eqn:robust_rank_identity2} in the final inequality.  
			
			By a similar calculation to \eqref{eqn:robust_rank_identity2} we have that $ [N^{\P}_{\epsilon,n},N^{\P}_{\epsilon,n}](T) = \int_0^T (\nabla \xi_\epsilon - \nabla \xi_{\epsilon,n})^\top \kappa (\nabla \xi_\epsilon - \nabla \xi_{\epsilon,n}) (X_{()}(t))dt$.
			Hence, it follows by the ergodic property that
			\[\lim_{T\to \infty} \frac{[N^{\P}_{\epsilon,n},N^{\P}_{\epsilon,n}](T)}{T}= \int_{\nabla^{d-1}} (\nabla \xi_\epsilon - \nabla \xi_{\epsilon,n})^\top \kappa (\nabla \xi_\epsilon(y) - \nabla \xi_{\epsilon,n})q, \quad \P\text{-a.s.}\]
			Since $\nabla \xi_\epsilon$ and $\nabla \xi_{\epsilon,n}$ are bounded, the right hand side is finite so by \cite[Lemma~1.3.2]{fernholz2002stochastic} we have that $\lim_{T \to \infty} T^{-1}N^{\P}_{\epsilon,n}(T) = 0$, $\P$-a.s. Next note that 
			\[\sum_{i=1}^d \(\frac{1}{T}\int_0^T\big|(\partial_i \varphi_\epsilon - \partial_i \varphi_{\epsilon,n})(X(t))\big|^{r}dt\)^{1/r} \hspace{-0.3cm} = \sum_{k=1}^d  \(\frac{1}{T}\int_0^T\big|(\partial_k \xi_\epsilon - \partial_k \xi_{\epsilon,n})(X_{()}(t))\big|^{r}dt\)^{1/r}.\]
			Also define
			$C^{r'} = \max_{i=1,\dots,d}\limsup_{T\to \infty}\frac{1}{T}\int_0^T|b^{\P}_i(t)|^{r'}dt,$ which is finite due to Definition~\ref{def:Pi_geq}\ref{item:drift_condition_rank}. It then follows from the ergodic property and \eqref{eqn:logV_n_estimate} that $\P$-a.s.,
			\begin{equation} \label{eqn:limsup_bound_epsilon}
				\begin{split} \limsup_{T \to \infty} \left|\frac{\log V^{\hat\theta^{\epsilon}}(T)}{T} -\frac{\log V^{\hat\theta^{\epsilon,n}}(T)}{T}\right| &  \leq C\sum_{k=1}^d \(\int_{\nabla^{d-1}}|\partial_k\xi_\epsilon - \partial_k \xi_{\epsilon,n}|^rq \)^{1/r}\\
					& \  + \frac{1}{2}\int_{\nabla^{d-1}} (\nabla \xi_\epsilon^\top \kappa \nabla\xi_\epsilon - \nabla \xi_{\epsilon,n}^\top \kappa \nabla\xi_{\epsilon,n})q.
				\end{split}
			\end{equation}
			Since $\nabla \xi_{\epsilon,n} \to \nabla \xi_\epsilon$ almost everywhere as $n \to \infty$ and the gradients are uniformly bounded in $n$ it follows by sending $n \to \infty$ that 
			\begin{equation} \label{eqn:uniform_epsilon_wealth}
				\lim_{n \to \infty}\limsup_{T \to \infty} \left|\frac{\log V^{\hat\theta^{\epsilon}}(T)}{T} -\frac{\log V^{\hat\theta^{\epsilon,n}}(T)}{T}\right| = 0, \quad \P\text{-a.s.}
			\end{equation}
			Now we estimate that 
			\begin{equation} \label{eqn:eta_epsilon_growth}
				\left|\frac{\log V^{\hat\theta^{\epsilon}}(T)}{T} - \eta_\epsilon\right|  \leq \left|\frac{\log V^{\hat\theta^{\epsilon}}(T)}{T} -\frac{\log V^{\hat\theta^{\epsilon,n}}(T)}{T}\right| + \left|\frac{\log V^{\hat\theta^{\epsilon,n}}(T)}{T} - \eta_{\epsilon,n}\right| + |\eta_{\epsilon,n} - \eta_\epsilon|,
			\end{equation}
			where we recall that $\eta_{\epsilon,n}$ and $\eta_\epsilon$ were defined in \eqref{eqn:epsilon_n_growth} and \eqref{eqn:epsilon_growth} respectively.
			Using \eqref{eqn:epsilon_n_growth}, \eqref{eqn:uniform_epsilon_wealth} and the fact that $\lim_{n \to \infty} \eta_{\epsilon,n} = \eta_\epsilon$ we obtain \eqref{eqn:epsilon_growth} by first sending $T \to \infty$ and then $n \to \infty$ in \eqref{eqn:eta_epsilon_growth}.
			
			We employ a similar technique to analyze $\log V^{\hat \theta}$. We use the definition of wealth process \eqref{eq_wealth_dynam_NEW} to obtain the semimartingale decomposition for $\log V^{\hat \theta}$ and $\log V^{\hat \theta^\epsilon}$ under any admissible measure $\P \in \Pi_{\geq}$. Proceeding as in \eqref{eqn:logV_n_estimate} we obtain a bound akin to \eqref{eqn:limsup_bound_epsilon},
			\begin{equation} \label{eqn:log_V_bound}
				\begin{split}
					\limsup_{T \to \infty}\left|\frac{\log V^{\hat\theta}(T)}{T} -\frac{\log V^{\hat\theta^{\epsilon}}(T)}{T}\right| \leq  & \ C\sum_{k=1}^d \(\int_{\nabla^{d-1}}|\partial_k\xi - \partial_k \xi_{\epsilon}|^rq \)^{1/r}\\
					& \  + \frac{1}{2}\int_{\nabla^{d-1}} (\nabla \xi^\top \kappa \nabla\xi - \nabla \xi_{\epsilon}^\top \kappa \nabla\xi_{\epsilon})q.
				\end{split}
			\end{equation}
			Direct calculations show that
			\[\sum_{k=1}^d |\partial_k \xi(y) - \partial_k \xi_\epsilon(y)|^r + |\nabla \xi^\top \kappa \nabla \xi(y) - \nabla \xi_\epsilon^\top \kappa \nabla \xi_\epsilon(y)|\leq \frac{\tilde C}{y_{N+1}^r}\]
			for some constant $\tilde C > 0$ and every $y \in \nabla^{d-1}_+$.
			By assumption \eqref{eqn:finite_growth_ass} on the parameter $a$ together with the definition of $r$ from Definition~\ref{def:Pi_geq}\ref{item:drift_condition_rank} we have by Lemma~\ref{lem:Q_finite} that 
			$\int_{\nabla^{d-1}} {y^{-r}_{N+1}}{q(y)}dy < \infty.$
			Hence, by sending $\epsilon \downarrow 0$ in \eqref{eqn:log_V_bound} we obtain by the dominated convergence theorem that
			\[
			\lim_{\epsilon \downarrow 0} \limsup_{T \to \infty}\left|\frac{\log V^{\hat\theta}(T)}{T} -\frac{\log V^{\hat\theta^{\epsilon}}(T)}{T}\right| = 0, \quad \P\text{-a.s.}
			\]
			Next note that, again by dominated convergence, we have
			\[\eta := \lim_{\epsilon \downarrow 0} \eta_\epsilon = \frac{1}{8}\int_{\nabla^{d-1}} \(\sum_{k=1}^N\frac{a_k^2}{y_k} + \frac{\bar a_{N+1}^2}{\bar y_{N+1}}\)q(y)dy - \frac{1}{8}\bar a_1^2. \]
			Then, by an analogous estimate to \eqref{eqn:eta_epsilon_growth} 
			we obtain that $\lim_{T \to \infty} T^{-1} \log V^{\hat \theta}(T) = \eta$, $\P$-a.s. establishing that $g(V^{\hat \theta};\P) = \eta$ for every $\P \in \Pi_{\geq}.$ This completes the proof.
		\end{proof}
	\end{appendix}

	\bibliographystyle{plain}
	\bibliography{references}

\end{document}